\documentclass[12pt]{article}

\usepackage[utf8]{inputenc}
\usepackage[T1]{fontenc}
\usepackage[margin=0.9in]{geometry}
\usepackage{amsmath,amssymb,amsthm,mathtools,bm}
\usepackage{booktabs,comment}
\usepackage{enumitem}
\usepackage{graphicx,titling}
\usepackage{subcaption}
\usepackage{pgfplots}
\pgfplotsset{compat=1.18}
\usepackage{setspace,xr-hyper}
\usepackage[authoryear]{natbib}
 \usepackage[colorlinks,linkcolor={blue},citecolor={blue},urlcolor={black},filecolor={black}]{hyperref}

\allowdisplaybreaks
\newtheorem{assumption}{Assumption}
\newtheorem{theorem}{Theorem}
\newtheorem{proposition}{Proposition}
\newtheorem{lemma}{Lemma}

\newcommand{\ind}{\mathbf{1}}

\DeclareMathOperator{\Var}{Var}
\DeclareMathOperator{\Cov}{Cov}
\DeclareMathOperator{\tr}{tr}

\title{Robust Inference for Dyadic Data with Spatially Dependent Nodes}
\author{%
{\large \textbf{Ulrich Hounyo}}\thanks{Department of Economics, University at Albany - State University of New York, Albany, NY 12222, United States. E-mail: khounyo@albany.edu.}
\and
{\large \textbf{Jiahao Lin}}\thanks{School of Digital Economy and Management, Fuyao University of Science and Technology, Fuzhou, 350109, China.
E-mail: jhlin@fyust.edu.cn.}
\and
{\large \textbf{Xiaojun Song}}\thanks{Department of Business Statistics and Econometrics, Peking University, Beijing, China. E-mail: sxj@gsm.pku.edu.cn.}
}
\date{\today}

\begin{document}
\setlength{\droptitle}{-4em}
\maketitle

\vspace{-3em}
\begin{abstract}
We develop inference for complete dyadic samples with spatial dependence governed by geographic distances between nodes, covering settings beyond the scope of existing dyadic inference theory. We propose spatial variance and corrected block jackknife estimators consistent in nondegenerate and degenerate Gaussian cases. Under degeneracy, spatial subsampling consistently estimates possibly non-Gaussian limits. Combining the jackknife with subsampling yields the max jackknife–subsampling (MJS) interval, which provides pointwise asymptotically exact coverage in both Gaussian cases and conservative coverage in the specified non-Gaussian case. Fixed-effect extensions show that estimating node effects can change the fast limiting distribution.  Simulations and an empirical illustration are provided.
\end{abstract}

\newpage
\doublespacing
\section{Introduction}

Dyadic data record economic relationships between pairs of units, including trade between countries, commuting between locations, and transactions between firms. In regional applications, these relationships connect local markets and transmit economic changes across places. Commuting, for example, connects local labor markets and is central to quantitative spatial models \citep{monte2018commuting,redding2017quantitative}. Flows involving the same location may respond to common local conditions, but nearby locations may also share industries, infrastructure, and labor markets. Dependence can therefore extend to flows with no location in common, as motivated by evidence on regional income dependence and spatial models of origin-destination flows \citep{rey1999income,lesage2008flows}.

Conventional dyadic clustering accounts for covariance between observations sharing a node \citep{cameron2014robust,aronow2015cluster,tabord2019inference}, but omits covariance between dyads whose distinct nodes are spatially dependent. Correlated conditions in locations $i$ and $p$, for example, can induce correlation between the disturbances for $(i,j)$ and $(p,q)$ even when all four locations differ. These covariances need not be negligible: as the number of nodes increases, so does the number of dyad pairs involving nearby nodes. Our descriptive exercise using U.S. county income data illustrates the concern: nearby counties have more strongly correlated income growth after allowing each county to respond differently to national growth.

This paper develops inference for complete dyadic samples with spatially dependent nodes. We observe every unordered pair among $n$ nodes and model node characteristics as a weakly dependent spatial field on an expanding rectangular lattice. Dependence follows geographic distances and is invariant to node numbering. Our objective goes beyond estimating the extra covariance due to spatial dependence. The first-order node contribution to the regression score may vanish, changing both the convergence rate and the limiting distribution. We develop variance estimators for the Gaussian cases and spatial subsampling for the possibly non-Gaussian case. Their combination provides a confidence interval without first selecting between these cases.

The first contribution is a limit theory distinguishing these sources of uncertainty. Despite the $n(n-1)/2$ observations, node variation generally produces a $\sqrt n$ Gaussian limit. When the first-order node contribution vanishes, which we call degeneracy, the estimator can converge at rate $n$. This limit is Gaussian when only conditionally independent dyad shocks remain, but can be non-Gaussian when node interactions remain. Under a summable interaction expansion, we derive the latter distribution and establish the same limit for growing spatial subsamples scaled by their own number of nodes.

The second contribution is a pair of variance estimators accounting for common and nearby nodes. The dependent-node dyadic estimator, or DN-Dyadic CRVE, includes score cross-products for dyads with nearby endpoints. The row-column spatial block jackknife, or JK-DN-Dyadic CRVE, deletes each node neighborhood and all dyads involving its nodes. Neighborhoods share a distance cutoff but may contain different numbers of nodes, particularly near the observation boundary. A correction accounts for how often each dyad enters the deleted-score sums. It is negligible when the first-order node contribution dominates, but necessary when dyad shocks determine the $n$-rate limit: without it, the block jackknife asymptotically counts their variance twice. We establish consistency of the corrected estimators in both Gaussian cases.

The third contribution concerns non-Gaussian inference. Following \citet{politis1994subsamples}, we reestimate the regression using all dyads within spatial cubes, preserving dependence among the retained nodes. Each cube estimate is centered at the full-sample estimate and scaled by the cube's node count. Under degeneracy, the empirical distribution of these differences estimates the distribution of the $n$-scaled full-sample estimation error. Consequently, the subsampling interval is asymptotically exact under the stated limit and quantile conditions, without estimating the functions or coefficients in the interaction expansion. Its $n$-rate calibration, however, gives an interval that is too short in the $\sqrt n$ Gaussian case. We therefore take the larger of the subsampling and Gaussian jackknife intervals, both centered at the full-sample estimate. The resulting \emph{max jackknife-subsampling (MJS) interval} is asymptotically exact in both Gaussian cases and has at least nominal coverage in the degenerate non-Gaussian case. These results are pointwise for fixed data-generating processes, rather than uniform along sequences approaching degeneracy.

The fourth contribution concerns node fixed effects. These effects absorb unobserved location characteristics but need not eliminate the node contribution to the score. Moreover, replacing sample demeaning by population demeaning can omit a term of order $n^{-1}$, negligible at the $\sqrt n$ rate but potentially important at the $n$ rate. We establish Gaussian fixed-effect results and derive the additional non-Gaussian contribution under further representation conditions. Subsampling accommodates this contribution by refitting node effects within each cube, without separately estimating the demeaning term.

Two related approaches already allow dependence between disjoint dyads. \citet{canen2024inference} develop inference for network spillovers using the limit theory of \citet{kojevnikov2021limit}. Their distance counts links through active dyads, with dyads directly linked when they share a node. In a complete sample, every pair of disjoint dyads has graph distance two regardless of geography. Our endpoint distances distinguish nearby from distant disjoint dyads, accommodating spatial node dependence while retaining all location pairs, including zero-flow pairs. The approaches therefore address different empirical settings. Section~\ref{subsec:empirical_motivation} compares the distances mathematically. We additionally address degeneracy and the consequences of estimating node effects.

\citet{jochmans2026two} studies dyads generated by dependent latent node variables indexed by an ordered sequence and develops HAC and bootstrap inference when the first-order long-run variance is positive. He also illustrates a non-Gaussian limit under degeneracy. Our setting allows dependence in two spatial dimensions without imposing a one-dimensional ordering. We further address inference when first-order variance does not determine uncertainty: corrections become essential in the degenerate Gaussian case, while node interactions can require non-Gaussian critical values. The corrected spatial jackknife, subsampling procedure, and fixed-effect analysis jointly address these problems.

Our analysis also connects spatial and clustered inference. Spatial HAC methods account for covariance that decreases with distance \citep{conley1999gmm,kelejian2007hac,kim2011spatial}. Here each observation has two locations and each node enters $n-1$ observations. On a $\nu$-dimensional lattice, a fixed dyad has $O(nL^\nu)$ other dyads with an endpoint within distance $L$ of one of its nodes. The variance and moment calculations must account for these increasing dyad counts as well as spatial decay. Related panel methods allow serially correlated time effects \citep{chiang2023standard}, develop fixed-$b$ inference \citep{chen2023fixed}, and use jackknife procedures \citep{hounyo2025jackknife}. In our setting, the same nodes appear at both endpoints, changing deletion and inference under degeneracy.\footnote{The degenerate-case argument in \citet{hounyo2025jackknife} uses a one-way cluster jackknife component associated with the panel's distinct row and time dimensions. Undirected dyads do not provide that separate cluster dimension: the observations involving nodes $i$ and $j$ both include $(i,j)$, and distinct nodes may also be spatially dependent. We instead establish validity by subsampling complete dyadic samples within spatial cubes.}

The jackknife builds on block methods for dependent observations \citep{kunsch1989jackknife,liu1992moving} and cluster jackknife inference \citep{hansen2022jackknife,mackinnon2023fast,mackinnon2024jackknife}. Multiway clustering and exchangeable-array methods provide a broader foundation for inference with shared indices \citep{cameron2011robust,thompson2011simple,davezies2021empirical}. Degeneracy is addressed by the bootstrap methods of \citet{menzel2021bootstrap} and the modified empirical likelihood of \citet{chiang2023multiway}; \citet{leucht2012degenerate} studies approximation methods for degenerate statistics under weak dependence. We establish the spatial dyadic and cube limits needed for subsampling, including the additional fixed-effect demeaning term, and combine this inference with the corrected Gaussian jackknife.

Monte Carlo experiments examine both Gaussian cases and the degenerate non-Gaussian case, with and without fixed effects. Conventional methods that omit spatial node dependence can substantially overreject. The corrected spatial jackknife JK-DN-Dyadic CRVE generally brings rejection frequencies closer to the nominal level than conventional alternatives and DN-Dyadic CRVE in the reported designs. MJS offers broader pointwise coverage protection but can be appreciably conservative in finite samples. 

An application using public U.S. Census commuting data for Raleigh-Durham compares the sensitivity of transformed commuting counts to distance across earnings groups, controlling for location fixed effects and other characteristics. The distance coefficient is more negative for higher-paid jobs, and the difference remains significant under the reported MJS interval. An additional same-county association, significant under conventional dyadic inference, is no longer significant under MJS. 

Section~\ref{sec:model} presents the model, empirical motivation, and score decomposition. Section~\ref{sec:estimators} introduces the inference procedures, Section~\ref{sec:theory} establishes their asymptotic properties, and Section~\ref{subsec:fixed_effects} studies node fixed effects. Sections~\ref{sec:simulations} and~\ref{sec:empirical_commuting} report simulation and empirical results. Section~\ref{sec:conclusion} concludes. The supplementary appendices contain proofs, implementation details, and additional results.
\section{Dyadic Model with Spatially Dependent Nodes}
\label{sec:model}

\subsection{Dyadic model and dependence}

Let $i,j\in\{1,\ldots,n\}$ index nodes. We observe a complete sample
of undirected dyads, with one observation for each unordered pair:
\[
 \mathcal D_n=\{(i,j):1\le i<j\le n\},
 \qquad
 M_n=|\mathcal D_n|=\frac{n(n-1)}2.
\]
For each $(i,j)\in\mathcal D_n$, consider the linear regression model
\begin{equation}
 y_{ij}=\bm{x}_{ij}^{\top}\bm{\beta}+u_{ij},
 \label{eq:model}
\end{equation}
where $\bm{x}_{ij}\in\mathbb R^K$ includes a constant,
$\bm{\beta}\in\mathbb R^K$ is the parameter of interest, and
$u_{ij}$ is the regression disturbance. The dimension $K$ is fixed.
Stacking observations over $\mathcal D_n$ gives
$\bm y=\bm X\bm\beta+\bm u$, where $\bm X$ is the $M_n\times K$
regressor matrix and $\bm y$ and $\bm u$ are $M_n\times1$ vectors.
The OLS estimator is
\begin{equation}
 \widehat{\bm{\beta}}
 =(\bm X^{\top}\bm X)^{-1}\bm X^{\top}\bm y.
 \label{eq:ols}
\end{equation}
We focus on inference for a fixed scalar contrast
$\bm a^{\top}\bm\beta$, where $\bm a\in\mathbb R^K$ is nonzero
and does not depend on $n$.

Define the population and estimated scores by
\[
 \bm s_{ij}=\bm x_{ij}u_{ij},
 \qquad
 \widehat{\bm s}_{ij}=\bm x_{ij}\widehat u_{ij},
 \qquad
 \widehat u_{ij}
 =y_{ij}-\bm x_{ij}^{\top}\widehat{\bm\beta}.
\]
The OLS estimation error satisfies the exact identity
\begin{equation}
 \widehat{\bm{\beta}}-\bm{\beta}
 =(\bm X^{\top}\bm X)^{-1}
   \sum_{(i,j)\in\mathcal D_n}\bm s_{ij}.
 \label{eq:ols_score_expansion}
\end{equation}
Thus, inference requires accounting for dependence among the
dyadic scores $\{\bm s_{ij}:(i,j)\in\mathcal D_n\}$.

We specify this dependence through node variables and dyad shocks.
Let $Z_i$ collect the random variables associated with node $i$,
and let $Q_{ij}$ denote a shock specific to dyad $(i,j)$.

\begin{assumption}[Dyadic representation]
\label{ass:representation}
For each $(i,j)\in\mathcal D_n$,
\begin{equation}
 (y_{ij},\bm x_{ij},u_{ij})=h(Z_i,Z_j,Q_{ij}),
 \label{eq:kernel_rep}
\end{equation}
where $h$ is a measurable function symmetric in its first two
arguments. The dyad shocks $\{Q_{ij}:i<j\}$ are i.i.d. and
independent of the entire collection of node variables.
\end{assumption}

Under Assumption~\ref{ass:representation}, dyads sharing a node
may be dependent through its node variable. For example,
$(i,j)$ and $(i,k)$ both depend on $Z_i$. If the node variables
are i.i.d., dyads with no common node are independent, yielding
the usual dissociated dyadic model. We allow dependence between
distinct node variables and specify its spatial structure in
the next subsection.

This latent-variable specification is motivated by the
Aldous-Hoover-Kallenberg representations for exchangeable arrays
\citep{aldous1981representations,hoover1979relations,kallenberg1989representation},
which are widely used in the multiway clustering literature; see,
for example, \citet{menzel2021bootstrap},
\citet{chiang2023standard}, and \citet{chen2023fixed}.
Spatial dependence among the node variables is an additional
modeling assumption, rather than a consequence of these
representation theorems.
\subsection{Spatial node dependence}
\label{subsec:spatial_dependence}

We allow dependence between distinct node variables $Z_i$ and $Z_j$ to depend on
the distance between their observed locations. Dyads may therefore
be dependent even when they have no node in common.

Following \citet{conley1999gmm}, we model the node variables as
observations of a spatial process:
\[
 Z_i=Z(\bm{s}_i),\qquad i=1,\ldots,n,
\]
where $\bm{s}_i\in\mathbb R^\nu$ is the fixed, observed location
of node $i$. We consider a regular lattice with unit spacing and
state the results for $\nu\in\{1,2\}$, with $\nu=2$ as the main
spatial setting.\footnote{The
construction can be extended to $\mathbb R^\nu$ for any fixed
integer $\nu>2$, using $\nu$-dimensional lattice neighborhoods and
a compactly supported node kernel that is positive definite in
that dimension; see \citet{wendland1995piecewise}. The relevant
counts become $O(L^\nu)$ per node and $O(nL^\nu)$ per dyad.
The moment, mixing, bandwidth, and subsampling arguments must use
these dimension-specific counts; the sufficient restrictions below
take the form $\mathfrak d>8\nu\lambda/(\lambda-1)$ and
$L^{3\nu}/n\to0$. Thus the extension requires the corresponding
kernel and rate adjustments. The theorems here are stated for
$\nu\in\{1,2\}$.}
The distance between nodes $i$ and $j$ is
\[
 d_{ij}=\|\bm{s}_i-\bm{s}_j\|_2.
\]
Dependence is determined by node locations, rather than by how
the nodes are numbered.

Distance can also describe economic proximity, such as similarity
in technology or product markets; see \citet{conley2007spatial}
and \citet{kelejian2007hac}. Extending the analysis to such distances
requires suitable conditions on node locations and on how dependence
decreases with distance. Here we use Euclidean distances on a
rectangular lattice.

For each coordinate direction $k=1,\ldots,\nu$, let $N_{k,n}$
denote the number of lattice positions along that direction.
The observed locations comprise all sites in
\[
 \Lambda_n=\prod_{k=1}^{\nu}\{1,\ldots,N_{k,n}\},
 \qquad
 n=\prod_{k=1}^{\nu}N_{k,n},
\]
with exactly one observed node at each site. For example,
when $\nu=2$,
\[
 \Lambda_n
 =\{1,\ldots,N_{1,n}\}\times\{1,\ldots,N_{2,n}\}
\]
is an $N_{1,n}\times N_{2,n}$ grid containing
$n=N_{1,n}N_{2,n}$ nodes.
As $n\to\infty$, we require
\[
 \min_{1\le k\le\nu}N_{k,n}\to\infty,
 \qquad
 \frac{\max_{1\le k\le\nu}N_{k,n}}
      {\min_{1\le k\le\nu}N_{k,n}}\le C,
\]
where $C<\infty$ does not depend on $n$. Thus, the observation
region grows in every direction, while the ratio of its longest
to shortest side remains bounded. The distance between adjacent
lattice sites remains fixed.

The lattice structure bounds the number of nodes at a given
distance from any node:
\[
 \sup_{1\le i\le n}
 \#\{j:1\le j\le n,\ r\le d_{ij}<r+1\}
 \le C(r+1)^{\nu-1},
 \qquad r\ge0.
\]
This bound is used below when summing covariances over pairs
of nodes. The case $\nu=1$ corresponds to equally spaced nodes
on a line.

For finite nonempty sets $A,B\subset\mathbb Z^\nu$, define
\[
 \mathcal Z_A=\sigma\{Z(\bm{s}):\bm{s}\in A\},
 \qquad
 \operatorname{dist}(A,B)
 =\min_{\bm{s}\in A,\bm{t}\in B}\|\bm{s}-\bm{t}\|_2.
\]
Thus, $\mathcal Z_A$ contains the information in the node variables
at sites in $A$, and $\operatorname{dist}(A,B)$ is the shortest
distance between a site in $A$ and a site in $B$.

For two sigma fields $\mathcal A$ and $\mathcal B$, define their
absolute-regularity coefficient by
\[
 \beta(\mathcal A,\mathcal B)
 =\frac12\sup
   \sum_r\sum_t
   \left|P(A_r\cap B_t)-P(A_r)P(B_t)\right|,
\]
where the supremum is over finite partitions $\{A_r\}$ and
$\{B_t\}$ of the probability space, with $A_r\in\mathcal A$
and $B_t\in\mathcal B$.

\begin{assumption}[Spatial weak dependence]
\label{ass:mixing}
The spatial field
$\{Z(\bm{s}):\bm{s}\in\mathbb Z^\nu\}$ is strictly stationary.
There exist $\lambda>1$, $C<\infty$, and
$\mathfrak d>8\nu\lambda/(\lambda-1)$ such that, for all finite
nonempty disjoint sets $A,B\subset\mathbb Z^\nu$,
\[
 \beta(\mathcal Z_A,\mathcal Z_B)
 \le C|A||B|
       \{1+\operatorname{dist}(A,B)\}^{-\mathfrak d}.
\]
\end{assumption}

Assumption~\ref{ass:mixing} bounds dependence between two groups
of node variables. For fixed group sizes, the bound decreases
as the distance between the groups increases. The factor $|A||B|$
allows the bound to increase with the number of nodes in each
group. The stated decay rate is sufficient for the moment and
block arguments below; no optimal rate is claimed.

The assumption holds when groups of node variables are independent
whenever the distance between them exceeds a fixed value.
A geometric bound
$C|A||B|\rho^{\operatorname{dist}(A,B)}$, with $0<\rho<1$,
also implies the stated polynomial bound.
\subsection{Why spatial node dependence matters}
\label{subsec:empirical_motivation}

In regressions for flows between regions, node variables can
represent local income prospects, productivity, and labor-market
conditions. These factors may be correlated across nearby regions
because they share industries, infrastructure, and labor markets.
Income prospects are an important determinant of migration
\citep{kennan2011income}, and \citet{rey1999income} document spatial
dependence in regional income patterns. Related work models
spatial dependence directly among origin-destination flows
\citep{lesage2008flows}. These findings motivate allowing
dependence between distinct node variables, in addition to
dependence between dyads sharing a node.

To see the implication for dyadic regressions, consider the
scalar example $u_{ij}=Z_i+Z_j+Q_{ij}$. For distinct dyads,
independence of the dyad shocks from each other and from the
node variables gives
\[
 \begin{aligned}
 \Cov(u_{ij},u_{pq})
 &=\Cov(Z_i,Z_p)+\Cov(Z_i,Z_q)+\Cov(Z_j,Z_p)+\Cov(Z_j,Z_q).
 \end{aligned}
\]
Spatial dependence between $Z_i$ and $Z_p$ can therefore induce
correlation between $u_{ij}$ and $u_{pq}$ even when
$i,j,p,q$ are all distinct. Conventional dyadic clustering does
not account for this correlation because the two dyads share
no node.

We illustrate the empirical relevance of spatial node dependence
by examining how the correlation of county income growth varies
with geographic distance. We use annual per capita personal
income from the U.S. Bureau of Economic Analysis and county
coordinates from the U.S. Census Bureau. The sample contains
3,055 counties and county equivalents in the contiguous United
States, with annual income growth observed from 2000 to 2019.
We first adjust each county's income growth for national income
growth. For each pair of counties, we then calculate the
correlation between their adjusted growth series over time.
Grouping county pairs by the distance between them, we report
the average correlation within each distance range.\footnote{Growth
is measured as 100 times the annual change in log nominal per
capita income. For each county separately, we regress its annual
growth on a constant and annual U.S. per capita income growth.
The coefficients can therefore differ across counties, allowing
different responses to national fluctuations. We use the resulting
residuals as adjusted growth. Each pairwise correlation is
calculated from the two counties' 20 annual residual observations
over 2000-2019. Each county pair receives equal weight when
averaging correlations within a distance range.}

\begin{figure}[t!]
 \centering
 \includegraphics[width=0.94\textwidth]{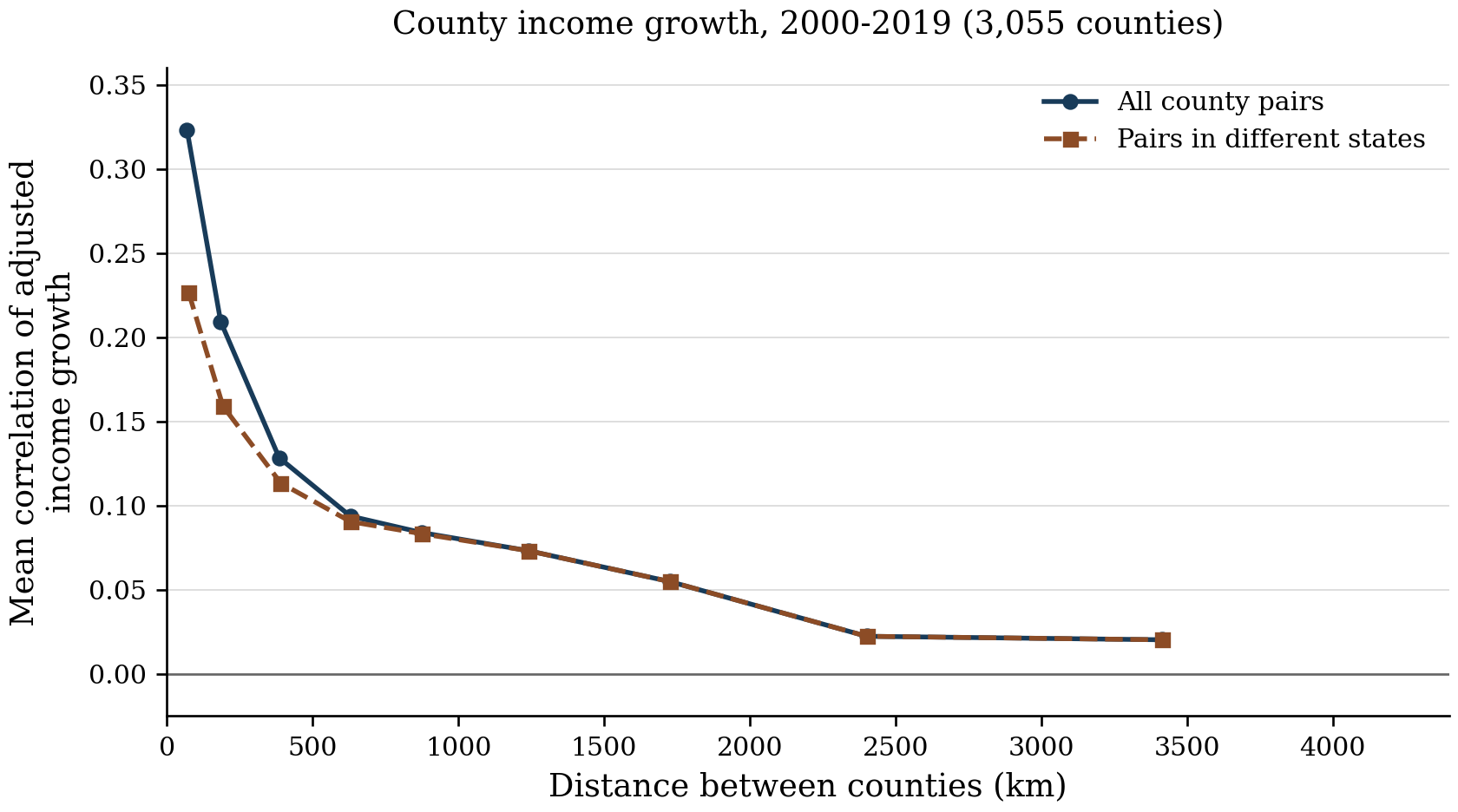}
 \caption{Spatial correlation of county income growth.
 Each county's annual log growth in nominal per capita personal
 income is regressed on a constant and U.S. per capita income
 growth over 2000-2019. The figure reports mean pairwise
 correlations of the resulting residuals for 3,055 counties.
 The solid line uses all county pairs; the dashed line uses only
 pairs in different states. Distance ranges are 0-100, 100-250,
 250-500, 500-750, 750-1,000, 1,000-1,500, 1,500-2,000,
 2,000-3,000, and 3,000-6,000 km.}
 \label{fig:spatial_income_dependence}
\end{figure}

Figure~\ref{fig:spatial_income_dependence} shows that nearby
counties have substantially stronger income comovement. The mean
correlation is 0.323 for pairs less than 100 km apart, compared
with 0.022 for pairs at least 2,000 km apart. This pattern also
appears across state boundaries: among pairs in different states
less than 100 km apart, the mean correlation is 0.227.
Thus, geographic distance remains informative about dependence
between regional economic conditions after accounting for
their different responses to national income growth.

To account for dependence through either common or nearby nodes,
define the distance between dyads $(i,j)$ and $(p,q)$ as
\begin{equation}
 \Delta((i,j),(p,q))
 =\min\{d_{ip},d_{iq},d_{jp},d_{jq}\}.
 \label{eq:endpoint_distance}
\end{equation}
This is the shortest geographic distance between a node in one
dyad and a node in the other. It equals zero when the dyads
share a node. Our spatial specification also allows nonzero
covariances when $\Delta>0$, with dependence controlled by the
distances between the node locations.

A related approach is developed by \citet{canen2024inference},
who allow dependence between dyads connected through an observed
network and provide a consistent network HAC estimator.
In their construction, two active dyads are directly linked
when they share a node. Their distance,
$\rho((i,j),(p,q))$, is the smallest number of links needed to
connect the two dyads through a chain of active dyads.
For the complete dyadic sample considered here,
\[
 \rho((i,j),(p,q))
 =
 \begin{cases}
 0, & (i,j)=(p,q),\\
 1, & \text{the two distinct dyads share a node},\\
 2, & \text{the dyads share no node}.
 \end{cases}
\]
The last case follows because $(i,j)$ is linked to $(i,p)$,
which is linked to $(p,q)$. The intermediate dyad $(i,p)$
is always observed in a complete sample, regardless of the
geographic distance between $i$ and $p$.

The distinction can be seen in a simple example. In one dimension,
suppose nodes $i,j,p,q$ are located at $1,2,h+1,h+2$,
respectively, where $h\ge2$ is an integer. Then
\[
 \rho((i,j),(p,q))=2,
 \qquad
 \Delta((i,j),(p,q))=h-1.
\]
As the two pairs become geographically farther apart, $\Delta$
increases while $\rho$ remains two. A kernel based on $\rho$
therefore assigns the same weight to all disjoint dyad pairs
in a complete sample. A kernel based on $\Delta$ can instead
give smaller weights to pairs whose nodes are farther apart.

These approaches accommodate different empirical settings.
The network approach is appropriate when dependence follows
connections among active dyads. Our approach uses geographic
distances to account for dependence between regional conditions
while retaining observations for every node pair.\footnote{The
network in \citet{canen2024inference} may itself reflect
geographical links. The distinction is between the number of
links along that network and the geographic distance between
node locations.}
The county-income illustration motivates this spatial
specification: nearby regions can experience correlated economic
conditions even after adjusting for national fluctuations.
When such conditions enter the disturbances of a regional-flow
regression, the covariance formula above shows how they can
generate dependence between dyads with no common node.
Using geographic distances allows inference to account for
this pattern.

For a distance cutoff $L\ge1$, consider dyad pairs satisfying
$\Delta((i,j),(p,q))\le L$. Under our lattice sampling scheme,
at most $CL^\nu$ nodes lie within distance $L$ of either node
in a fixed dyad. Pairing these nodes with the $n$ observed nodes
gives at most $CnL^\nu$ such dyads. Their covariance contributions
need not be negligible in the asymptotic variance. The bandwidth
conditions below account for this growing number of covariance
terms.

\subsection{Score decomposition and the choice of inference method}
\label{subsec:score_decomposition}

Spatial dependence determines which dyadic scores may be
correlated. The contribution of each node to the score also
determines whether a Gaussian approximation is appropriate.
We use the following decomposition to explain this distinction
and the role of subsampling when the limit is non-Gaussian.

Let $F$ be the common marginal distribution of $Z(\bm{s})$.
For a scalar dyadic variable $v_{ij}$, define
\[
 \begin{aligned}
 E_\perp[v_{12}]
 &:=\iint E[v_{12}\mid Z_1=z,Z_2=z']\,dF(z)\,dF(z'),\\
 E_\perp[v_{12}\mid Z_1=z]
 &:=\int E[v_{12}\mid Z_1=z,Z_2=z']\,dF(z').
 \end{aligned}
\]
The inner expectations average over the dyad shock in
\eqref{eq:kernel_rep}, holding both node variables fixed.
Thus, $E_\perp[v_{12}]$ averages over two independent node
draws from $F$ and the independent dyad shock.
The conditional expectation $E_\perp[v_{12}\mid Z_1=z]$
holds the first node at $z$ and averages over the second
node and the dyad shock. Conditioning on the second node
is defined similarly. Ordinary $E$ instead uses the actual
joint distribution of the spatially dependent nodes. 

Define the mean score for two independent node draws and the
contribution of node $i$ by
\begin{equation}
 \bm{\mu}=E_\perp[\bm{s}_{12}],\qquad
 \bm{\gamma}_i
 =E_\perp[\bm{s}_{ij}\mid Z_i]-\bm{\mu},
 \  j\ne i.
 \label{eq:gamma_def}
\end{equation}
Because the second node is averaged over its common marginal
distribution, $\bm{\gamma}_i$ does not depend on the choice
of $j$. The actual mean $E[\bm{s}_{ij}]$, however, may depend
on the difference between the node locations,
$\bm{s}_j-\bm{s}_i$. The regression moment restriction and
regularity conditions imposed below give
$E[\bm{s}_{ij}]=\bm0$ for every observed pair and $\bm{\mu}=\bm0$.

Define the remaining components by
\begin{align}
 \bm{\xi}_{ij}=\bm{\xi}(Z_i,Z_j)
 &=E[\bm{s}_{ij}\mid Z_i,Z_j]
   -\bm{\mu}-\bm{\gamma}_i-\bm{\gamma}_j,
 \label{eq:xi_def}\\
 \bm{\zeta}_{ij}
 &=\bm{s}_{ij}-E[\bm{s}_{ij}\mid Z_i,Z_j].
 \label{eq:zeta_def}
\end{align}
The score then satisfies
\begin{equation}
 \bm{s}_{ij}
 =\bm{\mu}+\bm{\gamma}_i+\bm{\gamma}_j
  +\bm{\xi}_{ij}+\bm{\zeta}_{ij},
 \label{eq:score_decomp}
\end{equation}
where
\[
 \begin{aligned}
 E[\bm{\gamma}_i]&=\bm0,
 &E[\bm{\zeta}_{ij}\mid Z_i,Z_j]&=\bm0,&E_\perp[\bm{\xi}_{ij}\mid Z_i]&=\bm0,
 &E_\perp[\bm{\xi}_{ij}\mid Z_j]&=\bm0.
 \end{aligned}
\]
The last two equalities follow from
\[
\begin{aligned}
&E_\perp\!\left[E[\bm s_{ij}\mid Z_i,Z_j]\mid Z_i=z\right]
=\int E[\bm s_{ij}\mid Z_i=z,Z_j=z']\,dF(z')=E_\perp[\bm s_{ij}\mid Z_i=z]\\
&E_\perp[\bm\gamma_j\mid Z_i]
=E[\bm\gamma_j]
=\int\left\{
  \int E[\bm s_{ij}\mid Z_i=z,Z_j=z']\,dF(z)
  -\bm\mu
  \right\}dF(z')=E_\perp[\bm s_{ij}]-\bm\mu=\bm0.\end{aligned}
\]

The component $\bm{\gamma}_i+\bm{\gamma}_j$ describes the
contribution of each node separately. The component
$\bm{\xi}_{ij}$ describes the remaining interaction between
the two node variables. Holding either node fixed and averaging
over an independent draw of the other node from $F$ gives zero.
It does not imply that $\bm{\xi}_{ij}$ is uncorrelated with
$\bm{\gamma}_i$ under the actual spatial distribution.

\paragraph{Example 1.}
{\itshape
Let $Z(\bm{s})$ be a stationary Gaussian spatial field with
mean zero and variance one, and write $r_{ij}=E[Z_iZ_j]$.
Consider the scalar score
\(
 s_{ij}=Z_i^2Z_j+Z_iZ_j^2.
\)
Both $E[s_{ij}]$ and $E_\perp[s_{ij}]$ equal zero. Integrating
over the marginal distribution of $Z_j$, while holding $Z_i$
fixed, gives
\[
 \gamma_i=Z_i,
 \qquad
 \xi_{ij}=(Z_i^2-1)Z_j+(Z_j^2-1)Z_i.
\]
Under the actual Gaussian distribution,
$E[Z_i^3Z_j]=3r_{ij}$ and
$E[Z_i^2Z_j^2]=1+2r_{ij}^2$. Hence
\[
 \begin{aligned}
 E[\gamma_i\xi_{ij}]
 &=E[Z_i^3Z_j]-E[Z_iZ_j]
   +E[Z_i^2Z_j^2]-E[Z_i^2]=3r_{ij}-r_{ij}+(1+2r_{ij}^2)-1=2r_{ij}(1+r_{ij}).
 \end{aligned}
\]
Thus, although $E_\perp[\xi_{ij}\mid Z_i]=0$,
it can be correlated with $\gamma_i$ under spatial dependence.
\par}

Finally, $\bm{\zeta}_{ij}$ is the part of the score due to the
dyad shock after holding both node variables fixed.
Assumption~\ref{ass:representation} implies that, conditional
on all node variables, the $\bm{\zeta}_{ij}$ are independent
and have mean zero across distinct dyads.

Each node appears in exactly $n-1$ dyads. Summing
\eqref{eq:score_decomp} over all unordered pairs therefore gives
\begin{align}
 \frac1{M_n}\sum_{1\le i<j\le n}\bm{s}_{ij}
 &=\bm{\mu}
   +\frac{n-1}{M_n}\sum_{i=1}^n\bm{\gamma}_i
   +\frac1{M_n}\sum_{1\le i<j\le n}
       (\bm{\xi}_{ij}+\bm{\zeta}_{ij})\nonumber\\
 &=\bm{\mu}
   +\frac2n\sum_{i=1}^n\bm{\gamma}_i
   +\frac1{M_n}\sum_{1\le i<j\le n}
       (\bm{\xi}_{ij}+\bm{\zeta}_{ij}).
 \label{eq:average_score_decomp}
\end{align}

When the node sum has positive long-run variance, it determines
the $\sqrt n$ Gaussian limit under the conditions below. If both
$\bm\gamma_i$ and $\bm\xi_{ij}$ vanish, only the centered dyad
shocks remain and the Gaussian limit instead has rate $n$.
In either case, a consistent variance estimator and normal
critical values give conventional inference. Section~\ref{sec:estimators}
provides spatial kernel and corrected jackknife estimators for
this purpose.

The degnerate case $\bm\gamma_i=\bm0$ with $\bm\xi_{ij}\ne\bm0$ requires a
different argument. The node average disappears, but interactions
between nodes can produce a non-Gaussian limit. For example,
consider the intercept regression
\[
 y_{ij}=\beta+Z_iZ_j,
 \qquad P(Z_i=1)=P(Z_i=-1)=\frac12,
\]
with independent node variables. Here $\gamma_i=0$,
$\xi_{ij}=Z_iZ_j$, and $\zeta_{ij}=0$. Since $Z_i^2=1$,
\[
 n(\widehat\beta-\beta)
 =\frac{n}{n-1}
   \left\{\left(\frac1{\sqrt n}\sum_{i=1}^n Z_i\right)^2-1\right\}
 \Rightarrow G^2-1,
 \qquad G\sim N(0,1).
\]
Thus, even in this independent-node special case, a variance
estimate does not justify using normal critical values.

To handle a possibly non-Gaussian limit, we apply the subsampling
approach of \citet{politis1994subsamples} to estimate the
distribution. Choose $q\to\infty$ with $q^\nu/n\to0$, and estimate
the regression using all dyads whose two nodes lie in a lattice
cube containing $q^\nu$ nodes. In the example, the same identity
holds with $q^\nu$ replacing $n$, so the scaled subsample error
has the same limit $G^2-1$.
More generally, let $\widehat{\bm\beta}_{\ell,q}$ denote the
estimate from cube $\ell$. Under the strongly degenerate model
and additional conditions stated below,
\[
 \begin{aligned}
 n\bm a^{\top}(\widehat{\bm\beta}-\bm\beta)&\Rightarrow W,\qquad q^\nu\bm a^{\top}(\widehat{\bm\beta}_{\ell,q}-\bm\beta)
 \Rightarrow W,
 \end{aligned}
\]
where $W$ need not be Gaussian. Keeping all dyads within each
cube preserves the dyadic structure and the spatial dependence
between nearby nodes. The common limit allows the distribution
of the subsample estimates to approximate the full-sample
distribution.

Although $\bm\beta$ is unknown, replacing it by the full-sample
estimate, $q^\nu\bm a^{\top}(\widehat{\bm\beta}_{\ell,q}-\widehat{\bm\beta})$, changes each scaled subsample error by at most
\[
 q^\nu\big|\bm a^{\top}(\widehat{\bm\beta}-\bm\beta)\big|
 =O_p(q^\nu/n)=o_p(1).
\]
Spatial weak dependence then permits averaging the distribution
indicators across cubes, even though the cubes may overlap.
At a regular target quantile, the resulting subsampling interval
has asymptotically exact coverage under the stated degenerate
model.  Theorems~\ref{thm:fast_nongaussian}
and~\ref{thm:conservative_coverage} establish it here for the
specified spatial dyadic interactions.

The subsampling interval uses the $n$ rate, so it cannot alone
cover the $\sqrt n$ Gaussian case. Section~\ref{subsec:conservative_procedure}
therefore combines it with the Gaussian jackknife interval by
taking the wider of the two intervals, which share the same center.
We call this the \emph{max jackknife-subsampling (MJS) interval}.
Under the conditions below, the MJS interval is asymptotically exact
in the two Gaussian cases and conservative in the specified
non-Gaussian case.\footnote{These are pointwise results for fixed data
generating processes, without a preliminary decision about which
case applies.}

\section{Variance Estimation and Subsampling Intervals}
\label{sec:estimators}

This section implements the inference methods motivated in
Section~\ref{subsec:score_decomposition}. We first define two
variance estimators for Gaussian inference. DN-Dyadic uses a
spatial endpoint kernel, and JK-DN-Dyadic uses neighborhood deletion
with a correction for dyads counted through both endpoints.
We then define the spatial subsampling interval and the MJS interval.

\subsection{Dependent-node dyadic CRVE}
\label{subsec:dn_dyadic_crve}

The conventional dyadic cluster-robust variance estimator (CRVE) assigns a nonzero weight only to dyads
sharing an endpoint. We enlarge this neighborhood using the minimum
endpoint distance in \eqref{eq:endpoint_distance}. This follows the
spatial HAC principle of retaining nearby score covariances.\footnote{See
\citet{conley1999gmm}, \citet{kelejian2007hac}, and
\citet{kim2011spatial} for spatial HAC estimation, and
\citet{kojevnikov2021limit} and \citet{canen2024inference} for
network-distance versions. Our endpoint weights use geographical
node distances, and the bandwidth conditions below account for the
number of dyadic score pairs retained. The consistency argument
therefore combines spatial kernel methods with dyadic pair counts.}
For $L>0$, set
\begin{equation}
 K^{\mathrm H}_{rs}=K(d_{rs}/L),\qquad K(u)=(1-u)_+^4(4u+1),\quad u\ge0.
 \label{eq:bartlett}
\end{equation}
Here $K$ is the Wendland $C^2$ kernel; see \citet{wendland1995piecewise}.
The superscript $\mathrm H$ identifies the HAC weight between nodes
$r$ and $s$; the bandwidth dependence is suppressed in the notation.
The matrix $[K^{\mathrm H}_{rs}]_{r,s=1}^n$ is positive semidefinite in the
one- and two-dimensional Euclidean settings considered here.  Define 
\begin{equation}
 \widehat{\bm{\Sigma}}_{\mathrm{DN}}
 =\sum_{(i,j)\in\mathcal D_n}\sum_{(p,q)\in\mathcal D_n}
 K\{\Delta((i,j),(p,q))/L\}\widehat{\bm{s}}_{ij}\widehat{\bm{s}}_{pq}^{\top},
 \label{eq:dn_dyadic_meat}
\end{equation}
and
\begin{equation}
 \widehat{\bm{V}}_{\mathrm{DN}}
 =(\bm{X}^{\top}\bm{X})^{-1}\widehat{\bm{\Sigma}}_{\mathrm{DN}}(\bm{X}^{\top}\bm{X})^{-1}.
 \label{eq:dn_dyadic_crve}
\end{equation}
On a unit lattice, $0<L\le1$ retains exactly the dyad pairs
sharing an endpoint, recovering the conventional dyadic CRVE.
Although the node kernel is positive semidefinite, applying it
to the minimum of the four endpoint distances need not produce
a positive-semidefinite dyad kernel. Consequently,
\eqref{eq:dn_dyadic_meat} need not be positive semidefinite
in finite samples.

Appendix~\ref{app:node_score_crve} considers the positive-semidefinite
node-score estimator $\widehat{\bm V}_{\mathrm{DN},0}$.
Subtracting $\widehat{\bm V}_I$, defined in \eqref{eq:jk_block},
gives an estimator asymptotically equivalent to
$\widehat{\bm V}_{\mathrm{DN}}$ in both Gaussian cases.
The appendix gives the construction, the duplication argument,
and the precise equivalence rates in
Proposition~\ref{prop:double_counting}.
An optional adjustment that replaces negative eigenvalues by zero
is discussed in Section~\ref{subsec:gaussian inference}.

The bandwidth $L$ determines how far each neighborhood extends
from its center. On the $\nu$-dimensional lattice, the number
of nodes in each neighborhood is of order $L^\nu$. A larger radius retains more
spatial covariance terms. We use a deterministic diverging sequence
satisfying Assumption~\ref{ass:bandwidth}; an explicit choice is given
in Appendix~\ref{app:implementation}. 

\subsection{Row-column spatial block jackknife}
\label{subsec:jk}

Jackknife variance estimators offer a practical alternative to
conventional CRVEs. In one-way clustered regressions, they can
alleviate downward bias and improve finite-sample inference
\citep{hansen2022jackknife,mackinnon2023fast}. Motivated by these
advantages, we construct a spatial block jackknife for dyadic
data by deleting neighborhoods of nodes and re-estimating the
regression using the remaining dyads.

Unlike the fixed-size blocks commonly used in the block jackknife
\citep{kunsch1989jackknife,liu1992moving,lahiri2002jackknife,hounyo2025jackknife},
our blocks use a common spatial distance cutoff but may contain
different numbers of nodes. Specifically, for each node
$\ell=1,\ldots,n$, define
\begin{equation}
 B_\ell=\{i:d_{i\ell}\le L\},\qquad m_\ell=|B_\ell|.
 \label{eq:block}
\end{equation}
Thus, $B_\ell$ contains all observed nodes within distance $L$
of node $\ell$. Blocks near the boundary may contain fewer nodes
because part of the neighborhood lies outside the observation
region.

Deleting $B_\ell$ removes every dyad with at least one endpoint
in the block. Let
\begin{equation}
 \mathcal A_\ell
 =\{(i,j)\in\mathcal D_n:i\in B_\ell\text{ or }j\in B_\ell\},
 \qquad
 \mathcal D_{n,-\ell}=\mathcal D_n\setminus\mathcal A_\ell
 \label{eq:touch_block}
\end{equation}
denote the deleted and retained dyads, respectively. The retained
sample contains all dyads among the remaining $n-m_\ell$ nodes.
The corresponding regression estimator is
\begin{equation}
 \widetilde{\bm{\beta}}_{(-\ell)}
 =\left(
   \sum_{(i,j)\in\mathcal D_{n,-\ell}}
       \bm{x}_{ij}\bm{x}_{ij}^{\top}
   \right)^+
   \sum_{(i,j)\in\mathcal D_{n,-\ell}}\bm{x}_{ij}y_{ij},
 \label{eq:beta_minus_block}
\end{equation}
where $+$ denotes the Moore-Penrose generalized inverse.
Under the assumptions below, we show that matrices inside the parentheses
in \eqref{eq:beta_minus_block} are nonsingular uniformly in $\ell$ with probability
tending to one. Thus, $\widetilde{\bm{\beta}}_{(-\ell)}$ is estimated  after
removing all dyads involving node $\ell$ or any node within
distance $L$ of it. Weighting each deletion by the inverse of its block's node count
gives the uncorrected jackknife CRVE
\begin{equation}
 \widehat{\bm{V}}^{\mathrm{JK}}_{\mathrm{DN},0}
 =\sum_{\ell=1}^n\frac1{m_\ell}
 (\widetilde{\bm{\beta}}_{(-\ell)}-\widehat{\bm{\beta}})
 (\widetilde{\bm{\beta}}_{(-\ell)}-\widehat{\bm{\beta}})^{\top}.
 \label{eq:jk_block_uncorrected}
\end{equation}

The uncorrected estimator is positive semidefinite, but it can
count a dyad's contribution more than once. The issue is easiest
to see when deleting one node at a time: dyad $(i,j)$ is removed
when node $i$ is deleted and again when node $j$ is deleted.
Its squared score therefore enters twice in the score approximation
to the jackknife variance, although its own variance enters the
variance of the full score only once. With spatial blocks, the
weights $1/m_\ell$ adjust for block size, but a dyad can still
contribute through neighborhoods of both endpoints.

This extra contribution is asymptotically negligible when the
first-order node component dominates. In the degenerate Gaussian
case, uncertainty comes
entirely from the dyad shocks. The uncorrected estimator then
converges, after appropriate scaling, to twice the asymptotic
variance; see Proposition~\ref{prop:double_counting}. A correction
is therefore needed for asymptotically exact inference in both
Gaussian cases.

To account for the actual block deletions, define the total
deletion weight of dyad $(i,j)$ by
\begin{equation}
 K^\cup_{ij}
 =\sum_{\ell=1}^n\frac1{m_\ell}
       \ind\{i\in B_\ell\text{ or }j\in B_\ell\}.
 \label{eq:dyad_incidence_weight}
\end{equation}
The superscript $\cup$ indicates that a block removes the dyad
if it contains either endpoint. In the score approximation,
$K^\cup_{ij}$ is the weight on the dyad's own squared score.
We adjust this weight to one by subtracting $K^\cup_{ij}-1$,
giving the corrected estimator
\begin{equation}
 \widehat{\bm{V}}^{\mathrm{JK}}_{\mathrm{DN}}
 =\widehat{\bm{V}}^{\mathrm{JK}}_{\mathrm{DN},0}-\widehat{\bm{V}}_I,\qquad
 \widehat{\bm{V}}_I=(\bm{X}^{\top}\bm{X})^{-1}
   \left\{
   \sum_{1\le i<j\le n}(K^\cup_{ij}-1)
       \widehat{\bm{s}}_{ij}\widehat{\bm{s}}_{ij}^{\top}
   \right\}
   (\bm{X}^{\top}\bm{X})^{-1}.
 \label{eq:jk_block}
\end{equation}
 When $0<L<1$, each block contains only its center node, so
$m_\ell=1$ and $K^\cup_{ij}=2$ for every dyad. The correction 
then subtracts the usual White variance estimator.  For larger
spatial blocks, $K^\cup_{ij}$ accounts for blocks containing
both endpoints and for smaller blocks near the boundary.
We compute the correction $\widehat{\bm{V}}_I$ using conventional sandwich form rather than a  jackknife variant, with the
aim of improving finite-sample performance and reducing the
likelihood of negative eigenvalues, see \citet{hounyo2025jackknife} for more details. Nevertheless, the corrected
estimator need not be positive semidefinite in finite samples.

\subsection{Spatial subsampling and the MJS interval}
\label{subsec:conservative_procedure}

The preceding CRVEs support inference in two Gaussian cases:
the nondegenerate case and the degenerate Gaussian case with
$\bm\gamma_i=\bm0$ and $\bm\xi_{ij}=\bm0$. In degenerate cases where $\bm\gamma_i=\bm0$ but
$\bm\xi_{ij}\neq\bm0$, however, the limiting distribution may
be non-Gaussian, so variance estimation alone does not
justify Gaussian critical values. To accommodate
this possibility, we now present the feasible spatial
subsampling procedure motivated in
Section~\ref{subsec:score_decomposition}. For a fixed contrast
$\bm a^{\top}\bm\beta$, we estimate the distribution using
regressions on smaller spatial cubes and combine the resulting
interval with the corrected Gaussian jackknife interval.

Choose an integer cube side length $q\to\infty$ with
$q^\nu/n\to0$. Let $I_{\ell,q}$ enumerate all
$q\times\cdots\times q$ lattice cubes contained in $\Lambda_n$.
Their number is
\(
 N_q=\prod_{k=1}^{\nu}(N_{k,n}-q+1)\sim n.
\)
Let each subsample contain exactly $q^\nu$ nodes and only dyads with
\emph{both} endpoints in its cube. Let $\widehat{\bm{\beta}}_{\ell,q}$
be its OLS estimator, using a Moore-Penrose inverse if necessary.
These equal-shape subsamples serve distributional calibration.
 To measure how its
estimate differs from the full-sample estimate, we calculate
\[
 q^\nu
 \left|\bm a^\top
 (\widehat{\bm\beta}_{\ell,q}-\widehat{\bm\beta})\right|,
 \qquad  \ell=1,\ldots,N_q.
\]
The absolute value is used because we construct a two-sided
confidence interval. The factor $q^\nu$ corresponds to the
number of nodes in the subsample, just as the factor $n$
scales the full-sample estimation error in the degenerate case.

We estimate the distribution of these quantities by
\begin{equation}
 \widehat F_{n,q}(t)
 =\frac1{N_q}\sum_{\ell=1}^{N_q}
 \ind\left\{
 q^\nu\left|\bm a^\top
 (\widehat{\bm\beta}_{\ell,q}-\widehat{\bm\beta})\right|
 \le t\right\}.
 \label{eq:subsampling_cdf}
\end{equation}
Thus, $\widehat F_{n,q}(t)$ is the empirical distribution of the subsampling statistics, and
its $(1-\alpha)$ quantile is
\(
 \widehat q_{n,q}(1-\alpha)
 =\inf\{t\ge0:\widehat F_{n,q}(t)\ge1-\alpha\}.
\)

Under the degenerate case, this quantile ${\widehat q_{n,q}(1-\alpha)}/{n}$ can estimate the
corresponding quantile of
$|\bm a^\top(\widehat{\bm\beta}-\bm\beta)|$.
The resulting subsampling confidence interval has asymptotically
exact coverage in both degenerate cases where
$\bm\gamma_i=\bm0$: the Gaussian case and the possibly non-Gaussian
case.\footnote{Asymptotic exactness is understood pointwise throughout. Thus, this result does not contradict Proposition 4.1 of \citet{menzel2021bootstrap}, which establishes the impossibility of uniform consistency.} However, the same scaling does not provide valid inference in
the nondegenerate case. In that case, the estimation error is
of order $n^{-1/2}$, whereas
\[
 \frac{\widehat q_{n,q}(1-\alpha)}{n}
 =O_p\left(\frac{q^{\nu/2}}{n}\right)
 =o_p(n^{-1/2}),
\]
because $q^\nu/n\to0$. The subsampling interval is therefore
too narrow relative to the estimation error: its coverage
tends to zero, and the corresponding test rejects a true
null hypothesis with probability tending to one.

To accommodate the nondegenerate case and two degenerate cases, we combine the subsampling interval
with the Gaussian jackknife interval. Specifically, define
\[
 r_{n,\alpha}
 =\max\left\{
 z_{1-\alpha/2}
 \sqrt{\max\{
 \bm a^\top\widehat{\bm V}^{\mathrm{JK}}_{\mathrm{DN}}\bm a,0\}},
 \frac{\widehat q_{n,q}(1-\alpha)}{n}
 \right\},
\]
where $z_p$ is the $p$ quantile of the standard normal
distribution. The first term uses the jackknife variance
estimate and a Gaussian critical value; the second uses
the subsampling quantile. The inner maximum ensures that
the square root is defined if the estimated contrast
variance is negative in a finite sample.

The MJS confidence interval is
\begin{equation}
 C_{n,\alpha}
 =
 \left[
 \bm a^\top\widehat{\bm\beta}-r_{n,\alpha},
 \ \bm a^\top\widehat{\bm\beta}+r_{n,\alpha}
 \right].
 \label{eq:conservative_interval}
\end{equation}
Both component intervals are centered at
$\bm a^\top\widehat{\bm\beta}$. Taking the maximum therefore
makes $C_{n,\alpha}$ contain both the Gaussian jackknife interval
and the subsampling interval, without requiring
independence between the two procedures.

In the nondegenerate case, the subsampling term
$\widehat q_{n,q}(1-\alpha)/n$ is asymptotically smaller
than the Gaussian jackknife term. Consequently,
$C_{n,\alpha}$ equals the Gaussian jackknife interval with
probability tending to one. When the $n$-rate limit is
non-Gaussian, coverage follows from containing the
subsampling interval under the conditions below.

\section{Asymptotic validity}
\label{sec:theory}

We first establish the two Gaussian limits and the consistency
of the proposed variance estimators. We then allow a nonzero
node interaction under degeneracy and establish the
validity of subsampling and the MJS confidence interval.

\subsection{Gaussian limits}

For a matrix $\bm A$, write $\bm A>0$ when it is positive
definite. Define
\[
 \bm Q_n
 =\frac1{M_n}\sum_{1\le i<j\le n}\bm x_{ij}\bm x_{ij}^{\top},
 \qquad
 \bm Q=E_\perp[\bm x_{12}\bm x_{12}^{\top}].
\]
Under the conditions below, $\bm Q_n$ converges to $\bm Q$.
The expectation $E_\perp$ uses two independent node draws,
as defined in Section~\ref{subsec:score_decomposition}.
It need not equal the actual second moment of the regressors
for a fixed pair of nearby nodes.

\begin{assumption}[Node projection and moments]
\label{ass:projection}
For the $\lambda>1$ in Assumption~\ref{ass:mixing},
$\bm Q>0$, and there exist $\delta>0$ and $C<\infty$ such that,
for every $n$ and $1\le i<j\le n$,
\[
 E[\bm x_{ij}u_{ij}]=\bm0,\qquad
 E\|\bm x_{ij}\|^{8(\lambda+\delta)}\le C,\qquad
 E|u_{ij}|^{8(\lambda+\delta)}\le C.
\]
\end{assumption}

Stationarity gives each node the same marginal distribution,
but the joint distribution of two nodes can depend on their
relative locations. The dyadic moments therefore need not be
identical, and Assumption~\ref{ass:projection} imposes a common
bound across all dyads and sample sizes. The following lemma
transfers this bound to independent node draws and establishes
that the mean score in \eqref{eq:score_decomp} is zero.

\begin{lemma}[Moments and mean score under independent node draws]
\label{lem:independent_node_moments}
Under Assumptions~\ref{ass:representation},
\ref{ass:mixing}, and~\ref{ass:projection},
\(
 E_\perp\|\bm x_{12}\|^{8(\lambda+\delta)}\le C,
 \ 
 E_\perp|u_{12}|^{8(\lambda+\delta)}\le C,
 \ 
 \bm\mu=E_\perp[\bm x_{12}u_{12}]=\bm0.
\)
\end{lemma}

With $\bm\mu=\bm0$, the score decomposition separates the
contributions of the node projection, the node interaction,
and the dyad shock. The two Gaussian cases below are governed
by the variance of the node projection and that of the
residual dyad component.

Let $\bm\gamma(\bm s)$ denote the node projection in
\eqref{eq:gamma_def} evaluated at $Z(\bm s)$. Its long-run
variance is
\begin{equation}
 \bm\Omega_\gamma
 =\sum_{\bm h\in\mathbb Z^\nu}
 E[\bm\gamma(\bm0)\bm\gamma(\bm h)^{\top}].
 \label{eq:longrun}
\end{equation}
This matrix accounts for covariances between node projections
at all spatial distances. The series converges absolutely
under the stated mixing and moment conditions.

For the residual dyad component, define
\begin{equation}
 \bm\Omega_\zeta
 =\lim_{n\to\infty}\frac{n^2}{M_n^2}
 \sum_{1\le i<j\le n}
 E[\bm\zeta_{ij}\bm\zeta_{ij}^{\top}]
 =2E_\perp[\bm\zeta_{12}\bm\zeta_{12}^{\top}].
 \label{eq:omega_zeta}
\end{equation}
The existence of this limit and the final equality follow
from Lemma~\ref{lem:dyadic_lln}. 

\begin{assumption}[Variance]
\label{ass:variance}
One of the following cases holds:
(i) $\bm\Omega_\gamma>0$ or
(ii) $\bm\gamma_i=\bm0$ and $\bm\xi_{ij}=\bm0$ almost surely
 for every dyad, and $\bm\Omega_\zeta>0$.
\end{assumption}

In case (i), the first-order node projection determines the
leading variation in the estimator. In case (ii), both node
components vanish, so $\bm s_{ij}=\bm\zeta_{ij}$ almost surely.
These residual scores are independent and centered conditional
on the full node field. The two cases parallel the Gaussian
cases studied in the two-way clustering literature; see, for
example, equations~(16)-(17) of \citet{mackinnon2021wild}
and Assumption~2 of \citet{chiang2023standard}.
The two cases are not exhaustive, and Section~\ref{subsec:nongaussian} considers
$\bm\gamma_i=\bm0$ with a nonzero interaction
$\bm\xi_{ij}$, which can produce a non-Gaussian limit.%
\footnote{Our analysis of degeneracy imposes
$\bm\gamma_i=\bm0$ almost surely, a stronger condition than
$\bm\Omega_\gamma=\bm0$, for simplicity. Under spatial dependence, nonzero
node projections can have zero long-run variance through
cancellation of spatial covariances. Their sums can nevertheless
contribute at the faster rate, requiring a separate analysis;
related cases with serially dependent time effects are treated
in \citet{hounyo2025jackknife}. We impose degeneracy
to focus on the interaction and dyad-shock components, with
the node projection vanishing exactly in both the full sample
and every spatial subsample.}

\begin{theorem}[Limit distribution]
\label{thm:clt}
Suppose Assumptions~\ref{ass:representation}-\ref{ass:variance}
hold. In case (i),
 $\sqrt n(\widehat{\bm\beta}-\bm\beta)
 \Rightarrow N(\bm0,\bm V_\gamma),
 \ 
 \bm V_\gamma
 =4\bm Q^{-1}\bm\Omega_\gamma\bm Q^{-1}.$
In case (ii),
$ n(\widehat{\bm\beta}-\bm\beta)
 \Rightarrow N(\bm0,\bm V_\zeta),
 \  
 \bm V_\zeta
 =\bm Q^{-1}\bm\Omega_\zeta\bm Q^{-1}.$
\end{theorem}

 The resulting convergence rates are
$\sqrt n$ and $n$, respectively. The factor four in $\bm V_\gamma$ comes from the coefficient
$2/n$ on the node sum in \eqref{eq:average_score_decomp}.
In case (ii), the number of conditionally independent residual
dyad shocks is of order $n^2$, giving the faster rate $n$.

\subsection{Gaussian inference}
\label{subsec:gaussian inference}
To make Gaussian inference feasible, we next establish that
the proposed variance estimators consistently estimate the
variance at the appropriate rate. This requires a restriction
on the spatial bandwidth.

\begin{assumption}[Spatial bandwidth]
\label{ass:bandwidth}
As $n\to \infty$, $L\to\infty$ and $L^{3\nu}/n\to0$.
\end{assumption}

The condition $L\to\infty$ allows the estimators to account
for dependence over increasing distances. The restriction
$L^{3\nu}/n\to0$ controls the estimation error from including
more covariance terms. For $\nu=2$, it requires
$L=o(n^{1/6})$. Together with the lattice sampling scheme,
it implies $m_\ell\asymp L^\nu$ uniformly,
$\max_\ell m_\ell^2/n\to0$, and a vanishing fraction of nodes
within distance $2L$ of the observation boundary, as demonstrated in Lemma \ref{lem:spatial_weights}. 

\begin{theorem}[Consistency of the spatial variance estimators]
\label{thm:hac_jk}
Suppose Assumptions~\ref{ass:representation}-\ref{ass:bandwidth}
hold. In case (i),
 $n\widehat{\bm V}_{\mathrm{DN}}\to^P\bm V_\gamma$
 and
 $n\widehat{\bm V}^{\mathrm{JK}}_{\mathrm{DN}}
 \to^P\bm V_\gamma$.
In case (ii),
 $n^2\widehat{\bm V}_{\mathrm{DN}}\to^P\bm V_\zeta$
 and 
$n^2\widehat{\bm V}^{\mathrm{JK}}_{\mathrm{DN}}
 \to^P\bm V_\zeta.$
Consequently, for each fixed nonzero $\bm a$,
\begin{equation}
 \frac{\bm a^{\top}(\widehat{\bm\beta}-\bm\beta)}
 {\sqrt{\bm a^{\top}\widehat{\bm V}_{\mathrm{DN}}\bm a}}
 \Rightarrow N(0,1),
 \qquad
 \frac{\bm a^{\top}(\widehat{\bm\beta}-\bm\beta)}
 {\sqrt{\bm a^{\top}
 \widehat{\bm V}^{\mathrm{JK}}_{\mathrm{DN}}\bm a}}
 \Rightarrow N(0,1).
 \label{eq:tstat_validity}
\end{equation}
The estimated variances are positive with probability
tending to one.
\end{theorem}

Theorem~\ref{thm:hac_jk} establishes consistency of both variance
estimators and asymptotically exact Gaussian inference in the
two cases of Assumption~\ref{ass:variance}.
Appendix~\ref{app:node_score_crve} explains why the duplication
correction is negligible in case (i) but necessary in case (ii).

Although the corrected variance estimators are consistent in both cases,
they need not be positive semidefinite in finite samples.
An optional adjustment replaces their negative eigenvalues
by zero. Specifically, for either corrected estimator, write
\(
 \widehat{\bm V}
 =\bm U\operatorname{diag}(\lambda_1,\ldots,\lambda_K)
 \bm U^{\top},
\)
where $\bm U$ contains its orthonormal eigenvectors, and set
\[
 \widehat{\bm V}_+
 =\bm U\operatorname{diag}
 \{\max(\lambda_1,0),\ldots,\max(\lambda_K,0)\}
 \bm U^{\top}.
\]
As shown in Theorem \ref{thm:hac_jk}, in either Gaussian case, the appropriately scaled estimator
converges to a positive-definite matrix, so
$P(\widehat{\bm V}_+=\widehat{\bm V})\to1$.
The adjustment therefore preserves the stated limits.
We note, however, that in finite samples inference based
on this adjustment is generally not invariant under nonsingular
affine transformations of the parameters.

\subsection{Non-Gaussian limits and subsampling coverage}
\label{subsec:nongaussian}

The degenerate Gaussian case above requires $\bm\xi_{ij}=\bm0$, and we now allow a nonzero
node interaction, which can produce a non-Gaussian limit.
The following representation provides sufficient conditions
for deriving the limit explicitly.

\begin{assumption}[Degeneracy and interaction expansion]
\label{ass:finite_interaction}
The first-order projection satisfies $\bm\gamma_i=\bm0$ almost
surely. There exist measurable functions $f_1,f_2,\ldots$ with
$E[f_a(Z_i)]=0$ and real coefficients $b_{ab,k}=b_{ba,k}$ such that
\(
 \sup_{a\ge1}E|f_a(Z_i)|^{4(\lambda+\delta)}\le C,
 \ 
 \sum_{a=1}^{\infty}\sum_{b=1}^{\infty}|b_{ab,k}|<\infty,
 \  k=1,\ldots,K,
\)
where $\lambda$ and $\delta$ are as in
Assumption~\ref{ass:projection}.
For each integer $r\ge1$, define
\(
 \xi_{ij,k}^{(r)}
 =\sum_{a=1}^{r}\sum_{b=1}^{r}
       b_{ab,k}f_a(Z_i)f_b(Z_j).
\)
For every pair $i<j$ and coordinate $k$,
\[
 E\!\left[|\xi_{ij,k}-\xi_{ij,k}^{(r)}|^2\right]\to0,
 \qquad
 E_\perp\!\left[|\xi_{ij,k}-\xi_{ij,k}^{(r)}|^2\right]\to0
 \quad\text{as }r\to\infty.
\]
\end{assumption}

Assumption~\ref{ass:finite_interaction} removes the first-order
node projection while allowing interactions between node
characteristics to remain. Each interaction is approximated
by products of centered functions of the two nodes. The
approximation uses the same functions and coefficients under
both the actual spatial distribution and independent draws
from $F$. 

Approximations by sums of products are widely used to study
degenerate interactions. For example, 
\citet{menzel2021bootstrap} (cf. Assumption 2.2) uses an orthonormal spectral
expansion for interactions between independent row and column
effects, while \citet{leucht2012degenerate} develops
approximation arguments for degenerate $U$-statistics under
weak dependence. In our spatial setting, orthonormality of
functions at one node would not eliminate their correlations
across different nodes. We retain these correlations through
the covariance matrices below and do not require the functions
$f_a$ to be orthogonal.

To describe the limiting distribution, collect the first
$r$ node functions and their coefficients as
\(
 \bm f_i^{(r)}
 =(f_1(Z_i),\ldots,f_r(Z_i))^{\top}
 \) and \(
 \bm B_k^{(r)}=[b_{ab,k}]_{a,b=1}^{r}.
\)
The finite approximation then has the matrix form
\(
 \xi_{ij,k}^{(r)}
 =(\bm f_i^{(r)})^{\top}\bm B_k^{(r)}\bm f_j^{(r)}.
\)
For each fixed $r$, its sum over dyads can be expressed using
the sum of the node vectors and their sample second moments.
A spatial central limit theorem and a law of large numbers
therefore determine its limiting distribution. The proof
first lets $n\to\infty$ with $r$ fixed and then lets
$r\to\infty$ to recover the full interaction.

Write
$\bm f^{(r)}(\bm s)
=(f_1(Z(\bm s)),\ldots,f_r(Z(\bm s)))^{\top}$ and define
\[
 \begin{aligned}
 \bm C_f^{(r)}
 &=E[\bm f^{(r)}(\bm0)\bm f^{(r)}(\bm0)^{\top}],\qquad \bm\Omega_f^{(r)}
 =\sum_{\bm h\in\mathbb Z^\nu}
 E[\bm f^{(r)}(\bm0)\bm f^{(r)}(\bm h)^{\top}].
 \end{aligned}
\]
The matrix $\bm C_f^{(r)}$ describes variation at a single
node, whereas $\bm\Omega_f^{(r)}$ also accounts for dependence
across node locations. In particular,
\[
 \frac1{\sqrt n}\sum_{i=1}^n\bm f_i^{(r)}
 \Rightarrow N(\bm0,\bm\Omega_f^{(r)}).
\]

Let $\bm{\mathcal G}=(\mathcal G_1,\mathcal G_2,\ldots)$
be a centered Gaussian sequence such that, for every $r$,
\(
 \bm{\mathcal G}^{(r)}
 :=(\mathcal G_1,\ldots,\mathcal G_r)^{\top}
 \sim N(\bm0,\bm\Omega_f^{(r)}).
\)
The interaction limit is
\begin{equation}
 \bm{\mathcal Q}_B(\bm{\mathcal G})
 =\lim_{r\to\infty}
 \left(
 (\bm{\mathcal G}^{(r)})^{\top}
       \bm B_k^{(r)}\bm{\mathcal G}^{(r)}
 -\operatorname{tr}(\bm B_k^{(r)}\bm C_f^{(r)})
 \right)_{k=1}^{K},
 \label{eq:interaction_series_limit}
\end{equation}
where the limit is in mean square. Its existence follows
from the approximation bound established in the proof below. The trace term arises because a quadratic form in the node
sum includes terms with $i=j$, whereas the dyadic sample
contains only pairs with $i<j$. Removing these terms subtracts
\(
 \frac1n\sum_{i=1}^n
   (\bm f_i^{(r)})^{\top}\bm B_k^{(r)}\bm f_i^{(r)}
 \to^P
 \operatorname{tr}(\bm B_k^{(r)}\bm C_f^{(r)}).
\)

\begin{theorem}[Fast limit with node interactions]
\label{thm:fast_nongaussian}
Suppose Assumptions~\ref{ass:representation},
\ref{ass:mixing}, \ref{ass:projection}, and
\ref{ass:finite_interaction} hold, then
\begin{equation}
 n(\widehat{\bm\beta}-\bm\beta)
 \Rightarrow
 \bm Q^{-1}
 \{\bm{\mathcal Q}_B(\bm{\mathcal G})
   +\bm{\mathcal E}_\zeta\},
 \label{eq:fast_nongaussian_limit}
\end{equation}
where $\bm{\mathcal E}_\zeta\sim N(\bm0,\bm\Omega_\zeta)$ is
independent of $\bm{\mathcal G}$.
Singular $\bm\Omega_f^{(r)}$ and $\bm\Omega_\zeta$ are allowed.
The same limit holds for growing lattice-cube subsamples
after replacing $n$ by their number of nodes.
\end{theorem}

The limit combines a generally non-Gaussian series of quadratic
terms from the node interaction with an independent Gaussian
component from the dyad shocks. The common limit for the full sample and growing spatial
subsamples provides the basis for subsampling inference.
Under degenerate case, subsampling consistently
estimates the critical value for the absolute estimation
error. The MJS interval inherits this coverage because
it contains the subsampling interval.

\begin{theorem}[Coverage of the MJS interval]
\label{thm:conservative_coverage}
Let $0<\alpha<1$, then \begin{enumerate}[label=(\roman*)]
 \item Under the assumptions of Theorem~\ref{thm:hac_jk},
 \(
 P\{\bm a^{\top}\bm\beta\in C_{n,\alpha}\}
 \longrightarrow1-\alpha.
 \)

 \item Under the assumptions of Theorem~\ref{thm:fast_nongaussian},
let
\(
 W=\bm a^\top\bm Q^{-1}
 \{\bm{\mathcal Q}_B(\bm{\mathcal G})
   +\bm{\mathcal E}_\zeta\}.
\)
Suppose the cdf of $|W|$ is continuous and strictly increasing
at its $(1-\alpha)$ quantile $q_{1-\alpha}\in(0,\infty)$.
Then,
\[
 P\left\{
 n\left|\bm a^\top(\widehat{\bm\beta}-\bm\beta)\right|
 \le\widehat q_{n,q}(1-\alpha)
 \right\}
 \longrightarrow1-\alpha.
\]
Consequently, the MJS interval satisfies
\(
 \liminf_{n\to\infty}
 P\{\bm a^\top\bm\beta\in C_{n,\alpha}\}
 \ge1-\alpha.
\)
\end{enumerate}
\end{theorem}

Theorem~\ref{thm:conservative_coverage} establishes that the
MJS interval has asymptotically exact coverage in both
Gaussian cases and remains conservative in the
degenerate case, including when the limit is non-Gaussian.

\section{Node fixed-effect models}
\label{subsec:fixed_effects}

In dyadic data, researchers often include node fixed effects to account for
unobserved node characteristics that may be correlated with the
regressors. Consider
\begin{equation}
 y_{ij}=\bm x_{ij}^{\top}\bm\beta+\alpha_i+\alpha_j+u_{ij},
 \qquad 1\le i<j\le n.
 \label{eq:fe_model}
\end{equation}
Here $K$ counts slope regressors only; an intercept and regressors
of the form $c_i+c_j$ are absorbed by the node effects.
As in \citet{chiang2023standard}, the Gaussian results extend
to fixed-effect estimation after accounting for sample demeaning.
The $n$-rate case requires additional care because demeaning
can affect the limiting distribution.

Let $\ddot v_{ij}$ denote the residual from regressing a dyadic
variable $v_{ij}$ on the node indicators. The exact within transformation is given in \eqref{eq:fe_within}
of Appendix~\ref{app:fe_projection}. The fixed-effect
estimator $\widehat{\bm\beta}_{\mathrm{FE}}$ is OLS of
$\ddot y_{ij}$ on $\ddot{\bm x}_{ij}$.
For inference, use the DN-Dyadic and corrected jackknife
estimators from Section~\ref{sec:estimators}, replacing the
regressors and scores by $\ddot{\bm x}_{ij}$ and
$\widehat{\bm s}_{ij,\mathrm{FE}}
=\ddot{\bm x}_{ij}(\ddot y_{ij}
-\ddot{\bm x}_{ij}^{\top}\widehat{\bm\beta}_{\mathrm{FE}})$.
Every deleted-sample and cube regression refits the retained
node effects. Denote the resulting corrected jackknife by
$\widehat{\bm V}^{\mathrm{JK}}_{\mathrm{FE}}$.

Define the population counterpart of
demeaning, for $v=u$ or a coordinate of $\bm x$, by
\[
 v_{ij}^{\circ}
 =v_{ij}-E_\perp[v_{ij}\mid Z_i]
        -E_\perp[v_{ij}\mid Z_j]+E_\perp[v_{ij}].
\]
Use the earlier definitions of the score components and
limiting covariance matrices, now applied to
$\bm s_{ij}^{\circ}=\bm x_{ij}^{\circ}u_{ij}^{\circ}$ and
$\bm Q^\circ=E_\perp[\bm x_{12}^{\circ}\bm x_{12}^{\circ\top}]$,
and indicate this change by the superscript $\circ$.
Although $E_\perp[v_{ij}^{\circ}\mid Z_i]=0$, the product
$\bm x_{ij}^{\circ}u_{ij}^{\circ}$ can have a nonzero node
projection. Thus fixed effects do not generally remove the
need to account for dyadic and spatial dependence.

\begin{assumption}[Fixed-effect regularity]
\label{ass:fe}
Suppose Assumption~\ref{ass:representation} holds for
\((\bm x_{ij},u_{ij})\), and Assumption~\ref{ass:mixing} holds
with the stronger bound
\(\beta(\mathcal Z_A,\mathcal Z_B)
\le C|A||B|\rho^{\operatorname{dist}(A,B)}\)
for all finite nonempty disjoint sets \(A,B\), where
\(C<\infty\) and \(0<\rho<1\).
Assume that \(K\) is fixed,
\(E_\perp[\|\bm x_{12}\|+|u_{12}|]<\infty\),
\(\bm Q^\circ>0\), and
\(E_\perp[\bm s_{12}^{\circ}]=\bm0\).
In addition, require
\(\|\bm x_{ij}^{\circ}\|+|u_{ij}^{\circ}|\le C\)
almost surely for every \(1\le i<j\le n\), under both the
actual spatial distribution and independent node draws
from \(F\), with \(C\) independent of \(n,i,j\).
\end{assumption}

Geometric mixing and bounded population-demeaned variables
are stronger sufficient conditions that simplify control of
sample demeaning. No bound on $\alpha_i$ is required because
the transformation removes the additive effects exactly.

\begin{theorem}[Fixed effects in the Gaussian cases]
\label{thm:fe_gaussian}
Suppose Assumptions~\ref{ass:bandwidth} and \ref{ass:fe} hold.
\begin{enumerate}[label=(\roman*)]
\item If $\bm\Omega_{\gamma^\circ}>0$, the $\sqrt n$ limit and
variance consistency conclusions of
Theorems~\ref{thm:clt} and \ref{thm:hac_jk} continue to hold
for $\widehat{\bm\beta}_{\mathrm{FE}}$ and the fixed-effect
versions of both variance estimators.
\item If, for every dyad,
\begin{equation}
 E[u_{ij}^{\circ}\mid\mathcal Z,\bm X]=0,
 \qquad \bm\Omega_{\zeta^\circ}>0,
 \label{eq:fe_fast_exogeneity}
\end{equation}
the corresponding $n$-rate Gaussian limit and consistency
result of
Theorems~\ref{thm:clt} and \ref{thm:hac_jk} continue to hold for $\widehat{\bm\beta}_{\mathrm{FE}}$ and
$\widehat{\bm V}^{\mathrm{JK}}_{\mathrm{FE}}$.
Here $\mathcal Z$ denotes the full node field.
\end{enumerate}
\end{theorem}

The distinction between the two rates follows from an exact
identity. Let $\bm M_D$ be the matrix implementing the within
transformation and $\bm P_D=\bm I_{M_n}-\bm M_D$. Since
$\bm M_D\bm v=\bm M_D\bm v^\circ$,
\begin{equation}
 \widehat{\bm\beta}_{\mathrm{FE}}-\bm\beta
 =\left(\frac{\bm X^{\circ\top}\bm M_D\bm X^\circ}{M_n}\right)^{-1}
 \left\{\frac1{M_n}\sum_{1\le i<j\le n}\bm s_{ij}^{\circ}
       -\frac{\bm X^{\circ\top}\bm P_D\bm u^\circ}{M_n}\right\}.
 \label{eq:fe_exact_score}
\end{equation}
The first term in braces is the score average studied earlier.
The second is due to estimating the node effects and is
$O_p(n^{-1})$. It is negligible at the $\sqrt n$ rate.
Condition~\eqref{eq:fe_fast_exogeneity} makes it
$o_p(n^{-1})$ and also negligible, giving the fast Gaussian result in part (ii). Detailed derivation is available in Appendix \ref{app:fe_projection}.

When $\bm\gamma_i^\circ=\bm0$ without this conditional
exogeneity restriction, the demeaning term can remain in the
$n$-rate limit, in addition to the node interaction and dyad
shocks. Thus Theorem~\ref{thm:fast_nongaussian} cannot simply
be applied to the population-demeaned score.

\paragraph{Example 2:}
{\itshape Let $Z_i$ and $e_{ij}$ take values $-1$ and $1$ with equal
probability, independently across nodes and dyads and of each
other. Set $x_{ij}=Z_iZ_j+e_{ij}$ and $u_{ij}=Z_iZ_j-e_{ij}$.
Both variables are already population demeaned, and
$x_{ij}u_{ij}=(Z_iZ_j)^2-e_{ij}^2=0$ for every dyad.
Thus, a limit calculation based only on the population-demeaned
score would give zero. However, the fixed-effect estimator satisfies
\(n(\widehat\beta_{\mathrm{FE}}-\beta)\Rightarrow1-G^2\),
where \(G\sim N(0,1)\).
The difference comes from the sample demeaning term in
\eqref{eq:fe_exact_score}, which remains nonnegligible at the
$n$ rate even though the population-demeaned score is zero.
}

Theorem~\ref{thm:fe_fast} in Appendix~\ref{app:fe_fast}
derives the full limit under additional finite-dimensional
representations of the score interaction and the conditional
means of the regressors and disturbances.

Subsampling captures this additional contribution by refitting
the effects within every cube. Under the conditions of
Theorem~\ref{thm:fe_fast}, the cube estimates have the same
limit with $q^\nu$ replacing $n$, so the coverage argument
of Theorem~\ref{thm:conservative_coverage} applies.
With $q\to\infty$ and $q^\nu/n\to0$, the subsampling interval
is asymptotically exact whenever the absolute limiting
contrast satisfies that theorem's quantile condition; the
MJS interval is then asymptotically conservative.
In both Gaussian cases of Theorem~\ref{thm:fe_gaussian},
the MJS interval remains asymptotically exact.
The procedure uses the same critical-value calculation as
before and does not require estimating or separately correcting
the demeaning contribution.

\section{Monte Carlo Simulations}
\label{sec:simulations}

We study the proposed procedures in the linear model
$y_{ij}=\bm x_{ij}^{\top}\bm\beta+u_{ij}$, with
$\bm\beta=(1,\ldots,1)^{\top}$. There is one node at each site
of a $\sqrt n\times\sqrt n$ unit lattice and one observation
for each unordered pair, giving $M_n=n(n-1)/2$ dyads. The regressors include a constant; the remaining regressors
are independent dyad-specific draws taking values $-1$ and $1$
with equal probability. In the nondegenerate design, the
disturbance includes
$\bm x_{ij}^{\top}(\bm A_i+\bm A_j)/\sqrt K$,
where $\bm A_i$ contains spatially dependent node effects.
Consequently,
$E[\bm s_{ij}\mid\mathcal Z]=(\bm A_i+\bm A_j)/\sqrt K$.
We also include a centered dyad shock whose conditional
variance depends on the two endpoint nodes. Appendix~\ref{app:simulation_results} gives the precise construction.

We consider three designs. Additive node contributions to the
score give the nondegenerate Gaussian case,
$\bm\Omega_\gamma>0$. Removing these contributions gives the
degenerate Gaussian case, $\bm\gamma_i=\bm\xi_{ij}=\bm0$.
Replacing them with interactions between centered node effects
gives $\bm\gamma_i=\bm0$ but $\bm\xi_{ij}\ne\bm0$, with a
non-Gaussian limit. The parameter $\rho$ controls spatial
dependence between nodes, while $\omega$ controls interaction
strength; $\omega=0$ recovers the degenerate Gaussian case.

We vary $n$, $K$, $\rho$, and $\omega$ one at a time, with
baseline values $n=36$, $K=10$, $\rho=0.25$, and $\omega=1$.
Each design uses 10,000 replications to test $\beta_K=1$ at the
5\% level. We compare White, conventional two-way clustering
(TW), conventional dyadic clustering (Dyadic), DN-Dyadic (DN), JK-DN-Dyadic (JK-DN),
and MJS. 

\paragraph{Nondegenerate Gaussian results.}
Figure~\ref{fig:sim_n}(a) varies the number of nodes. White's
rejection frequency increases from 0.261 at $n=16$ to 0.605
at $n=100$. Although there are $O(n^2)$ dyads, the leading
score depends on $n$ node contributions: the estimator's
variance is $O(n^{-1})$, whereas the own-dyad terms estimated
by White are $O(n^{-2})$. More observations therefore do not
compensate for ignoring dependence. Accounting for dyadic and spatial dependence generally improves
size control. At $n=100$, TW and Dyadic reject with frequencies
$0.240$ and $0.168$, compared with $0.146$ for DN. JK-DN further
reduces the rejection frequency to $0.085$, substantially closer
to the nominal $0.05$ level. Despite this improvement, its
rejection frequency increases with $n$ over the reported grid.

In this nondegenerate design, MJS closely follows JK-DN, with
both rejecting at frequency $0.085$ when $n=100$. This similarity
is consistent with the theory: in the $\sqrt n$ Gaussian case,
the subsampling term is asymptotically smaller than the jackknife
term, so MJS approaches the Gaussian jackknife interval.

Panel (b) helps explain these comparisons by reporting the
ratio of the average estimated variance to the Monte Carlo
variance. A ratio of one indicates accurate variance estimation
on average, whereas values below one indicate downward bias. At $n=100$, the variance ratio is only 0.07 of the Monte Carlo variance for
White, 0.51 for Dyadic, and 0.58 for DN. Shared-node clustering
omits substantial uncertainty, and DN recovers only part of it
in these samples. The jackknife variance ratio falls from 1.666
at $n=16$ to 1.011 at $n=100$. Its average variance estimate
becomes more accurate, but this alone does not ensure accurate
Gaussian rejection probabilities in finite samples.

Figure~\ref{fig:sim_K_rho}(a) examines the robustness of size
control across specifications with different numbers of
regressors and strengths of node dependence. Increasing $K$
adds regressors while the $1/\sqrt K$ normalization keeps
disturbance variance comparable and reduces the node
contribution to the tested score. JK-DN rejection frequencies decrease from $0.097$ at $K=2$
to $0.051$ at $K=20$; the corresponding MJS frequencies are
$0.097$ and $0.049$. Both procedures deliver better size
control than the alternatives throughout this exercise,
with rejection frequencies closer to the nominal level
in specifications with more regressors and a weaker node
contribution.

Panel (b) varies spatial dependence. At $\rho=0$, distinct
node effects are independent, but shared-node dependence
remains. White ignores this dependence, while TW accounts
for a shared node only when it occupies the same recorded
endpoint position in both dyads. Dyadic accounts for all
pairs sharing a node. Accordingly, rejection frequencies
decrease from $0.281$ for White to $0.117$ for TW and $0.081$
for Dyadic. DN performs similarly to Dyadic, rejecting with
frequency $0.084$.

As spatial dependence increases, dependence between dyads
with no common node becomes important. At $\rho=0.5$, Dyadic,
which omits these covariance contributions, rejects with
frequency $0.345$. DN reduces this to $0.296$, while JK-DN
and MJS both reject with frequency approximately $0.134$.
Spatial block deletion substantially reduces the distortion,
although strong dependence still produces overrejection at
$n=36$. The close JK-DN and MJS curves are consistent with
the Gaussian jackknife interval usually determining MJS
in this design.

\begin{figure}[!tbp]
\centering
\captionsetup[subfigure]{font=normalsize}
\begin{subfigure}[t]{0.49\textwidth}
\centering
\begin{tikzpicture}
\begin{axis}[
 width=0.94\linewidth,
 height=7.2cm,
 font=\normalsize,
 tick label style={font=\normalsize},
 label style={font=\normalsize},
 scale only axis=false,
 xmin=1,
 xmax=6,
 ymin=0,
 ymax=0.75,
 ytick={0,0.15,0.30,0.45,0.60,0.75},
 xtick={1,2,3,4,5,6},
 xticklabels={16,25,36,49,64,100},
 xlabel={$n$},
 ylabel={Rejection frequency},
 tick align=inside,
 line width=0.8pt,
 axis on top,
 scaled ticks=false,
 y tick label style={/pgf/number format/fixed},
 enlarge x limits=0.025,
 legend to name=sim_n_legend,
 legend columns=4,
 legend style={draw=none,font=\normalsize,/tikz/every even column/.append style={column sep=10pt}}
]
\addplot[black,line width=1.3pt,densely dotted,mark=none] coordinates {(1,0.2612) (2,0.3385) (3,0.4162) (4,0.4721) (5,0.5241) (6,0.6053)};
\addlegendentry{White}
\addplot[black,line width=1.3pt,densely dotted,mark=square,mark size=3pt, mark options={solid}] coordinates {(1,0.2525) (2,0.2292) (3,0.235) (4,0.2238) (5,0.2339) (6,0.2404)};
\addlegendentry{TW}
\addplot[black,line width=1.3pt,dashed,mark=none] coordinates {(1,0.2473) (2,0.1911) (3,0.183) (4,0.1685) (5,0.1655) (6,0.1675)};
\addlegendentry{Dyadic}
\addplot[black,line width=1.3pt,dashed,mark=o,mark size=3pt, mark options={solid}] coordinates {(1,0.2483) (2,0.1901) (3,0.1804) (4,0.161) (5,0.1528) (6,0.1457)};
\addlegendentry{DN}
\addplot[black,line width=1.3pt,solid,mark=none,line width=1.5pt] coordinates {(1,0.048) (2,0.0635) (3,0.0754) (4,0.0812) (5,0.0815) (6,0.0852)};
\addlegendentry{JK-DN}
\addplot[black,line width=1.3pt,solid,mark=diamond*,mark size=2.6pt] coordinates {(1,0.0475) (2,0.061) (3,0.0732) (4,0.081) (5,0.0814) (6,0.0852)};
\addlegendentry{MJS}
\addplot[black,line width=1.3pt,dashdotted,line width=0.9pt] coordinates {(1,0.05) (6,0.05)};
\addlegendentry{Reference}
\end{axis}
\end{tikzpicture}
\caption{Rejection frequency}
\end{subfigure}\hfill
\begin{subfigure}[t]{0.49\textwidth}
\centering
\begin{tikzpicture}
\begin{axis}[
 width=0.94\linewidth,
 height=7.2cm,
 font=\normalsize,
 tick label style={font=\normalsize},
 label style={font=\normalsize},
 scale only axis=false,
 xmin=1,
 xmax=6,
 ymin=0,
 ymax=1.8,
 ytick={0,0.5,1,1.5},
 xtick={1,2,3,4,5,6},
 xticklabels={16,25,36,49,64,100},
 xlabel={$n$},
 ylabel={Variance ratio},
 tick align=inside,
 line width=0.8pt,
 axis on top,
 scaled ticks=false,
 y tick label style={/pgf/number format/fixed},
 enlarge x limits=0.025
]
\addplot[black,line width=1.3pt,densely dotted,mark=none] coordinates {(1,0.3345909727) (2,0.2483722937) (3,0.1777919812) (4,0.1406955545) (5,0.1071399724) (6,0.07025738352)};
\addplot[black,line width=1.3pt,densely dotted,mark=square,mark size=3pt, mark options={solid}] coordinates {(1,0.4167552459) (2,0.4274321001) (3,0.3985125443) (4,0.4020638738) (5,0.3796813351) (6,0.3634018074)};
\addplot[black,line width=1.3pt,dashed,mark=none] coordinates {(1,0.4554121457) (2,0.5171497583) (3,0.509274431) (4,0.5331165547) (5,0.5165417085) (6,0.5102124084)};
\addplot[black,line width=1.3pt,dashed,mark=o,mark size=3pt, mark options={solid}] coordinates {(1,0.4536132359) (2,0.5189849629) (3,0.5228295671) (4,0.56006589) (5,0.5564413521) (6,0.5803393142)};
\addplot[black,line width=1.3pt,solid,mark=none,line width=1.5pt] coordinates {(1,1.665528441) (2,1.488636689) (3,1.25494266) (4,1.163097126) (5,1.048927937) (6,1.010626161)};
\addplot[black,line width=1.3pt,dashdotted,line width=0.9pt] coordinates {(1,1) (6,1)};
\end{axis}
\end{tikzpicture}
\caption{Variance estimation}
\end{subfigure}
\par\medskip
\pgfplotslegendfromname{sim_n_legend}
\caption{Nondegenerate Gaussian design: varying the number of nodes. $K=10$ and $\rho=0.25$. The variance ratio is the mean estimated variance divided by the Monte Carlo variance of $\widehat\beta_K$. Reference lines indicate 0.05 in panel (a) and one in panel (b). MJS is shown only in panel (a), since it does not define a separate variance estimator.}
\label{fig:sim_n}
\end{figure}

\begin{figure}[!tbp]
\centering
\captionsetup[subfigure]{font=normalsize}
\begin{subfigure}[t]{0.49\textwidth}
\centering
\begin{tikzpicture}
\begin{axis}[
 width=0.94\linewidth,
 height=7.2cm,
 font=\normalsize,
 tick label style={font=\normalsize},
 label style={font=\normalsize},
 scale only axis=false,
 xmin=1,
 xmax=10,
 ymin=0,
 ymax=0.75,
 ytick={0,0.15,0.30,0.45,0.60,0.75},
 xtick={1,2,3,4,5,6,7,8,9,10},
 xticklabels={2,4,6,8,10,12,14,16,18,20},
 xlabel={$K$},
 ylabel={Rejection frequency},
 tick align=inside,
 line width=0.8pt,
 axis on top,
 scaled ticks=false,
 y tick label style={/pgf/number format/fixed},
 enlarge x limits=0.025,
 legend to name=sim_K_rho_legend,
 legend columns=4,
 legend style={draw=none,font=\normalsize,/tikz/every even column/.append style={column sep=10pt}}
]
\addplot[black,line width=1.3pt,densely dotted,mark=none] coordinates {(1,0.6917) (2,0.5715) (3,0.5088) (4,0.4574) (5,0.4162) (6,0.3647) (7,0.3462) (8,0.325) (9,0.2968) (10,0.2754)};
\addlegendentry{White}
\addplot[black,line width=1.3pt,densely dotted,mark=square,mark size=3pt, mark options={solid}] coordinates {(1,0.2745) (2,0.2439) (3,0.2408) (4,0.2401) (5,0.235) (6,0.2109) (7,0.2121) (8,0.2111) (9,0.2021) (10,0.1884)};
\addlegendentry{TW}
\addplot[black,line width=1.3pt,dashed,mark=none] coordinates {(1,0.194) (2,0.1741) (3,0.1766) (4,0.1792) (5,0.183) (6,0.1677) (7,0.1749) (8,0.1719) (9,0.1684) (10,0.1591)};
\addlegendentry{Dyadic}
\addplot[black,line width=1.3pt,dashed,mark=o,mark size=3pt, mark options={solid}] coordinates {(1,0.1858) (2,0.1686) (3,0.1724) (4,0.175) (5,0.1804) (6,0.1646) (7,0.1723) (8,0.1704) (9,0.1679) (10,0.1586)};
\addlegendentry{DN}
\addplot[black,line width=1.3pt,solid,mark=none,line width=1.5pt] coordinates {(1,0.0971) (2,0.0849) (3,0.0794) (4,0.0771) (5,0.0754) (6,0.0644) (7,0.0652) (8,0.0614) (9,0.0507) (10,0.0505)};
\addlegendentry{JK-DN}
\addplot[black,line width=1.3pt,solid,mark=diamond*,mark size=2.6pt] coordinates {(1,0.0971) (2,0.0848) (3,0.0789) (4,0.0759) (5,0.0732) (6,0.0609) (7,0.0574) (8,0.0513) (9,0.0377) (10,0.0492)};
\addlegendentry{MJS}
\addplot[black,line width=1.3pt,dashdotted,line width=0.9pt] coordinates {(1,0.05) (10,0.05)};
\addlegendentry{Reference}
\end{axis}
\end{tikzpicture}
\caption{Number of regressors}
\end{subfigure}\hfill
\begin{subfigure}[t]{0.49\textwidth}
\centering
\begin{tikzpicture}
\begin{axis}[
 width=0.94\linewidth,
 height=7.2cm,
 font=\normalsize,
 tick label style={font=\normalsize},
 label style={font=\normalsize},
 scale only axis=false,
 xmin=1,
 xmax=6,
 ymin=0,
 ymax=0.75,
 ytick={0,0.15,0.30,0.45,0.60,0.75},
 xtick={1,2,3,4,5,6},
 xticklabels={0,0.1,0.2,0.3,0.4,0.5},
 xlabel={$\rho$},
 ylabel={Rejection frequency},
 tick align=inside,
 line width=0.8pt,
 axis on top,
 scaled ticks=false,
 y tick label style={/pgf/number format/fixed},
 enlarge x limits=0.025
]
\addplot[black,line width=1.3pt,densely dotted,mark=none] coordinates {(1,0.2805) (2,0.3331) (3,0.3808) (4,0.4355) (5,0.5008) (6,0.5767)};
\addplot[black,line width=1.3pt,densely dotted,mark=square,mark size=3pt, mark options={solid}] coordinates {(1,0.1165) (2,0.1574) (3,0.1979) (4,0.2512) (5,0.3124) (6,0.3986)};
\addplot[black,line width=1.3pt,dashed,mark=none] coordinates {(1,0.0808) (2,0.113) (3,0.1498) (4,0.1996) (5,0.2597) (6,0.3452)};
\addplot[black,line width=1.3pt,dashed,mark=o,mark size=3pt, mark options={solid}] coordinates {(1,0.0841) (2,0.1146) (3,0.147) (4,0.1854) (5,0.2286) (6,0.2963)};
\addplot[black,line width=1.3pt,solid,mark=none,line width=1.5pt] coordinates {(1,0.0448) (2,0.0534) (3,0.063) (4,0.0793) (5,0.1002) (6,0.1344)};
\addplot[black,line width=1.3pt,solid,mark=diamond*,mark size=2.6pt] coordinates {(1,0.0407) (2,0.05) (3,0.0602) (4,0.0772) (5,0.0984) (6,0.1337)};
\addplot[black,line width=1.3pt,dashdotted,line width=0.9pt] coordinates {(1,0.05) (6,0.05)};
\end{axis}
\end{tikzpicture}
\caption{Spatial dependence}
\end{subfigure}
\par\medskip
\pgfplotslegendfromname{sim_K_rho_legend}
\caption{Nondegenerate Gaussian design: varying $K$ and $\rho$. Panel (a) holds $n=36$ and $\rho=0.25$ fixed; panel (b) holds $n=36$ and $K=10$ fixed. The reference line is 0.05. JK-DN and MJS are plotted separately, although their rejection frequencies are often close.}
\label{fig:sim_K_rho}
\end{figure}

\paragraph{Degenerate non-Gaussian results.}
Table~\ref{tab:sim_nongaussian}, Panel A, considers a vanishing
node projection with a nonzero interaction. At $n=100$, Dyadic
and DN reject with frequencies 0.070 and 0.073, compared with
0.048 for JK-DN and 0.026 for MJS. The favorable jackknife
result in this design does not justify Gaussian critical
values for general non-Gaussian limits. MJS provides the
subsampling protection of Theorem~\ref{thm:conservative_coverage},
but is conservative here as predicted by the theory: its rejection frequency ranges from
0.012 to 0.026 over the sample-size grid.

MJS contains the Gaussian jackknife interval in every replication,
so it cannot reject more often than JK-DN. This difference is
small in the nondegenerate figures and larger in the degenerate
design. It reflects the cost of retaining both intervals without
first choosing a limiting distribution. Panel B varies
interaction strength. White is close to the nominal level at
$\omega=0$, rejecting with frequency 0.047, but reaches 0.077
at $\omega=2$. MJS remains conservative throughout; its coverage
protection does not imply exact finite-sample coverage.

\begin{table}[!tbp]
\centering
\caption{Degenerate non-Gaussian design: rejection frequencies}
\label{tab:sim_nongaussian}
\begingroup
\footnotesize
\setlength{\tabcolsep}{7pt}
\renewcommand{\arraystretch}{1.05}
\begin{tabular}{lrrrrrr}
\toprule
Setting & White & TW & Dyadic & DN & JK-DN & MJS \\
\midrule
\multicolumn{7}{l}{\textit{Panel A: Varying $n$ ($\omega=1$)}}\\
16 & 0.078 & 0.170 & 0.236 & 0.239 & 0.014 & 0.013 \\
25 & 0.068 & 0.120 & 0.159 & 0.161 & 0.022 & 0.012 \\
36 & 0.065 & 0.094 & 0.114 & 0.122 & 0.032 & 0.012 \\
49 & 0.065 & 0.076 & 0.091 & 0.096 & 0.039 & 0.014 \\
64 & 0.069 & 0.077 & 0.084 & 0.086 & 0.049 & 0.015 \\
100 & 0.067 & 0.067 & 0.070 & 0.073 & 0.048 & 0.026 \\
\midrule
\multicolumn{7}{l}{\textit{Panel B: Varying $\omega$ ($n=36$)}}\\
0 & 0.047 & 0.091 & 0.118 & 0.127 & 0.034 & 0.009 \\
0.4 & 0.057 & 0.093 & 0.120 & 0.128 & 0.033 & 0.011 \\
0.8 & 0.065 & 0.098 & 0.119 & 0.127 & 0.034 & 0.011 \\
1.2 & 0.071 & 0.094 & 0.119 & 0.125 & 0.035 & 0.012 \\
1.6 & 0.073 & 0.092 & 0.116 & 0.123 & 0.032 & 0.012 \\
2 & 0.077 & 0.095 & 0.114 & 0.121 & 0.032 & 0.012 \\
\bottomrule
\end{tabular}
\endgroup
\begin{flushleft}\footnotesize
Notes: Rejection frequencies at the 5\% level, based on 10,000 replications. Unless varied, $n=36$, $K=10$, $\rho=0.25$, and $\omega=1$. DN denotes the node-score implementation described in Appendix~\ref{app:simulation_results}; JK-DN denotes the corrected spatial jackknife. At $\omega=0$ the design is Gaussian; at $\omega>0$ its limit contains a non-Gaussian interaction. \end{flushleft}
\end{table}

Appendix~\ref{app:simulation_results} reports the remaining designs,
including the degenerate Gaussian case, both Gaussian fixed-effect
designs with effects refitted in every deletion and subsample,
and comparisons with the uncorrected estimators. Overall, we recommend JK-DN as a practical choice based on its
finite-sample performance in the simulations. MJS offers broader
theoretical coverage protection, but often produces conservative
inference in the finite samples.\footnote{Additional unreported experiments,
including non-Gaussian designs with small $K$ or large $\omega$, also show favorable
performance for JK-DN.} This recommendation is based on finite-sample performance:
our theory establishes the validity of JK-DN in the Gaussian
cases, but not for general non-Gaussian limits. MJS remains the appropriate choice when coverage protection
across the Gaussian and non-Gaussian regimes considered here
takes priority over finite-sample conservatism.

\section{Empirical Illustration: Earnings Composition of Commuting}
\label{sec:empirical_commuting}

Commuting connects workers to employment opportunities, but
commuting flows decline strongly with geographic distance
\citep{monte2018commuting}. Understanding this sensitivity is
important for assessing how closely different locations are
connected within a local labor market. Workers weigh earnings
against commuting costs when choosing jobs
\citep{rouwendal1999spatial}, and the resulting commuting
patterns need not be the same across earnings groups.
For example, access to
jobs matters for lower-paid displaced workers \citep{andersson2018job}.
Motivated by this literature, we examine whether commuting
flows associated with lower- and higher-paid jobs differ
in their sensitivity to distance, after accounting for
location characteristics and other controls. We also examine
whether county affiliation explains additional differences
in their earnings composition, and whether the evidence for
these associations remains significant after allowing spatial
node dependence.

We use the U.S. Census Bureau's public LEHD Origin-Destination Employment
Statistics (LODES) for 2019. The sample is divided into $n=256$ five-kilometer cells. We retain all
$M_n=32{,}640$ unordered cell pairs, including zero flows. There are
565,239 published primary-job links between distinct cells: 80,859 in
the lower earnings category, 145,916 in the middle category, and 338,464
in the higher category.\footnote{We use North Carolina's JT01 primary-job
OD file and geographic crosswalk, retaining links with both endpoints
inside the square centered at $35.9^\circ$N, $78.9^\circ$W. Census block
centroids determine cell membership; distances are measured between
cell centers. Bilateral flows sum both directions, excluding within-cell
links. LODES records jobs rather than daily journeys and applies
disclosure-avoidance procedures.}

Let $F^g_{ij}$ denote the number of published links from residences in
cell $i$ to workplaces in cell $j$ for earnings group $g\in\{L,H\}$, and
write $Q^g_{ij}=F^g_{ij}+F^g_{ji}$. The lower category consists of jobs
paying at most \$1,250 per month; the higher category consists of jobs
paying more than \$3,333 per month. These are monthly job-earnings
categories, not hourly wages or household incomes. Our main outcome is the earnings contrast,
\begin{equation}
 Y_{ij}=\log(1+Q^L_{ij})-\log(1+Q^H_{ij}), \qquad i<j,
 \label{eq:commuting_contrast}
\end{equation}
where a larger $Y_{ij}$ indicates more commuters in the lower earnings
group relative to the higher earnings group between locations
$i$ and $j$, as measured by the ratio
$(1+Q^L_{ij})/(1+Q^H_{ij})$.
We estimate
\begin{equation}
 Y_{ij}=\alpha_i+\alpha_j+\beta\log d_{ij}
       +\gamma C_{ij}+\delta I_{ij}+\eta B_{ij}+u_{ij},
 \label{eq:commuting_controlled}
\end{equation}
where $d_{ij}$ is geographic distance in kilometers. The fixed
effects absorb additive location characteristics that influence
$Y_{ij}$. Our main parameters of interest are the distance coefficient
$\beta$ and the same-county coefficient $\gamma$.
The coefficient $\beta$ measures the difference between the
two earnings groups in their sensitivity to distance.
When commuting flows decline with distance for both groups,
$\beta>0$ indicates a steeper decline for higher-paid jobs
on the $\log(1+Q)$ scale.
The coefficient $\gamma$ measures how the earnings contrast
differs between pairs of locations within the same county
and pairs in different counties, holding location effects
and the other controls fixed.

We include three dyad-level controls. First, $C_{ij}$ indicates that the
two cells $i$ and $j$ belong to the same county, allowing commuting composition to
differ across county boundaries conditional on distance. Each cell is
assigned to the county containing the largest number of its Census block
centroids. Second, industry matching $I_{ij}$ measures the compatibility between
resident workers' industries and workplace industries at the other
endpoint. Let $r_{is}$ and $w_{is}$ be the shares of resident and workplace
primary jobs in industry $s$ in 2018. Define
\begin{equation}
 I_{ij}=\frac12\sum_{s=1}^{20}(r_{is}w_{js}+r_{js}w_{is}),
 \qquad
 B_{ij}=|b_i-b_j|,\quad
 b_i=\log(1+W_i)-\log(1+R_i),
 \label{eq:commuting_controls}
\end{equation}
where $W_i$ and $R_i$ are 2018 workplace and resident primary-job totals.
The third control, $B_{ij}$, allows the contrast to differ between cells
with similar or different balances of jobs and resident workers. These
two controls use the 2018 Workplace Area Characteristics (WAC) and Residence Area Characteristics (RAC) files, so they are not constructed
from the 2019 bilateral outcome. Every cell has positive resident and
workplace totals. 

Nearby cells may share employment centers, housing markets, or transport
connections. Such common influences can correlate dyads even when they
have no endpoint in common. We compare White, TW, Dyadic, DN-Dyadic, JK-DN-Dyadic, and the
95\% MJS interval.\footnote{The MJS interval takes the
larger of the corrected jackknife normal half-width and the subsampling
term. We use all 100 contained $7\times7$ cubes, which identify all four
coefficients in every subsample. Fixed geographic regressors, possible spatial
nonstationarity, and the common-target requirement for subsampling remain
substantive qualifications: the formal assumptions do not establish
coverage for this application.}

\begin{table}[!t]
 \centering\footnotesize
 \caption{Controlled earnings contrast: standard errors and 95\% intervals}
 \label{tab:commuting_controls}
 \begingroup
\setlength{\tabcolsep}{5pt}
\renewcommand{\arraystretch}{1.20}
\begin{tabular}{@{}lcccc@{}}
\toprule
 & \shortstack{Log\\distance} & \shortstack{Same\\county} & \shortstack{Industry\\matching} & \shortstack{Jobs--residents\\balance gap} \\
 \midrule
Estimate & $0.166$ & $-0.090$ & $-1.412$ & $-0.016$ \\
White & $0.007\;(0.153,\,0.179)$ & $0.011\;(-0.111,\,-0.069)$ & $0.294\;(-1.987,\,-0.836)$ & $0.004\;(-0.023,\,-0.009)$ \\
TW & $0.015\;(0.136,\,0.195)$ & $0.025\;(-0.139,\,-0.041)$ & $0.623\;(-2.633,\,-0.190)$ & $0.011\;(-0.038,\,0.005)$ \\
Dyadic & $0.018\;(0.130,\,0.201)$ & $0.030\;(-0.148,\,-0.032)$ & $0.734\;(-2.850,\,0.027)$ & $0.013\;(-0.041,\,0.008)$ \\
DN & $0.019\;(0.129,\,0.203)$ & $0.032\;(-0.153,\,-0.027)$ & $0.754\;(-2.889,\,0.066)$ & $0.014\;(-0.043,\,0.010)$ \\
JK-DN & $0.026\;(0.114,\,0.217)$ & $0.051\;(-0.191,\,0.010)$ & $1.120\;(-3.606,\,0.783)$ & $0.022\;(-0.059,\,0.026)$ \\
\addlinespace[2pt]
\textbf{MJS} & \textemdash\ $(0.096,\,0.236)$ & \textemdash\ $(-0.191,\,0.010)$ & \textemdash\ $(-4.104,\,1.281)$ & \textemdash\ $(-0.066,\,0.033)$ \\
\bottomrule
\end{tabular}
\endgroup

 \begin{flushleft}\footnotesize
 Notes: The dependent variable is the earnings contrast, $Y_{ij}$.
 All columns use 256 cell effects and 32,640 unordered dyads.
 Entries report SE (95\% interval). DN and JK-DN abbreviate DN-Dyadic
 and JK-DN-Dyadic. MJS denotes the max jackknife-subsampling interval;
- indicates that it has no separate standard error.
 TW clusters the two recorded endpoint positions separately under a fixed
 random labeling unrelated to geography.
 \end{flushleft}
\end{table}

Table~\ref{tab:commuting_controls} reports $\widehat\beta=0.166$.
The distance coefficients for the lower and higher earnings
groups are $-0.421$ and $-0.587$, respectively. Both transformed
commuting counts decrease with distance, but the lower earnings
count decreases less steeply. Thus, distance is associated
with the earnings composition of commuting as well as its
overall volume. 

The difference remains statistically significant under MJS,
whose interval is $[0.096,0.236]$. Although this interval is
about twice as wide as the conventional dyadic interval,
its positive lower endpoint preserves the direction of the
comparison. Spatial dependence therefore matters for the
precision of the estimated difference without eliminating
the evidence for it. The subsampling component also matters
in this application: it extends the interval beyond the
corrected jackknife interval.

The same-county result leads to a different conclusion.
The estimate, $\widehat\gamma=-0.090$, suggests that commuting
within a county has a lower fitted lower-to-higher earnings
contrast than commuting across counties, conditional on
distance and the other controls. Conventional dyadic inference
would identify this additional association as significant.
The MJS interval, $[-0.191,0.010]$, includes zero, however,
so the evidence for a separate county association is sensitive
to the treatment of dependence. Figure~\ref{fig:commuting_inference}
illustrates this distinction: the distance comparison remains
significant, whereas the county comparison is not resolved
by MJS. The latter does not establish that county affiliation
is irrelevant.

The industry-matching and jobs-residents balance coefficients
already lose significance when moving from White to conventional
dyadic inference. Additional spatial adjustment makes their
estimates less precise. These results concern differences in
the earnings composition of commuting; they do not establish
that industrial compatibility or the distribution of jobs and
residents is unimportant for total commuting.

\begin{figure}[!htbp]
 \centering
 \includegraphics[width=\textwidth]{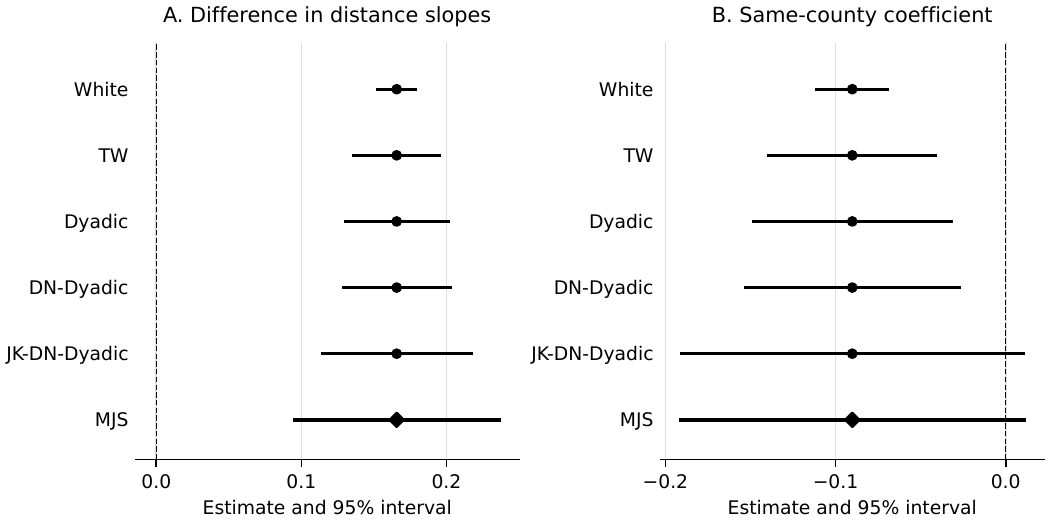}
 \caption{Inference for two coefficients in the same controlled regression.
 Panel A reports the difference in distance sensitivity between the
 lower and higher earnings groups; Panel B
 reports the same-county coefficient. Horizontal lines are pointwise
 95\% intervals. MJS denotes the max jackknife-subsampling interval,
 reported for both coefficients using all 100 contained $7\times7$ cubes.
 Source: Census LODES; authors' calculations.}
 \label{fig:commuting_inference}
\end{figure}

\begin{table}[!t]
 \centering\footnotesize
 \caption{Distance coefficients: standard errors and 95\% intervals across specifications}
 \label{tab:commuting_specifications}
 \begingroup
\setlength{\tabcolsep}{5pt}
\renewcommand{\arraystretch}{1.20}
\begin{tabular}{@{}lrccc@{}}
\toprule
Specification & Estimate & Dyadic & JK-DN & MJS \\
\midrule
Total flow, no controls & $-0.943$ & $0.055\;(-1.050,\,-0.836)$ & $0.280\;(-1.492,\,-0.394)$ & \textemdash\ $(-1.492,\,-0.394)$ \\
Contrast, no controls & $0.198$ & $0.015\;(0.169,\,0.228)$ & $0.018\;(0.163,\,0.234)$ & \textemdash\ $(0.114,\,0.283)$ \\
Contrast, controls & $0.166$ & $0.018\;(0.130,\,0.201)$ & $0.026\;(0.114,\,0.217)$ & \textemdash\ $(0.096,\,0.236)$ \\
Controls, shift $a=0.5$ & $0.186$ & $0.022\;(0.143,\,0.229)$ & $0.034\;(0.119,\,0.253)$ & \textemdash\ $(0.104,\,0.268)$ \\
Controls, shift $a=2$ & $0.145$ & $0.016\;(0.114,\,0.176)$ & $0.033\;(0.081,\,0.209)$ & \textemdash\ $(0.081,\,0.209)$ \\
\bottomrule
\end{tabular}
\endgroup

 \begin{flushleft}\footnotesize
 Notes: Every row uses the same 32,640 dyads and cell effects. The first
 row uses $\log(1+Q^{\mathrm{all}}_{ij})$; subsequent rows use
 $\log(a+Q^L_{ij})-\log(a+Q^H_{ij})$, with $a=1$ unless specified.
 Controls are those in equation~\eqref{eq:commuting_controls} and the
 same-county indicator. Entries report SE (95\% interval);
 JK-DN abbreviates JK-DN-Dyadic. MJS reports an em dash followed by its
 interval. All MJS intervals use $7\times7$ cubes. 
 \end{flushleft}
\end{table}

Table~\ref{tab:commuting_specifications} first reports the
total-flow regression, whose negative distance coefficient
indicates that overall commuting declines with distance.
The remaining specifications use the earnings contrast
$Y_{ij}$ as the outcome. The results show that the positive
 $\widehat{\beta}$ is not confined to the baseline specification. Adding the three controls reduces the estimate from $0.198$
to $0.166$, but the difference remains significant under MJS.
Changing the constant added before taking logarithms from one
to $0.5$ or two also preserves the positive MJS interval.
The estimated magnitude varies across these transformations,
so the direction is more stable than the size of the difference
in the reported specifications.

The interpretation remains tied to the transformed outcome.
Lower- and higher-paid commuting counts are zero for $72.6\%$
and $58.9\%$ of dyads, respectively. The transformation therefore
reflects both whether commuting is observed and how many
workers commute. Moreover, a smaller count can have a smaller
change in $\log(1+Q)$ even when its proportional decline is
identical to that of a larger count. Such dependence on the
scale of the outcome is an important feature of log-like
transformations with zeros \citep{chen2024logs}.

Overall, we find that commuting flows associated with higher-paid jobs
are more sensitive to distance, with transformed counts
declining more steeply than those for lower-paid jobs.
This difference remains statistically significant under MJS
across the reported specifications. The additional same-county
association, however, is no longer significant under MJS.

\section{Conclusion}
\label{sec:conclusion}

Spatial dependence between distinct nodes creates uncertainty that conventional dyadic clustering omits. We develop a spatial dyadic variance estimator and a corrected block jackknife for the two Gaussian cases, and subsampling inference for the degenerate case in which the limit can be non-Gaussian. The MJS interval combines these procedures and provides pointwise coverage protection across the cases studied. The fixed-effect analysis also shows why estimating node effects must be accounted for at the faster convergence rate. Simulations support JK-DN as a practical choice for finite-sample inference, while MJS offers broader theoretical protection with potentially greater conservativeness. In the commuting application, differences in sensitivity to distance across earnings groups remain significant under spatial inference, whereas the additional same-county association does not.

\clearpage

\title{Supplementary Appendix for ``Robust Inference for Dyadic Data with Spatially Dependent Nodes''}
\author{%
{\large \textbf{Ulrich Hounyo}}\thanks{Department of Economics, University at Albany - State University of New York, Albany, NY 12222, United States. E-mail: khounyo@albany.edu.}
\and
{\large \textbf{Jiahao Lin}}\thanks{School of Digital Economy and Management, Fuyao University of Science and Technology, Fuzhou, 350109, China.
E-mail: jhlin@fyust.edu.cn.}
\and
{\large \textbf{Xiaojun Song}}\thanks{Department of Business Statistics and Econometrics, Peking University, Beijing, China. E-mail: sxj@gsm.pku.edu.cn.}
}
\date{\today}

\setcounter{footnote}{0}
\emptythanks
\maketitle

 \noindent This supplement discusses node-score variance estimation and the duplication correction (Appendix A), provides proofs of the main theorems (Appendix B) and the results for models with node fixed effects (Appendix C), and collects technical lemmas and their proofs (Appendix D). It also provides implementation details (Appendix E) and presents the simulation designs and additional results (Appendix F).
\appendix

\section{Node-score variance estimation and the duplication correction}
\label{app:node_score_crve}

An alternative to the endpoint-kernel estimator in Section~\ref{subsec:dn_dyadic_crve} first aggregates scores by node:
\begin{equation}
 \widehat{\bm V}_{\mathrm{DN},0}
 =(\bm X^{\top}\bm X)^{-1}
 \left\{
 \sum_{r=1}^n\sum_{s=1}^n
 K^{\mathrm H}_{rs}
 \widehat{\bm G}_r\widehat{\bm G}_s^{\top}
 \right\}
 (\bm X^{\top}\bm X)^{-1},\qquad \widehat{\bm{G}}_r=\sum_{\substack{j=1\\j\ne r}}^n\widehat{\bm{s}}_{rj}
 \label{eq: twoterm}
\end{equation}
Here $\widehat{\bm{s}}_{rj}=\widehat{\bm{s}}_{jr}$ when $j<r$. Expansion gives
\[
 \sum_{1\le i<j\le n}\sum_{1\le p<q\le n}
 \{K^{\mathrm H}_{ip}+K^{\mathrm H}_{iq}+K^{\mathrm H}_{jp}+K^{\mathrm H}_{jq}\}
 \widehat{\bm{s}}_{ij}\widehat{\bm{s}}_{pq}^{\top}.
\]
The node construction sums the four endpoint weights. DN-Dyadic
uses their maximum, since $K$ is nonincreasing:
\[
 K\{\Delta((i,j),(p,q))/L\}
 =\max\{K^{\mathrm H}_{ip},K^{\mathrm H}_{iq},
         K^{\mathrm H}_{jp},K^{\mathrm H}_{jq}\}.
\]
The two expressions
differ when multiple endpoint pairings are close. In particular,
the node construction counts a dyad through both endpoints. This
distinction is negligible at the slow rate but first order when
only residual dyad shocks remain.

The estimator $\widehat{\bm V}_{\mathrm{DN},0}$ is positive
semidefinite because the node-weight matrix
$[K^{\mathrm H}_{rs}]_{r,s=1}^n$ is positive semidefinite.
It is consistent in the nondegenerate Gaussian case but
asymptotically overestimates the variance by a factor of two
in the degenerate Gaussian case. This parallels the distinction
between two-term and three-term CRVEs under two-way clustering:
omitting the subtraction retains contributions counted twice
\citep{mackinnon2021wild}.\footnote{Related variance inflation
from separately resampling rows and columns is discussed by
\citet{owen2007pigeonhole}.}

The correction $\widehat{\bm V}_I$ in \eqref{eq:jk_block}
removes the excess contribution asymptotically. Its weights
$K^\cup_{ij}-1$ come from the jackknife blocks and equal one for all
but $o(n^2)$ dyads and are uniformly bounded. Consequently,
$\widehat{\bm V}_{\mathrm{DN},0}-\widehat{\bm V}_I$ has the
same probability limit as $\widehat{\bm V}_{\mathrm{DN}}$
after scaling by $n$ or $n^2$, as appropriate.

\begin{proposition}[Duplication correction and asymptotic equivalence]
\label{prop:double_counting}
Under the assumptions of Theorem~\ref{thm:hac_jk}, in case (i),
\[
 n\widehat{\bm V}_{\mathrm{DN},0}\to^P\bm V_\gamma,
 \qquad
 n\widehat{\bm V}^{\mathrm{JK}}_{\mathrm{DN},0}\to^P\bm V_\gamma,
\]
whereas in case (ii),
\[
 n^2\widehat{\bm V}_{\mathrm{DN},0}\to^P2\bm V_\zeta,
 \qquad
 n^2\widehat{\bm V}^{\mathrm{JK}}_{\mathrm{DN},0}\to^P2\bm V_\zeta.
\]
The corrected node-score estimator satisfies
\begin{equation}
 \begin{aligned}
 n\bigl\{\widehat{\bm V}_{\mathrm{DN},0}-\widehat{\bm V}_I
             -\widehat{\bm V}_{\mathrm{DN}}\bigr\}
 &\to^P\bm0, &&\text{in case (i)},\\
 n^2\bigl\{\widehat{\bm V}_{\mathrm{DN},0}-\widehat{\bm V}_I
             -\widehat{\bm V}_{\mathrm{DN}}\bigr\}
 &\to^P\bm0, &&\text{in case (ii)}.
 \end{aligned}
 \label{eq:node_score_equivalence}
\end{equation}
The same equivalence holds if $\widehat{\bm V}_I$ is replaced by
the usual White estimator
\[
 (\bm X^\top\bm X)^{-1}
 \left\{\sum_{1\le i<j\le n}
 \widehat{\bm s}_{ij}\widehat{\bm s}_{ij}^\top\right\}
 (\bm X^\top\bm X)^{-1}.
\]
\end{proposition}

The factor two in case (ii) is already visible for independent centered
scalar dyadic scores with variance $\sigma^2$:
\[
 E\left[
 \sum_{i=1}^n
 \left(\sum_{\substack{j=1\\j\ne i}}^n s_{ij}\right)^2
 \right]
 =2M_n\sigma^2,
 \qquad
 \operatorname{Var}\left(
 \sum_{1\le i<j\le n}s_{ij}\right)
 =M_n\sigma^2.
\]
Each dyad appears in two node sums, whereas its own variance
enters the variance of the full score only once. 

The equivalence in \eqref{eq:node_score_equivalence} concerns
the two Gaussian cases. It does not assert finite-sample equality
or extend the Gaussian critical values to non-Gaussian limits.

\begin{proof}
\textit{Case (i).}\quad
The node-score and block-score calculations in
Appendices~\ref{app:proof_dn} and~\ref{app:proof_jk} give
$n\widehat{\bm V}_{\mathrm{DN},0}\to^P\bm V_\gamma$ and
$n\widehat{\bm V}^{\mathrm{JK}}_{\mathrm{DN},0}\to^P\bm V_\gamma$.
Moreover, boundedness of $K^\cup_{ij}$ and the moment bounds imply
\[
 \|\widehat{\bm V}_I\|
 \le C\|(\bm X^\top\bm X)^{-1}\|^2
       \sum_{1\le i<j\le n}\|\widehat{\bm s}_{ij}\|^2
 =O_p(n^{-2}).
\]
Here the squared inverse norm is $O_p(n^{-4})$ and the score
sum is $O_p(n^2)$. The same bound holds for the White estimator.
Both corrections are therefore $o_p(n^{-1})$.
Together with Theorem~\ref{thm:hac_jk}, this proves
\eqref{eq:node_score_equivalence} in case (i).

\textit{Case (ii).}\quad
Now $\bm s_{ij}=\bm\zeta_{ij}$, and write
$\bm G_r=\sum_{\substack{j=1\\j\ne r}}^n\bm\zeta_{rj}$.
Expanding the node-score quadratic form, the coefficient on
$\bm\zeta_{ij}\bm\zeta_{ij}^\top$ is
$2+2K^{\mathrm H}_{ij}$. Conditional on all node variables,
the scores of distinct dyads are independent and centered.
When squaring the sum over distinct dyads, only identical or
reversed dyad pairs remain. The endpoint weights are bounded
by four, and at most $Cn^3L^\nu$ dyad pairs receive nonzero
weight. Thus
\[
 E\left\|
 \sum_{r=1}^n\sum_{s=1}^nK^{\mathrm H}_{rs}\bm G_r\bm G_s^\top
 -\sum_{1\le i<j\le n}(2+2K^{\mathrm H}_{ij})
                    \bm\zeta_{ij}\bm\zeta_{ij}^\top
 \right\|^2
 \le Cn^3L^\nu=o(n^4).
\]
Only $CnL^\nu=o(n^2)$ dyads have $d_{ij}\le L$.
The uniform moments therefore give
\[
 \sum_{1\le i<j\le n}K^{\mathrm H}_{ij}\bm\zeta_{ij}\bm\zeta_{ij}^\top
 =O_p(nL^\nu)=o_p(n^2).
\]
Combining these bounds yields
\[
 \sum_{r=1}^n\sum_{s=1}^nK^{\mathrm H}_{rs}\bm G_r\bm G_s^\top
 =2\sum_{1\le i<j\le n}\bm\zeta_{ij}\bm\zeta_{ij}^\top+o_p(n^2).
\]
The score-replacement bounds in Step 1 of
Appendix~\ref{app:proof_dn} also apply to the sum of the four
endpoint weights: these weights are uniformly bounded and
retain the same order of dyad pairs. Replacing population
scores by estimated scores consequently changes this matrix
by $o_p(n^2)$. Lemma~\ref{lem:dyadic_lln} then gives
$n^2\widehat{\bm V}_{\mathrm{DN},0}\to^P2\bm V_\zeta$.

Next, Lemma~\ref{lem:dyad_deletion_weights} gives
$K^\cup_{ij}=2$ except for $o(n^2)$ dyads, with bounded weights
everywhere. By the uniform moments,
\[
 \sum_{1\le i<j\le n}(K^\cup_{ij}-2)
       \bm\zeta_{ij}\bm\zeta_{ij}^\top=o_p(n^2).
\]
For any uniformly bounded deterministic weights, replacing
$\bm\zeta_{ij}\bm\zeta_{ij}^\top$ by
$\widehat{\bm s}_{ij}\widehat{\bm s}_{ij}^\top$ in such a
diagonal sum has norm at most
$O_p(n^2)(\|\widehat{\bm\beta}-\bm\beta\|
+\|\widehat{\bm\beta}-\bm\beta\|^2)=O_p(n)=o_p(n^2)$.
It follows that
\[
 \widehat{\bm V}_I
 =(\bm X^\top\bm X)^{-1}
 \left\{\sum_{1\le i<j\le n}\bm\zeta_{ij}\bm\zeta_{ij}^\top\right\}
 (\bm X^\top\bm X)^{-1}+o_p(n^{-2}),
 \qquad n^2\widehat{\bm V}_I\to^P\bm V_\zeta.
\]
The White estimator has the same expansion.
Subtracting this limit from the node-score limit and applying
Theorem~\ref{thm:hac_jk} proves
\eqref{eq:node_score_equivalence} in case (ii), with either
correction. Finally, the exact identity
$\widehat{\bm V}^{\mathrm{JK}}_{\mathrm{DN},0}
=\widehat{\bm V}^{\mathrm{JK}}_{\mathrm{DN}}+\widehat{\bm V}_I$
and the same theorem give
$n^2\widehat{\bm V}^{\mathrm{JK}}_{\mathrm{DN},0}\to^P2\bm V_\zeta$.
\end{proof}

\clearpage
\section{Proofs of Main Theorems}
\label{app:main_proofs}

Throughout the proofs, $\mathcal Z=\sigma\{Z(\bm s):\bm s\in\mathbb Z^\nu\}$
denotes the full node field. A set $T\subset\{1,\ldots,n\}$ is a
deterministic selection of nodes, with $m=|T|$.
When dyads are indexed by $e=(i,j)$ and $f=(p,q)$, we write
$\bm s_e=\bm s_{ij}$ and $\Delta(e,f)=\Delta((i,j),(p,q))$,
and likewise for other dyadic variables. In a sum
$\sum_{\substack{j=1\\j\ne i}}^n$, only $j$ is summed and $i$ is fixed
within that sum. When both indices vary, both summation ranges
are shown explicitly.
The norm is Euclidean for vectors and the operator norm for matrices;
$\|\cdot\|_F$ denotes the Frobenius norm. The constant $C$ may change
between inequalities but does not depend on $n,L,i$, or the selected
set $T$. Since the parameter dimension is fixed, coordinatewise moment
bounds also give the corresponding vector and matrix bounds.

\subsection{Proof of Theorem~\ref{thm:clt}}
\begin{proof}
\textit{Step 1: Separate the score components.}\quad
Lemma~\ref{lem:uniform_deleted_block_lln} gives $\bm{Q}_n\to^P \bm{Q}$.
Since $\bm{\mu}=\bm0$, the exact score decomposition is
\begin{equation}
 \widehat{\bm{\beta}}-\bm{\beta}
 =\bm{Q}_n^{-1}\left\{\frac2n\sum_{i=1}^n\bm{\gamma}_i
                 +\frac1{M_n}\sum_{1\le i<j\le n}(\bm{\xi}_{ij}+\bm{\zeta}_{ij})\right\}.
 \label{eq:proof_score_decomp_2}
\end{equation}
The moment bounds in Lemma~\ref{lem:spatial_moments} imply
\[
 E\Big\|\sum_{1\le i<j\le n}\bm{\xi}_{ij}\Big\|^2\le Cn^2,\qquad
 E\Big\|\sum_{1\le i<j\le n}\bm{\zeta}_{ij}\Big\|^2
       =\sum_{1\le i<j\le n}E\|\bm{\zeta}_{ij}\|^2\le Cn^2.
\]
The second identity follows from conditional independence and
centering given the full node field $\mathcal Z$.
Hence the average remainder is $O_p(n^{-1})$.
In case (i), Lemma~\ref{lem:spatial_clt} gives
\[
 n^{-1/2}\sum_{i=1}^n\bm{\gamma}_i\Rightarrow N(\bm0,\bm{\Omega}_\gamma).
\]
Substitution into \eqref{eq:proof_score_decomp_2} proves
the limit in case (i).

\textit{Step 2: Treat residual dyad variation.}\quad
In case (ii), $\bm{\gamma}_i=\bm{\xi}_{ij}=\bm0$ almost surely, so
$\bm{s}_{ij}=\bm{\zeta}_{ij}$. For a fixed contrast $\bm{a}$, write
$\zeta_{ij}^a=\bm{a}^{\top}\bm{\zeta}_{ij}$. Conditional on $\mathcal Z$,
these variables are independent and centered, with normalized variance
\begin{align}
 \Var\left(\frac n{M_n}\sum_{1\le i<j\le n}\zeta_{ij}^a
                       \,\middle|\,\mathcal Z\right)
 &=\frac{n^2}{M_n^2}\sum_{1\le i<j\le n}
                       E[(\zeta_{ij}^a)^2\mid Z_i,Z_j]
 \to^P2E_\perp[(\zeta_{12}^a)^2]
       =\bm{a}^{\top}\bm{\Omega}_\zeta \bm{a}.
 \label{eq:cond_var_zeta_limit}
\end{align}
The convergence is Lemma~\ref{lem:dyadic_lln} applied
to the conditional-variance kernel, using $n^2/M_n\to2$.
For a sufficiently small $\epsilon>0$, the uniform moments give
\[
 E\left[\frac{n^2}{M_n^2}\sum_{1\le i<j\le n}
 E\left\{(\zeta_{ij}^a)^2
 \ind\{|\zeta_{ij}^a|>cM_n/n\}\mid\mathcal Z\right\}\right]
 \le C\frac{n^2}{M_n}\left(\frac n{M_n}\right)^\epsilon
 =O(n^{-\epsilon})
\]
for each fixed $c>0$. This proves conditional Lindeberg by
Markov's inequality. Conditional Lindeberg-Feller and
Cram\'er-Wold give
\[
 \frac n{M_n}\sum_{1\le i<j\le n}\bm{\zeta}_{ij}
 \Rightarrow N(\bm0,\bm{\Omega}_\zeta).
\]
Together with $\bm{Q}_n\to^P \bm{Q}$, this proves the limit in case (ii).
\end{proof}

\subsection{Proof of Theorem~\ref{thm:hac_jk} for the DN-Dyadic CRVE}
\label{app:proof_dn}
\begin{proof}
Throughout this proof, $e,f\in\mathcal D_n$ denote dyads.
For example, $\bm s_e=\bm s_{ij}$ when $e=(i,j)$.
Define the population-score version of the meat matrix by
\[
 \bm\Sigma_{\mathrm{DN}}
 =\sum_{e\in\mathcal D_n}\sum_{f\in\mathcal D_n}
 K\{\Delta(e,f)/L\}\bm s_e\bm s_f^\top.
\]
We first show that
\[
 \|\widehat{\bm\Sigma}_{\mathrm{DN}}-\bm\Sigma_{\mathrm{DN}}\|
 =
 \begin{cases}
 o_p(n^3), & \text{in case (i)},\\
 o_p(n^2), & \text{in case (ii)}.
 \end{cases}
\]
Thus estimated and population scores give the same limit after
normalization by $n^3$ in case (i) and by $n^2$ in case (ii).
We then establish these limits for $\bm\Sigma_{\mathrm{DN}}$.

\textit{Step 1: Control the score-estimation error.}\quad
Write
\(
 \bm H_{ij}=\bm x_{ij}\bm x_{ij}^\top\) and \(
 \bm\delta_n=\widehat{\bm\beta}-\bm\beta.
\)
Then $\widehat{\bm s}_e=\bm s_e-\bm H_e\bm\delta_n$, so
\[
\begin{aligned}
 \widehat{\bm\Sigma}_{\mathrm{DN}}-\bm\Sigma_{\mathrm{DN}}
 &=-\sum_{e\in\mathcal D_n}\sum_{f\in\mathcal D_n}
 K\{\Delta(e,f)/L\}
 \bigl\{\bm H_e\bm\delta_n\bm s_f^\top
       +\bm s_e\bm\delta_n^\top\bm H_f\bigr\}\\
 &\quad+\sum_{e\in\mathcal D_n}\sum_{f\in\mathcal D_n}
 K\{\Delta(e,f)/L\}
 \bm H_e\bm\delta_n\bm\delta_n^\top\bm H_f.
\end{aligned}
\]
To bound the first line, let $H_{e,ab}$ denote entry $(a,b)$
of $\bm H_e$, and let $s_{f,c}$ denote coordinate $c$ of
$\bm s_f$. Since the dimension $K$ is fixed and the weights are symmetric,
\begin{equation}
\begin{aligned}
 \|\widehat{\bm\Sigma}_{\mathrm{DN}}-\bm\Sigma_{\mathrm{DN}}\|
 &\le C\|\bm\delta_n\|
 \max_{1\le a,b,c\le K}
 \left|\sum_{e\in\mathcal D_n}\sum_{f\in\mathcal D_n}
 K\{\Delta(e,f)/L\}H_{e,ab}s_{f,c}\right|\\
 &\quad+C\|\bm\delta_n\|^2
 \sum_{e\in\mathcal D_n}\sum_{f\in\mathcal D_n}
 K\{\Delta(e,f)/L\}\|\bm H_e\|\|\bm H_f\|.
\end{aligned}
\label{eq:dn_score_replacement}
\end{equation}

For each fixed dyad $f$, a nonzero weight requires one endpoint
of $e$ to lie within $L$ of an endpoint of $f$. There are at most
$CL^\nu$ choices for that endpoint and $n$ choices for the other.
Thus there are at most $CnL^\nu$ such dyads $e$, and at most
$Cn^3L^\nu$ weighted dyad pairs in total.
The uniform moment bounds and Markov's inequality give
\[
 \sum_{e\in\mathcal D_n}\sum_{f\in\mathcal D_n}
 K\{\Delta(e,f)/L\}\|\bm H_e\|\|\bm H_f\|
 =O_p(n^3L^\nu).
\]

In case (i), the same first-moment argument gives, for each
$a,b,c$,
\[
 \sum_{e\in\mathcal D_n}\sum_{f\in\mathcal D_n}
 K\{\Delta(e,f)/L\}H_{e,ab}s_{f,c}
 =O_p(n^3L^\nu).
\]
Since $\|\bm\delta_n\|=O_p(n^{-1/2})$, we have 
\(
 \|\widehat{\bm\Sigma}_{\mathrm{DN}}-\bm\Sigma_{\mathrm{DN}}\|
 =O_p(n^{5/2}L^\nu+n^2L^\nu)
 =o_p(n^3),
\)
where the last equality uses $L^{2\nu}/n\to0$.

In case (ii), $\bm s_e=\bm\zeta_e$ and
$\|\bm\delta_n\|=O_p(n^{-1})$. We need a sharper bound for the
first term in \eqref{eq:dn_score_replacement}.
For fixed coordinates $a,b,c$, write
\[
 \overline H_{e,ab}=E[H_{e,ab}\mid\mathcal Z],
 \qquad
 \widetilde H_{e,ab}=H_{e,ab}-\overline H_{e,ab}.
\]
Separating the conditional mean and the terms with $e=f$ gives
\[
\begin{aligned}
 \sum_{e\in\mathcal D_n}\sum_{f\in\mathcal D_n}
 K\{\Delta(e,f)/L\}H_{e,ab}\zeta_{f,c}
 &=
 \sum_{f\in\mathcal D_n}
 \left\{\sum_{e\in\mathcal D_n}
 K\{\Delta(e,f)/L\}\overline H_{e,ab}\right\}\zeta_{f,c}\\
 &\qquad+
 \sum_{e\in\mathcal D_n}
 \sum_{\substack{f\in\mathcal D_n\\f\ne e}}
 K\{\Delta(e,f)/L\}\widetilde H_{e,ab}\zeta_{f,c}
 +\sum_{e\in\mathcal D_n}\widetilde H_{e,ab}\zeta_{e,c}.
\end{aligned}
\]

For the first sum, its coefficients are measurable with respect
to $\mathcal Z$. Conditional independence and centering of
$\zeta_{f,c}$ therefore remove terms with different $f$ when
we square the sum. Applying Cauchy-Schwarz inequality for the inner sum, which
contains at most $CnL^\nu$ terms, gives
\begin{align*}
 E\left|
 \sum_{f\in\mathcal D_n}
 \left\{\sum_{e\in\mathcal D_n}
 K\{\Delta(e,f)/L\}\overline H_{e,ab}\right\}\zeta_{f,c}
 \right|^2
 &=
 \sum_{f\in\mathcal D_n}
 E\left[
 \left\{\sum_{e\in\mathcal D_n}
 K\{\Delta(e,f)/L\}\overline H_{e,ab}\right\}^2
 \zeta_{f,c}^2\right]\\
 &\quad\le
 CnL^\nu
 \sum_{f\in\mathcal D_n}
 \sum_{\substack{e\in\mathcal D_n\\\Delta(e,f)\le L}}
 E[\overline H_{e,ab}^2\zeta_{f,c}^2]
 \le Cn^4L^{2\nu}.
\end{align*}
Here conditional Jensen's inequality and the uniform moment
bounds control each expectation.

For the second sum, conditional independence across dyads and
conditional centering imply
\[
 E[\widetilde H_{e,ab}\zeta_{f,c}
    \widetilde H_{g,ab}\zeta_{h,c}\mid\mathcal Z]=0
\]
unless $(g,h)=(e,f)$ or $(g,h)=(f,e)$, with $e\ne f$ and
$g\ne h$. Because in every other configuration some dyad
appears only once and contributes a conditionally centered
factor. The number of remaining terms and the moment bounds give
\[
 E\left|
 \sum_{e\in\mathcal D_n}
 \sum_{\substack{f\in\mathcal D_n\\f\ne e}}
 K\{\Delta(e,f)/L\}\widetilde H_{e,ab}\zeta_{f,c}
 \right|^2
 \le Cn^3L^\nu.
\]
For the third sum, we do not require
$\widetilde H_{e,ab}\zeta_{e,c}$ to be centered.
Cauchy-Schwarz inequality gives directly
\(
 E\left|\sum_{e\in\mathcal D_n}
          \widetilde H_{e,ab}\zeta_{e,c}\right|^2
 \le M_n\sum_{e\in\mathcal D_n}
       E[\widetilde H_{e,ab}^2\zeta_{e,c}^2]
 \le Cn^4.
\)
Combining these three bounds yields
\[
 \max_{1\le a,b,c\le K}
 \left|\sum_{e\in\mathcal D_n}\sum_{f\in\mathcal D_n}
 K\{\Delta(e,f)/L\}H_{e,ab}\zeta_{f,c}\right|
 =O_p(n^2L^\nu).
\]
Substitution into \eqref{eq:dn_score_replacement} gives
\(
 \|\widehat{\bm\Sigma}_{\mathrm{DN}}-\bm\Sigma_{\mathrm{DN}}\|
 =O_p(nL^\nu)=o_p(n^2)
\) in case (ii).

\textit{Step 2: Establish the population-score limit in case (i).}
\quad
Define the remainder in each node score by
\(
 \bm R_r=\sum_{\substack{j=1\\j\ne r}}^n
              (\bm\xi_{rj}+\bm\zeta_{rj})
 \) and \(
 \bm S_\gamma=\sum_{i=1}^n\bm\gamma_i.
\)
Since $\bm s_{ij}=\bm\gamma_i+\bm\gamma_j+
\bm\xi_{ij}+\bm\zeta_{ij}$, the node score satisfies
\[
 \bm G_r=\sum_{\substack{j=1\\j\ne r}}^n\bm s_{rj}
        =(n-2)\bm\gamma_r+\bm S_\gamma+\bm R_r.
\]
Now we focus on the node-score
quadratic form $\frac1{n^3}\sum_{r,s=1}^nK^{\mathrm H}_{rs}\bm G_r\bm G_s^\top$. Lemma~\ref{lem:spatial_moments}(ii)-(iii) gives
$E\|\bm R_r\|^2\le Cn$ uniformly in $r$. Summing over $r$ yields
\(
 E\sum_{r=1}^n\|\bm R_r\|^2
 =\sum_{r=1}^nE\|\bm R_r\|^2\le Cn^2.
\)
Applying Lemma~\ref{lem:spatial_moments} (i) with $T=\{1,\ldots,n\}$ also gives
\(
 E\|\bm S_\gamma\|^2
 =E\Big\|\sum_{i=1}^n\bm\gamma_i\Big\|^2\le Cn.
\)
Consequently,
\[
 \sum_{r=1}^n\|\bm S_\gamma+\bm R_r\|^2
 \le 2n\|\bm S_\gamma\|^2+
      2\sum_{r=1}^n\|\bm R_r\|^2
 =O_p(n^2).
\]
By Lemma~\ref{lem:spatial_weights}, the node-weight matrix
$[K^{\mathrm H}_{rs}]$ is positive semidefinite and has
operator norm at most $CL^\nu$. Hence
\[
 \left\|\sum_{r,s=1}^nK^{\mathrm H}_{rs}
 (\bm S_\gamma+\bm R_r)(\bm S_\gamma+\bm R_s)^\top\right\|
 \le CL^\nu\sum_{r=1}^n\|\bm S_\gamma+\bm R_r\|^2
 =O_p(n^2L^\nu).
\]
Meanwhile, \eqref{eq:spatial_hac_lln} in Lemma \ref{lem:spatial_covariance} implies
\(
 \frac1n\sum_{r,s=1}^nK^{\mathrm H}_{rs}
       \bm\gamma_r\bm\gamma_s^\top
 \to^P\bm\Omega_\gamma.
\)
The contribution of $(n-2)\bm\gamma_r$ to the node-score
quadratic form $\frac1{n^3}\sum_{r,s=1}^nK^{\mathrm H}_{rs}\bm G_r\bm G_s^\top$ satisfies
\[
\begin{aligned}
 &\left\|
 \sum_{r,s=1}^n K^{\mathrm H}_{rs}
       \{(n-2)\bm\gamma_r\}\{(n-2)\bm\gamma_s\}^{\top}
 \right\|=
 n(n-2)^2
 \left\|\frac1n\sum_{r,s=1}^n K^{\mathrm H}_{rs}
                 \bm\gamma_r\bm\gamma_s^\top\right\|
 =O_p(n^3),
\end{aligned}
\]
because the normalized matrix inside the norm converges in
probability to $\bm\Omega_\gamma$.
Together with the $O_p(n^2L^\nu)$ bound for the remainder,
this allows us to bound the cross terms by Cauchy-Schwarz.
Applying Cauchy-Schwarz inequality for the positive-semidefinite weight matrix
bounds each cross term by
\[
 O_p\!\left\{(n^3)^{1/2}(n^2L^\nu)^{1/2}\right\}
 =o_p(n^3).
\]
Combining the above results and expanding the node-score quadratic form therefore give
\[
 \frac1{n^3}\sum_{r,s=1}^nK^{\mathrm H}_{rs}\bm G_r\bm G_s^\top
 =
 \frac{(n-2)^2}{n^2}
 \frac1n\sum_{r,s=1}^nK^{\mathrm H}_{rs}
             \bm\gamma_r\bm\gamma_s^\top+o_p(1)
 \to^P\bm\Omega_\gamma.
\]

It remains to compare this node-score expression with
$\bm\Sigma_{\mathrm{DN}}$. For dyads $(i,j)$ and $(p,q)$,
the former assigns weight
\[
 K^{\mathrm H}_{ip}+K^{\mathrm H}_{iq}
 +K^{\mathrm H}_{jp}+K^{\mathrm H}_{jq},
\]
whereas the latter assigns the maximum of these four weights.
They agree whenever at most one weight is nonzero.
A difference therefore requires at least two endpoint distances
to be at most $L$.

There are six ways to choose two of the four endpoint pairings.
For each choice, there are at most $n^2$ choices for two node
locations and at most $CL^\nu$ choices for each of the remaining
two locations. This bounds the number of dyad pairs with
different weights by $Cn^2L^{2\nu}$.
Since the weights are bounded, uniform moments and Markov's
inequality give
\[
 \left\|\bm\Sigma_{\mathrm{DN}}-
       \sum_{r,s=1}^nK^{\mathrm H}_{rs}\bm G_r\bm G_s^\top\right\|
 =O_p(n^2L^{2\nu})=o_p(n^3).
\]
Together with Step 1, this proves
$n^{-3}\widehat{\bm\Sigma}_{\mathrm{DN}}\to^P\bm\Omega_\gamma$.
Since $M_n/n^2\to1/2$ and $\bm Q_n\to^P\bm Q$ by Lemma \ref{lem:uniform_deleted_block_lln}, one can deduce that
\[
 n\widehat{\bm V}_{\mathrm{DN}}
 =\frac{n^4}{M_n^2}\bm Q_n^{-1}
   \frac{\widehat{\bm\Sigma}_{\mathrm{DN}}}{n^3}
   \bm Q_n^{-1}
 \to^P4\bm Q^{-1}\bm\Omega_\gamma\bm Q^{-1}.
\]

\textit{Step 3: Establish the population-score limit in case (ii).}
\quad
Here $\bm s_e=\bm\zeta_e$. Separate equal and distinct dyads:
\[
 \bm\Sigma_{\mathrm{DN}}
 =\sum_{e\in\mathcal D_n}\bm\zeta_e\bm\zeta_e^\top
 +\sum_{e\in\mathcal D_n}
  \sum_{\substack{f\in\mathcal D_n\\f\ne e}}
  K\{\Delta(e,f)/L\}\bm\zeta_e\bm\zeta_f^\top,
\]
where the first sum has weight one because $K(0)=1$.

For fixed coordinates $a,b$, squaring the second sum produces
expectations of
$\zeta_{e,a}\zeta_{f,b}\zeta_{g,a}\zeta_{h,b}$.
Conditional independence and centering make these expectations
zero unless $(g,h)=(e,f)$ or $(g,h)=(f,e)$.
There are at most $Cn^3L^\nu$ weighted pairs $(e,f)$, and each
remaining expectation is uniformly bounded. Thus, in fixed
dimension,
\[
 E\left\|
 \sum_{e\in\mathcal D_n}
 \sum_{\substack{f\in\mathcal D_n\\f\ne e}}
 K\{\Delta(e,f)/L\}\bm\zeta_e\bm\zeta_f^\top
 \right\|^2
 \le Cn^3L^\nu.
\]
Markov's inequality and $L^\nu/n\to0$ imply that this sum is
$o_p(n^2)$.

By Lemma~\ref{lem:dyadic_lln},
\[
 \frac1{M_n}\sum_{1\le i<j\le n}
       \bm\zeta_{ij}\bm\zeta_{ij}^\top
 \to^P E_\perp[\bm\zeta_{12}\bm\zeta_{12}^\top].
\]
Since $n^2/M_n\to2$, Step 1 and the preceding bounds give
\(
 \frac{n^2}{M_n^2}\widehat{\bm\Sigma}_{\mathrm{DN}}
 \to^P
 2E_\perp[\bm\zeta_{12}\bm\zeta_{12}^\top]
 =\bm\Omega_\zeta.
\)
Therefore, we have
\[
 n^2\widehat{\bm V}_{\mathrm{DN}}
 =\bm Q_n^{-1}
   \frac{n^2}{M_n^2}\widehat{\bm\Sigma}_{\mathrm{DN}}
   \bm Q_n^{-1}
 \to^P\bm Q^{-1}\bm\Omega_\zeta\bm Q^{-1}.
\]
The studentized limits in both cases follow from
Theorem~\ref{thm:clt} and Slutsky's theorem.
\end{proof}

\subsection{Proof of Theorem~\ref{thm:hac_jk} for the JK-DN-Dyadic CRVE}
\label{app:proof_jk}
\begin{proof}
\textit{Step 1: Express each deleted-sample difference using scores.}
\quad
Let
\(
 \bm H_\ell
 =\sum_{(i,j)\in\mathcal A_\ell}
       \bm x_{ij}\bm x_{ij}^{\top}.
\)
On the event that the full and all deleted design matrices are
nonsingular, the normal equations give
\begin{equation}
 \widehat{\bm\beta}-\widetilde{\bm\beta}_{(-\ell)}
 =(\bm X^\top\bm X)^{-1}
 \left\{
 \sum_{(i,j)\in\mathcal A_\ell}\bm s_{ij}
 -\bm H_\ell(\widetilde{\bm\beta}_{(-\ell)}-\bm\beta)
 \right\}.
 \label{eq:beta_delete_exact}
\end{equation}
This event has probability tending to one by
Lemma~\ref{lem:uniform_deleted_block_lln}.
The first term uses population scores; the second accounts for
estimating the coefficient in the deleted sample.

Define the corrected population-score matrix by
\begin{equation}
\begin{aligned}
 \bm\Psi^{\mathrm{JK}}
 &=\sum_{\ell=1}^n\frac1{m_\ell}
 \left(\sum_{(i,j)\in\mathcal A_\ell}\bm s_{ij}\right)
 \left(\sum_{(i,j)\in\mathcal A_\ell}\bm s_{ij}\right)^\top-\sum_{1\le i<j\le n}
       (K^\cup_{ij}-1)\bm s_{ij}\bm s_{ij}^\top.
\end{aligned}
 \label{eq:jk_population_meat}
\end{equation}
Steps 2 and 3 establish its limit in the two cases.
Step 4 then shows that using estimated coefficients and scores
changes the variance estimator by $o_p(n^{-1})$ in case (i)
and $o_p(n^{-2})$ in case (ii).

\textit{Step 2: Population-score limit in case (i).}\quad
We first show that each deleted-score sum can be approximated
by $n\sum_{i\in B_\ell}\bm\gamma_i$ in the block quadratic form.

Each node in $B_\ell$ belongs to $n-1$ deleted dyads, whereas
each node outside $B_\ell$ belongs to $m_\ell$ deleted dyads.
Therefore
\begin{align}\label{eq:jkeq0}
 \sum_{(i,j)\in\mathcal A_\ell}(\bm\gamma_i+\bm\gamma_j)
 &=(n-1)\sum_{i\in B_\ell}\bm\gamma_i
   +m_\ell\sum_{i\notin B_\ell}\bm\gamma_i=n\sum_{i\in B_\ell}\bm\gamma_i
   -(m_\ell+1)\sum_{i\in B_\ell}\bm\gamma_i
   +m_\ell\sum_{i=1}^n\bm\gamma_i.
\end{align}
Lemma~\ref{lem:spatial_moments}(i) gives
\(
 E\left\|\sum_{i\in B_\ell}\bm\gamma_i\right\|^2\le Cm_\ell
 \) and \(
 E\left\|\sum_{i=1}^n\bm\gamma_i\right\|^2\le Cn.
\)
Using $\|\bm a+\bm b\|^2\le
2\|\bm a\|^2+2\|\bm b\|^2$ and $m_\ell\le CL^\nu$,
we obtain
\begin{align}\label{eq:jkeq1}
 &E\sum_{\ell=1}^n\frac1{m_\ell}
 \left\|
 -(m_\ell+1)\sum_{i\in B_\ell}\bm\gamma_i
 +m_\ell\sum_{i=1}^n\bm\gamma_i
 \right\|^2\le
 C\sum_{\ell=1}^n\{(m_\ell+1)^2+nm_\ell\}
 \le C(nL^{2\nu}+n^2L^\nu)
 =o(n^3).
\end{align}

The remaining score components are
$\bm\xi_{ij}+\bm\zeta_{ij}$. When a squared block sum is expanded,
a product involving dyads $(i,j)$ and $(p,q)$ appears only in
blocks that delete both dyads. Its total weight is
\begin{equation}
 K^\cup_{ij,pq}
 =\sum_{\ell=1}^n\frac1{m_\ell}
   \ind\{i\in B_\ell\text{ or }j\in B_\ell\}
   \ind\{p\in B_\ell\text{ or }q\in B_\ell\}.
 \label{eq: Kupijpq}
\end{equation}
The comma separates the two dyads. Each indicator requires
at least one endpoint of its dyad to belong to the block.
For the same dyad, the indicators coincide, so
\[
 K^\cup_{ij,ij}=K^\cup_{ij}.
\]
Thus the coefficient on each $\bm s_{ij}\bm s_{ij}^{\top}$
in the uncorrected score quadratic form is $K^\cup_{ij}$.
Subtracting $K^\cup_{ij}-1$ makes that coefficient exactly one,
as in \eqref{eq:jk_population_meat}.

For the interaction and residual terms, the uniform bound on
$K^\cup_{ij,pq}$ from Lemma~\ref{lem:dyad_deletion_weights} gives
\begin{align}\label{eq:jkeq2}
 &E\sum_{\ell=1}^n\frac1{m_\ell}
 \left\|\sum_{(i,j)\in\mathcal A_\ell}
       (\bm\xi_{ij}+\bm\zeta_{ij})\right\|^2\le
 C\sum_{1\le i<j\le n}\sum_{1\le p<q\le n}
 \left\|E\left[
   (\bm\xi_{ij}+\bm\zeta_{ij})
   (\bm\xi_{pq}+\bm\zeta_{pq})^\top
 \right]\right\|
 \le Cn^2,
\end{align}
where the last inequality is
Lemma~\ref{lem:spatial_moments}(ii)-(iii).
Combining the preceding bounds in \eqref{eq:jkeq0}-\eqref{eq:jkeq2} and applying Markov's inequality
therefore yields
\begin{align}\label{eq:jkeq3}
 \sum_{\ell=1}^n\frac1{m_\ell}
 \left\|
 \sum_{(i,j)\in\mathcal A_\ell}\bm s_{ij}
 -n\sum_{i\in B_\ell}\bm\gamma_i
 \right\|^2=o_p(n^3).
\end{align}

For the leading node component, expanding the block sums
instead requires the total weight of blocks containing both
nodes $r$ and $s$. Define
\begin{equation}
 K^\cap_{rs}
 =\sum_{\ell=1}^n\frac1{m_\ell}
       \ind\{r,s\in B_\ell\}.
 \label{eq:block_overlap_kernel}
\end{equation}
The superscript $\cap$ indicates that both nodes belong to
the same block. Rearranging the sums now gives
\begin{align}\label{eq:jkeq4}
 &\frac1{n^3}\sum_{\ell=1}^n\frac1{m_\ell}
 \left(n\sum_{i\in B_\ell}\bm\gamma_i\right)
 \left(n\sum_{i\in B_\ell}\bm\gamma_i\right)^\top=
 \frac1n\sum_{r,s=1}^nK^\cap_{rs}
                \bm\gamma_r\bm\gamma_s^\top
 \to^P\bm\Omega_\gamma
\end{align}
by Lemma~\ref{lem:spatial_covariance}.
Thus its unnormalized order is $O_p(n^3)$.
Cauchy-Schwarz bounds the cross terms between the leading
block sums and their errors by
$O_p(n^{3/2})o_p(n^{3/2})=o_p(n^3)$.
Consequently, by \eqref{eq:jkeq3}-\eqref{eq:jkeq4}, we have
\[
 \frac1{n^3}\sum_{\ell=1}^n\frac1{m_\ell}
 \left(\sum_{(i,j)\in\mathcal A_\ell}\bm s_{ij}\right)
 \left(\sum_{(i,j)\in\mathcal A_\ell}\bm s_{ij}\right)^\top
 \to^P\bm\Omega_\gamma.
\]
The same argument, taking traces in fixed dimension, gives
\(
 \sum_{\ell=1}^n\frac1{m_\ell}
 \left\|\sum_{(i,j)\in\mathcal A_\ell}\bm s_{ij}\right\|^2
 =O_p(n^3).
\)

The correction in \eqref{eq:jk_population_meat} has order
$O_p(n^2)$, because $K^\cup_{ij}$ is uniformly bounded and
$\sum_{i<j}E\|\bm s_{ij}\|^2\le Cn^2$.
Hence, in case (i), we have
\(
 \frac1{n^3}\bm\Psi^{\mathrm{JK}}
 \to^P\bm\Omega_\gamma.
\)

\textit{Step 3: Population-score limit in case (ii).}\quad
Here $\bm s_{ij}=\bm\zeta_{ij}$.
Expanding the block sums and separating equal and distinct
dyads gives
\begin{align}\label{eq:jkeq6}
 \sum_{\ell=1}^n\frac1{m_\ell}
 \left(\sum_{(i,j)\in\mathcal A_\ell}\bm\zeta_{ij}\right)
 \left(\sum_{(i,j)\in\mathcal A_\ell}\bm\zeta_{ij}\right)^\top=
 \sum_{1\le i<j\le n}
       K^\cup_{ij}\bm\zeta_{ij}\bm\zeta_{ij}^\top+
 \sum_{1\le i<j\le n}
 \sum_{\substack{1\le p<q\le n\\(p,q)\ne(i,j)}}
 K^\cup_{ij,pq}\bm\zeta_{ij}\bm\zeta_{pq}^\top.
\end{align}
The first sum is $O_p(n^2)$ by bounded weights and second moments.

For the second sum, expand the square of each matrix entry.
Conditional independence and centering imply that a product
has zero expectation whenever one dyad appears only once.
Thus, for each pair of distinct dyads, only the same pair
or its reverse contributes.
The fourth-moment bound and
Lemma~\ref{lem:dyad_deletion_weights} therefore imply
\[
\begin{aligned}
 &E\left\|
 \sum_{1\le i<j\le n}
 \sum_{\substack{1\le p<q\le n\\(p,q)\ne(i,j)}}
 K^\cup_{ij,pq}\bm\zeta_{ij}\bm\zeta_{pq}^\top
 \right\|^2\le
 C\sum_{1\le i<j\le n}\sum_{1\le p<q\le n}
 \ind\{\Delta((i,j),(p,q))\le2L\}
 \le Cn^3L^\nu.
\end{aligned}
\]
Since $L^\nu/n\to0$, the sum is $o_p(n^2)$.
The uncorrected block matrix therefore has order $O_p(n^2)$.
As it is positive semidefinite, taking its trace also gives
\(
 \sum_{\ell=1}^n\frac1{m_\ell}
 \left\|\sum_{(i,j)\in\mathcal A_\ell}\bm\zeta_{ij}\right\|^2
 =O_p(n^2).
\)

The correction subtracts $K^\cup_{ij}-1$ from each diagonal
weight $K^\cup_{ij}$, leaving weight one. Hence, 
\(
 \bm\Psi^{\mathrm{JK}}
 =\sum_{1\le i<j\le n}
       \bm\zeta_{ij}\bm\zeta_{ij}^\top+o_p(n^2).
\)
Lemma~\ref{lem:dyadic_lln} and $n^2/M_n\to2$ now give
\[
 \frac{n^2}{M_n^2}\bm\Psi^{\mathrm{JK}}
 \to^P
 2E_\perp[\bm\zeta_{12}\bm\zeta_{12}^\top]
 =\bm\Omega_\zeta.
\]

\textit{Step 4: Control coefficient and score estimation errors.}
\quad
We now return to the second term in
\eqref{eq:beta_delete_exact}.
The bound $E\|\bm H_\ell\|^2\le Cn^2m_\ell^2$ implies
\[
 E\sum_{\ell=1}^n\frac{\|\bm H_\ell\|^2}{m_\ell}
 \le Cn^2\sum_{\ell=1}^n m_\ell
 \le Cn^3L^\nu.
\]
Thus this sum is $O_p(n^3L^\nu)$ by Markov's inequality.
Using $\|(\bm X^\top\bm X)^{-1}\|=O_p(n^{-2})$ and
Lemma~\ref{lem:delete_block_beta_rate}, we obtain
\[
\begin{aligned}
 \sum_{\ell=1}^n\frac1{m_\ell}
 \left\|(\bm X^\top\bm X)^{-1}\bm H_\ell
       (\widetilde{\bm\beta}_{(-\ell)}-\bm\beta)\right\|^2
 &\le
 \|(\bm X^\top\bm X)^{-1}\|^2
 \max_{1\le\ell\le n}
       \|\widetilde{\bm\beta}_{(-\ell)}-\bm\beta\|^2
 \sum_{\ell=1}^n\frac{\|\bm H_\ell\|^2}{m_\ell}\\
 &\quad=
 \begin{cases}
 O_p(L^\nu n^{-2})=o_p(n^{-1}),&\text{case (i)},\\
 O_p(L^\nu n^{-3})=o_p(n^{-2}),&\text{case (ii)}.
 \end{cases}
\end{aligned}
\]
Steps 2 and 3 give, respectively,
\[
 \sum_{\ell=1}^n\frac1{m_\ell}
 \left\|(\bm X^\top\bm X)^{-1}
       \sum_{(i,j)\in\mathcal A_\ell}\bm s_{ij}\right\|^2
 =
 \begin{cases}
 O_p(n^{-1}),&\text{case (i)},\\
 O_p(n^{-2}),&\text{case (ii)}.
 \end{cases}
\]
Applying Cauchy-Schwarz inequality bounds the cross terms by the product of
the square roots of these two bounds. They are therefore
$o_p(n^{-1})$ in case (i) and $o_p(n^{-2})$ in case (ii).

It remains to replace estimated scores in the correction.
Let $\bm\delta_n=\widehat{\bm\beta}-\bm\beta$.
Since
$\widehat{\bm s}_{ij}
=\bm s_{ij}-\bm x_{ij}\bm x_{ij}^\top\bm\delta_n$,
bounded $K^\cup_{ij}$ and the moment assumptions give
\[
\begin{aligned}
 \left\|\sum_{1\le i<j\le n}(K^\cup_{ij}-1)
 \left(\widehat{\bm s}_{ij}\widehat{\bm s}_{ij}^\top
                 -\bm s_{ij}\bm s_{ij}^\top\right)\right\| &\le
 C\|\bm\delta_n\|
       \sum_{1\le i<j\le n}\|\bm x_{ij}\|^2\|\bm s_{ij}\|
 +C\|\bm\delta_n\|^2
       \sum_{1\le i<j\le n}\|\bm x_{ij}\|^4\\
 &\quad=O_p(n^2)
       \{\|\bm\delta_n\|+\|\bm\delta_n\|^2\}
 =o_p(n^2).
\end{aligned}
\]
Multiplying on both sides by $(\bm X^\top\bm X)^{-1}$
makes this error $o_p(n^{-2})$.

Combining these bounds proves
\[
 \left\|
 \widehat{\bm V}^{\mathrm{JK}}_{\mathrm{DN}}
 -(\bm X^\top\bm X)^{-1}\bm\Psi^{\mathrm{JK}}
       (\bm X^\top\bm X)^{-1}
 \right\|
 =
 \begin{cases}
 o_p(n^{-1}),&\text{case (i)},\\
 o_p(n^{-2}),&\text{case (ii)}.
 \end{cases}
\]
Finally, $\bm X^\top\bm X=M_n\bm Q_n$,
$\bm Q_n\to^P\bm Q$, and $M_n/n^2\to1/2$.
The limits established in Steps 2 and 3 therefore yield
\[
 \begin{aligned}
 n\widehat{\bm V}^{\mathrm{JK}}_{\mathrm{DN}}
 &\to^P4\bm Q^{-1}\bm\Omega_\gamma\bm Q^{-1}
   =\bm V_\gamma,
 &&\text{in case (i)},\\
 n^2\widehat{\bm V}^{\mathrm{JK}}_{\mathrm{DN}}
 &\to^P\bm Q^{-1}\bm\Omega_\zeta\bm Q^{-1}
   =\bm V_\zeta,
 &&\text{in case (ii)}.
 \end{aligned}
\]
The studentized limits follow from Theorem~\ref{thm:clt}
and Slutsky's theorem.
\end{proof}

\subsection{Proof of Theorem~\ref{thm:fast_nongaussian}}
\begin{proof}
Since $\bm\gamma_i=\bm0$ almost surely and $\bm\mu=\bm0$, the score
decomposition gives
\[
 n(\widehat{\bm\beta}-\bm\beta)
 =\bm Q_n^{-1}
 \left\{
 \frac n{M_n}\sum_{1\le i<j\le n}\bm\xi_{ij}
 +\frac n{M_n}\sum_{1\le i<j\le n}\bm\zeta_{ij}
 \right\}.
\]
We first derive the limit using finitely many node functions,
then control the omitted interaction terms and combine the
result with the residual limit.

\textit{Step 1: Find the limit for each finite approximation.}\quad
Fix $r$ and let
\[
 \bm F_n^{(r)}=n^{-1/2}\sum_{i=1}^n\bm f_i^{(r)}.
\]
The centering and moment conditions in
Assumption~\ref{ass:finite_interaction} allow us to apply
Lemma~\ref{lem:spatial_clt}, giving
$\bm F_n^{(r)}\Rightarrow\bm{\mathcal G}^{(r)}$.
For the sample second moments, H\"older's inequality gives
\[
 E|f_a(Z_i)f_b(Z_i)|^{2(\lambda+\delta)}\le C.
\]
The same coupling and covariance argument as in
Lemma~\ref{lem:spatial_moments}(i) therefore gives
\[
 \left|\Cov\{f_a(Z_i)f_b(Z_i),f_a(Z_j)f_b(Z_j)\}\right|
 \le C(1+d_{ij})^{-\mathfrak d\{1-1/(\lambda+\delta)\}}.
\]
Summing this bound over lattice locations yields
\[
 E\left|
 \frac1n\sum_{i=1}^n f_a(Z_i)f_b(Z_i)
 -E[f_a(Z_i)f_b(Z_i)]
 \right|^2\le\frac Cn.
\]
Consequently, for each fixed $r$,
$n^{-1}\sum_{i=1}^n\bm f_i^{(r)}(\bm f_i^{(r)})^{\top}
\to^P\bm C_f^{(r)}$.

For each coordinate $k$, symmetry of $\bm B_k^{(r)}$ implies
\[
 2\sum_{1\le i<j\le n}\xi_{ij,k}^{(r)}
 =\left(\sum_{i=1}^n\bm f_i^{(r)}\right)^{\top}
       \bm B_k^{(r)}\left(\sum_{i=1}^n\bm f_i^{(r)}\right)
   -\sum_{i=1}^n(\bm f_i^{(r)})^{\top}\bm B_k^{(r)}\bm f_i^{(r)}.
\]
The subtraction removes terms with identical endpoints.
Since $M_n=n(n-1)/2$, it follows that
\[
 \frac n{M_n}\sum_{1\le i<j\le n}\xi_{ij,k}^{(r)}
 =\frac n{n-1}
 \left\{
 (\bm F_n^{(r)})^{\top}\bm B_k^{(r)}\bm F_n^{(r)}
 -\tr\left(\bm B_k^{(r)}\frac1n\sum_{i=1}^n
               \bm f_i^{(r)}(\bm f_i^{(r)})^{\top}\right)
 \right\}.
\]
The continuous mapping theorem and Slutsky's theorem give,
jointly over $k=1,\ldots,K$,
\[
 \frac n{M_n}\sum_{1\le i<j\le n}\xi_{ij,k}^{(r)}
 \Rightarrow
 (\bm{\mathcal G}^{(r)})^{\top}\bm B_k^{(r)}\bm{\mathcal G}^{(r)}
 -\tr(\bm B_k^{(r)}\bm C_f^{(r)}).
\]

\textit{Step 2: Bound the omitted interaction terms.}\quad
The node-sum argument in Lemma~\ref{lem:spatial_moments}(i)
applies to each centered function $f_a(Z_i)$. The common moment
bound makes its constant independent of $a$, so
\[
 E\left|\frac1{\sqrt n}\sum_{i=1}^n f_a(Z_i)\right|^4\le C.
\]
Also, Jensen's and Cauchy-Schwarz inequalities give, uniformly
in $a,b$,
\[
 E\left|\frac1n\sum_{i=1}^n f_a(Z_i)f_b(Z_i)\right|^2
 \le \frac1n\sum_{i=1}^n E[f_a(Z_i)^2f_b(Z_i)^2]\le C.
\]
Thus the product of two normalized node sums and the sample
second moment each have a uniformly bounded second moment.
Applying the quadratic identity in Step 1 to the omitted
coefficients gives
\begin{align*}
 &\frac n{M_n}\sum_{1\le i<j\le n}
       (\xi_{ij,k}-\xi_{ij,k}^{(r)})\\
 &\quad=\frac n{n-1}
 \sum_{\substack{a,b\ge1\\\max(a,b)>r}}b_{ab,k}
 \left\{
 \left(\frac1{\sqrt n}\sum_{i=1}^n f_a(Z_i)\right)
 \left(\frac1{\sqrt n}\sum_{i=1}^n f_b(Z_i)\right)
 -\frac1n\sum_{i=1}^n f_a(Z_i)f_b(Z_i)
 \right\}.
\end{align*}
The series converges in mean square by absolute summability
and the preceding bounds. Minkowski's inequality therefore yields
\begin{equation}
 \sup_{n\ge2}E\left|
 \frac n{M_n}\sum_{1\le i<j\le n}
       (\xi_{ij,k}-\xi_{ij,k}^{(r)})
 \right|^2
 \le C\left(\sum_{\substack{a,b\ge1\\\max(a,b)>r}}
                 |b_{ab,k}|\right)^2\longrightarrow0.
 \label{eq:interaction_truncation}
\end{equation}
Here $r\to\infty$; the bound is uniform in $n$.

The Gaussian limit has $\sup_a E[\mathcal G_a^4]\le C$:
the covariance bound for centered node functions makes their
long-run variances uniformly bounded, and a centered Gaussian
variable has fourth moment equal to three times its variance
squared. Hence the same bound as in
\eqref{eq:interaction_truncation} applies to the omitted terms
$b_{ab,k}\{\mathcal G_a\mathcal G_b-E[f_a(Z_i)f_b(Z_i)]\}$.
This proves existence of the mean-square limit in
\eqref{eq:interaction_series_limit}. Together with Step 1,
\eqref{eq:interaction_truncation} allows us to first let
$n\to\infty$ for fixed $r$ and then let $r\to\infty$, obtaining
\[
 \frac n{M_n}\sum_{1\le i<j\le n}\bm\xi_{ij}
 \Rightarrow\bm{\mathcal Q}_B(\bm{\mathcal G}).
\]

\textit{Step 3: Combine the interaction and residual limits.}\quad
Conditional on the entire node field $\mathcal Z$, the
$\bm\zeta_{ij}$ are independent and centered.
The conditional-variance and Lindeberg calculations in the
proof of Theorem~\ref{thm:clt} give, for every $\bm t\in\mathbb R^K$,
\[
 E\left[
 \left.
 \exp\left(
 i\bm t^\top\frac n{M_n}
       \sum_{1\le i<j\le n}\bm\zeta_{ij}
 \right)
 \right|\mathcal Z
 \right]
 \to^P
 \exp\left(-\frac12\bm t^\top\bm\Omega_\zeta\bm t\right).
\]
We next explain why the residual limit is independent of
$\bm{\mathcal G}$, although the residuals need not be
independent of the node field in finite samples.
The difference from the preceding limit is bounded in absolute
value by two, so this convergence also holds in $L^1$.
For each fixed $r$, $\bm F_n^{(r)}$ is measurable with respect
to $\mathcal Z$. Conditioning on $\mathcal Z$ consequently gives,
for every $\bm u\in\mathbb R^r$,
\[
\begin{aligned}
 E\left[\exp\left(
 i\bm u^\top\bm F_n^{(r)}+
 i\bm t^\top\frac n{M_n}
       \sum_{1\le i<j\le n}\bm\zeta_{ij}
 \right)\right]
 &\quad=
 E\left[e^{i\bm u^\top\bm F_n^{(r)}}
 E\left[\left.
 \exp\left(i\bm t^\top\frac n{M_n}
       \sum_{1\le i<j\le n}\bm\zeta_{ij}\right)
 \right|\mathcal Z\right]\right]\\
 &\quad=
 E[e^{i\bm u^\top\bm F_n^{(r)}}]
 \exp\left(-\frac12\bm t^\top\bm\Omega_\zeta\bm t\right)+o(1)\\
 &\quad\longrightarrow
 \exp\left(-\frac12\bm u^\top\bm\Omega_f^{(r)}\bm u\right)
 \exp\left(-\frac12\bm t^\top\bm\Omega_\zeta\bm t\right).
\end{aligned}
\]
Thus the residual Gaussian limit is independent of every
finite subvector $\bm{\mathcal G}^{(r)}$, and hence of the
Gaussian sequence $\bm{\mathcal G}$. This also allows singular
covariance matrices. Combining this joint convergence with
Step 1 gives the joint limit of each finite interaction
approximation and the residual sum. The mean-square bounds
in Step 2 then allow $r\to\infty$, yielding
\[
 \frac n{M_n}\sum_{1\le i<j\le n}
       (\bm\xi_{ij}+\bm\zeta_{ij})
 \Rightarrow
 \bm{\mathcal Q}_B(\bm{\mathcal G})
 +\bm{\mathcal E}_\zeta.
\]
Finally, $\bm Q_n\to^P\bm Q$ and Slutsky's theorem prove
\eqref{eq:fast_nongaussian_limit}.

For a growing lattice-cube subsample, the same proof applies
with its number of nodes replacing $n$.
Stationarity gives the same marginal moments and expectations
defined by $E_\perp$. For each fixed $\bm h\in\mathbb Z^\nu$,
the fraction of nodes whose shifted location
$\bm s_i+\bm h$ remains in the cube tends to one.
Together with absolute covariance summability, this gives
the same $\bm\Omega_f^{(r)}$ for every fixed $r$.
The bound in \eqref{eq:interaction_truncation} is also uniform
over these cubes, because the node moment bounds depend only
on their number of nodes. Letting $r\to\infty$ therefore gives
the same full limiting distribution for the subsamples.
\end{proof}

\subsection{Proof of Theorem~\ref{thm:conservative_coverage}}
\begin{proof}
Recall that each subsample is a lattice cube of side length $q$
containing $q^\nu$ nodes. There are
$N_q\sim n$ such cubes.
We first prove that the subsampling interval has asymptotically
exact coverage under part (ii). Since the MJS interval
contains this interval, conservative coverage follows.
We then establish exact coverage in the two Gaussian cases.

\textit{Step 1: Estimate the distribution using the true center.}
\quad
Under part (ii), Theorem~\ref{thm:fast_nongaussian} gives
\[
 n\bm{a}^{\top}(\widehat{\bm{\beta}}-\bm{\beta})\Rightarrow W,
 \qquad
 q^\nu \bm{a}^{\top}(\widehat{\bm{\beta}}_{\ell,q}-\bm{\beta})\Rightarrow W.
\]
For a fixed $t$, define
\[
 J_{\ell,q}(t)
 =\ind\{q^\nu |\bm{a}^{\top}(\widehat{\bm{\beta}}_{\ell,q}-\bm{\beta})|\le t\},
 \qquad
 F_q(t)=E[J_{\ell,q}(t)].
\]
Stationarity and the common cube shape imply that $F_q(t)$
does not depend on $\ell$. The subsample limit therefore gives
\[
 F_q(t)\longrightarrow F_{|W|}(t):=P(|W|\le t)
\]
at every continuity point of $F_{|W|}$. 

We now show that \(
 \frac1{N_q}\sum_{\ell=1}^{N_q}J_{\ell,q}(t)
 \to^P F_{|W|}(t)
\)
at every continuity point of $F_{|W|}$. We bound its variance:
\[
 \Var\left\{\frac1{N_q}\sum_{\ell=1}^{N_q}J_{\ell,q}(t)\right\}
 =
 \frac1{N_q^2}\sum_{\ell=1}^{N_q}\sum_{k=1}^{N_q}
 \Cov\{J_{\ell,q}(t),J_{k,q}(t)\}.
\]
Fix a cube $I_{\ell,q}$. For each cube, call the lattice site
with the smallest coordinate in every direction its starting site.
Each starting site determines one cube. The lattice count therefore
implies that at most $Cq^\nu$ cubes have starting sites within
distance $2\sqrt\nu q$ of the starting site of $I_{\ell,q}$.
For these pairs, including overlapping cubes, use
\(
 |\Cov\{J_{\ell,q}(t),J_{k,q}(t)\}|\le1.
\)

Now consider two cubes whose starting sites are distance
$d>2\sqrt\nu q$ apart. Relative to its starting site, each node
has coordinate offsets in $\{0,\ldots,q-1\}$.
The offsets of two nodes therefore differ by at most $q-1$
in each direction, which yields a distance of at most $\sqrt{\nu}(q-1)$. By the triangle inequality, the distance
between any two nodes in two cubes is at least
\[
 d-\sqrt\nu(q-1)\ge d -\sqrt\nu q \ge d-d/2  = d/2.
\]
Thus the cubes are disjoint and separated by at least $d/2$.

Conditional on $\mathcal Z$, the two indicators depend on
disjoint sets of independent dyad shocks. Hence
\[
 \Cov\{J_{\ell,q}(t),J_{k,q}(t)\}
 =
 \Cov\!\left\{
 E[J_{\ell,q}(t)\mid\mathcal Z],
 E[J_{k,q}(t)\mid\mathcal Z]
 \right\}.
\]
Each $J_{k,q}(t)$ depends on the $q^\nu$ nodes in its cube.
 To control dependence through the nodes, apply
Berbee's coupling lemma (as in the proof to Lemma \ref{lem:spatial_moments}(i)) to the entire node vector in cube $k$.
It gives a copy $(Z_i^*)_{i\in I_{k,q}}$ with the same joint
distribution as $(Z_i)_{i\in I_{k,q}}$, independent of the node
variables in cube $\ell$, such that
\[
\begin{aligned}
 P\{Z_i^*\ne Z_i\text{ for some }i\in I_{k,q}\}
 &\le C|I_{\ell,q}|\,|I_{k,q}|(1+d/2)^{-\mathfrak d}\le Cq^{2\nu}(1+d)^{-\mathfrak d}.
\end{aligned}
\]
Here we used the $\beta$-mixing assumption, the separation
of at least $d/2$, and
$|I_{\ell,q}|=|I_{k,q}|=q^\nu$.
The coupling preserves dependence among nodes within cube $k$.

Construct the copy using auxiliary randomness independent of
all dyad shocks. Let $J_{k,q}^*(t)$ be the indicator computed
from the copied nodes, keeping the internal dyad shocks unchanged.
Then $J_{k,q}^*(t)$ has the same distribution as $J_{k,q}(t)$
and is independent of $J_{\ell,q}(t)$.
Also, the two indicators agree whenever all copied nodes
equal their original values. Therefore
\[
\begin{aligned}
 |\Cov\{J_{\ell,q}(t),J_{k,q}(t)\}|
 &=\left|E\left[
 J_{\ell,q}(t)\{J_{k,q}(t)-J_{k,q}^*(t)\}
 \right]\right|\le P\{J_{k,q}(t)\ne J_{k,q}^*(t)\}\\
 &\le P\{Z_i^*\ne Z_i\text{ for some }i\in I_{k,q}\}\le Cq^{2\nu}(1+d)^{-\mathfrak d}.
\end{aligned}
\]

For the fixed cube $I_{\ell,q}$, at most $C(d+1)^{\nu-1}$
other cubes have starting sites at distance in $[d,d+1)$
from its starting site. Summing these covariance bounds over
nearby and distant cubes yields
\[
\begin{aligned}
 \Var\left\{\frac1{N_q}\sum_{\ell=1}^{N_q}J_{\ell,q}(t)\right\}
 &\le\frac C{N_q}
 \left\{q^\nu +q^{2\nu}\sum_{d\ge q}(1+d)^{\nu-1}(1+d)^{-\mathfrak d}\right\}\le\frac C{N_q}\{q^\nu +q^{3\nu-\mathfrak d}\}
 \le C\frac{q^\nu}{n}\longrightarrow0.
\end{aligned}
\]
For the last inequality, use $N_q\sim n$ and
\(
 q^{3\nu-\mathfrak d}\le q^\nu,
\)
since $\mathfrak d>2\nu$.
Chebyshev's inequality and the convergence of $F_q(t)$ now imply
\[
 \frac1{N_q}\sum_{\ell=1}^{N_q}J_{\ell,q}(t)
 \to^P F_{|W|}(t)
\]
at every continuity point of $F_{|W|}$.

\textit{Step 2: Replace the true center and establish coverage.}
\quad
The feasible cdf uses $\bm{a}^{\top}\widehat{\bm{\beta}}$ instead of $\bm{a}^{\top}\bm{\beta}$.
The reverse triangle inequality gives, for every subsample,
\begin{align}\label{eq: boundbeta}
 \left|
 q^\nu |\bm{a}^{\top}(\widehat{\bm{\beta}}_{\ell,q}-\widehat{\bm{\beta}})|
 -q^\nu |\bm{a}^{\top}(\widehat{\bm{\beta}}_{\ell,q}-\bm{\beta})|
 \right|
 \le q^\nu |\bm{a}^{\top}(\widehat{\bm{\beta}}-\bm{\beta})|
 =O_p(q^\nu/n)=o_p(1).
\end{align}
Therefore, on the event
$q^\nu|\bm a^\top(\widehat{\bm\beta}-\bm\beta)|\le\epsilon$,
replacing the true center $\bm \beta$ by its estimate $\widehat{\bm\beta}$ in $\bm{a}^{\top}(\widehat{\bm{\beta}}_{\ell,q}-\bm{\beta})$ changes  the statistic by at most $\epsilon$. Recall that $\widehat{F}_{n,q}(t)
 =\frac1{N_q}\sum_{\ell=1}^{N_q}
       \ind\{q^\nu |\bm{a}^{\top}(\widehat{\bm{\beta}}_{\ell,q}-\widehat{\bm{\beta}})|\le t\}$, and hence we can deduce that
\begin{align}\label{eq: proof theorem4-1}
 \frac1{N_q}\sum_{\ell=1}^{N_q}J_{\ell,q}(t-\epsilon)
 \le \widehat F_{n,q}(t)
 \le \frac1{N_q}\sum_{\ell=1}^{N_q}J_{\ell,q}(t+\epsilon).
\end{align}
For each fixed $\epsilon>0$, this event $q^\nu|\bm a^\top(\widehat{\bm\beta}-\bm\beta)|\le\epsilon$ has probability tending
to one by the preceding centering bound in \eqref{eq: boundbeta}.

Fix a continuity point $t$ of $F_{|W|}$ and any $\eta>0$.
By continuity at $t$, choose $\epsilon>0$ sufficiently small
that $t-\epsilon$ and $t+\epsilon$ are also continuity points and
\begin{align}\label{eq: proof theorem4-2}
 F_{|W|}(t)-\frac\eta2
 <F_{|W|}(t-\epsilon)
 \le F_{|W|}(t+\epsilon)
 <F_{|W|}(t)+\frac\eta2.
\end{align}
Step 1 gives
\(
 \frac1{N_q}\sum_{\ell=1}^{N_q}J_{\ell,q}(t\pm\epsilon)
 \to^P F_{|W|}(t\pm\epsilon).
\)
Thus, with probability tending to one, each average differs
from its limit by less than $\eta/2$.
Combining this fact with \eqref{eq: proof theorem4-1} and \eqref{eq: proof theorem4-2} gives
\[
 F_{|W|}(t)-\eta
 <\widehat F_{n,q}(t)
 <F_{|W|}(t)+\eta
\]
with probability tending to one.
Since $\eta>0$ is arbitrary,
$\widehat F_{n,q}(t)\to^P F_{|W|}(t)$.

Let $q_{1-\alpha}$ be the $1-\alpha$ quantile of $|W|$.
The assumed continuity and strict increase at $q_{1-\alpha}$ imply,
for arbitrarily small $\epsilon>0$ with continuity points
$q_{1-\alpha}-\epsilon$ and $q_{1-\alpha}+\epsilon$,
\[
 F_{|W|}(q_{1-\alpha}-\epsilon)<1-\alpha
 <F_{|W|}(q_{1-\alpha}+\epsilon).
\]
The same inequalities hold for $\widehat F_{n,q}$ with
probability tending to one. By the definition of its quantile,
\[
 P\{q_{1-\alpha}-\epsilon<\widehat q_{n,q}(1-\alpha)
                      \le q_{1-\alpha}+\epsilon\}\longrightarrow1.
\]
Hence $\widehat q_{n,q}(1-\alpha)\to^P q_{1-\alpha}$.

Since $n|\bm{a}^{\top}(\widehat{\bm{\beta}}-\bm{\beta})|\Rightarrow|W|$ and
$P(|W|\le q_{1-\alpha})=1-\alpha$, Slutsky's theorem gives
\[
 P\left\{|\bm{a}^{\top}(\widehat{\bm{\beta}}-\bm{\beta})|
       \le\frac{\widehat q_{n,q}(1-\alpha)}n\right\}
 \longrightarrow1-\alpha.
\]
By construction,
$r_{n,\alpha}\ge\widehat q_{n,q}(1-\alpha)/n$.
Therefore
\(
 \liminf_{n\to\infty}P\{\bm{a}^{\top}\bm{\beta}\in C_{n,\alpha}\}
 \ge1-\alpha,
\)
which proves part (ii).

\textit{Step 3: Establish exact coverage in the Gaussian cases.}
\quad
In the slow Gaussian case, the full-sample and subsample limits are
\[
\begin{aligned}
 \sqrt n\bm{a}^{\top}(\widehat{\bm{\beta}}-\bm{\beta})
 &\Rightarrow N(0,\bm a^\top\bm V_\gamma\bm a),\qquad
 q^{\nu/2}\bm{a}^{\top}(\widehat{\bm{\beta}}_{\ell,q}-\bm{\beta})
 \Rightarrow N(0,\bm a^\top\bm V_\gamma\bm a).
\end{aligned}
\]
Repeat Step 1 with the indicators
$\ind\{q^{\nu/2}|\bm{a}^{\top}(\widehat{\bm{\beta}}_{\ell,q}-\bm{\beta})|\le t\}$.
The covariance bounds are unchanged because the indicators
still take values in $[0,1]$ and depend on the same subsamples.
The centering error is now
\[
 q^{\nu/2}|\bm{a}^{\top}(\widehat{\bm{\beta}}-\bm{\beta})|
 =O_p\bigl(\frac{q^{\nu/2}}{\sqrt n}\bigr)=o_p(1),
\]
so Step 2 also applies. Recall that $\widehat{F}_{n,q}(t)
 =\frac1{N_q}\sum_{\ell=1}^{N_q}
       \ind\{q^\nu |\bm{a}^{\top}(\widehat{\bm{\beta}}_{\ell,q}-\widehat{\bm{\beta}})|\le t\}$. Thus, we can plug in $q^{\nu/2}\,t$, and the argument above gives, for every $t\ge0$,
\[
 \widehat F_{n,q}(q^{\nu/2}\,t)
 =\frac1{N_q}\sum_{\ell=1}^{N_q}
 \ind\{q^{\nu/2}|\bm{a}^{\top}(\widehat{\bm{\beta}}_{\ell,q}-\widehat{\bm{\beta}})|\le t\} \to^P
 2\Phi\left(\frac{t}{\sqrt{\bm a^\top\bm V_\gamma\bm a}}\right)-1,
\]
where $\Phi$ is the standard normal cdf.
The $1-\alpha$ quantile of the rescaled empirical cdf is exactly
\[
 \inf\{t\ge0:\widehat F_{n,q}(q^{\nu/2}t)\ge1-\alpha\}
 =\frac{\widehat q_{n,q}(1-\alpha)}{q^{\nu/2}}.
\]
The limiting cdf is that of an absolute Gaussian variable,
whose $1-\alpha$ quantile is
$z_{1-\alpha/2}\sqrt{\bm a^\top\bm V_\gamma\bm a}$.
It is continuous and strictly increasing at this quantile.
Thus the same quantile argument as in Step 2 gives
\(
 \frac{\widehat q_{n,q}(1-\alpha)}{q^{\nu/2}}
 \to^P
 z_{1-\alpha/2}\sqrt{\bm a^\top\bm V_\gamma\bm a}.
\)
The subsampling term $\widehat q_{n,q}(1-\alpha)/n$ therefore
satisfies
\[
 \sqrt n\,\frac{\widehat q_{n,q}(1-\alpha)}n
 =
 \sqrt{\frac{q^\nu}{n}}\,
 \frac{\widehat q_{n,q}(1-\alpha)}{q^{\nu/2}}
 \to^P0.
\]
Meanwhile, Theorem~\ref{thm:hac_jk} gives
\(
 \sqrt n\,z_{1-\alpha/2}
 \sqrt{\max\{\bm a^\top
       \widehat{\bm V}^{\mathrm{JK}}_{\mathrm{DN}}\bm a,0\}}
 \to^P
 z_{1-\alpha/2}\sqrt{\bm a^\top\bm V_\gamma\bm a}.
\)
After multiplication by $\sqrt n$, the Gaussian jackknife term
converges to a strictly positive constant, while the subsampling
term converges to zero. Thus the maximum defining $r_{n,\alpha}$
selects the Gaussian jackknife term with probability tending to one,
and
\[
 \sqrt n\,r_{n,\alpha}
 \to^P z_{1-\alpha/2}\sqrt{\bm a^\top\bm V_\gamma\bm a}.
\]
Together with the full-sample Gaussian limit, this proves
coverage converges to $1-\alpha$ in the slow Gaussian case.

In the fast Gaussian case, the same subsampling argument
applies at rate $q^\nu$, with limiting distribution
$N(0,\bm a^\top\bm V_\zeta\bm a)$. Thus
\[
 \widehat q_{n,q}(1-\alpha)
 \to^P z_{1-\alpha/2}\sqrt{\bm a^\top\bm V_\zeta\bm a}.
\]
Theorem~\ref{thm:hac_jk} also gives
\(
 n\,z_{1-\alpha/2}
 \sqrt{\max\{\bm a^\top
       \widehat{\bm V}^{\mathrm{JK}}_{\mathrm{DN}}\bm a,0\}}
 \to^P
 z_{1-\alpha/2}\sqrt{\bm a^\top\bm V_\zeta\bm a}.
\)
After multiplication by $n$, both terms in the maximum defining
$r_{n,\alpha}$ converge to the same positive constant. Hence
\[
 n\,r_{n,\alpha}
 \to^P z_{1-\alpha/2}\sqrt{\bm a^\top\bm V_\zeta\bm a}.
\]
Combining this with
$n\bm{a}^{\top}(\widehat{\bm{\beta}}-\bm{\beta})
\Rightarrow N(0,\bm a^\top\bm V_\zeta\bm a)$
proves coverage converges to $1-\alpha$ in the fast Gaussian
case and completes part (i).
\end{proof}

\section{Proof for Fixed-effect Models}
\label{app:fe_projection}
Let $\bm D$ be the $M_n\times n$ matrix whose row for dyad
$(i,j)$ is $(\bm e_i+\bm e_j)^\top$, where $\bm e_i$ is the
$i$th unit vector. Then
$\bm P_D=\bm D(\bm D^\top\bm D)^{-1}\bm D^\top$ and
$\bm M_D=\bm I_{M_n}-\bm P_D$. For $n>2$,
\begin{equation}
 \bm D^\top\bm D=(n-2)\bm I_n+\bm1_n\bm1_n^\top,
 \qquad
 (\bm D^\top\bm D)^{-1}
 =\frac1{n-2}
 \left(\bm I_n-\frac{\bm1_n\bm1_n^\top}{2(n-1)}\right).
 \label{eq:fe_projection}
\end{equation}
Here $\bm1_n$ is the vector of ones. For a scalar dyadic
variable, write
$r_i^v=\sum_{\substack{j=1\\j\ne i}}^n v_{ij}$ and
$S_v=\sum_{1\le i<j\le n}v_{ij}$. The exact transformation is
\begin{equation}
 \ddot v_{ij}
 =v_{ij}-\frac{r_i^v+r_j^v}{n-2}
       +\frac{2S_v}{(n-2)(n-1)}.
 \label{eq:fe_within}
\end{equation}
Apply this coordinatewise to the regressors and to the outcome.
The fixed-effect slope is
$\widehat{\bm\beta}_{\mathrm{FE}}
=(\bm X^\top\bm M_D\bm X)^{-1}\bm X^\top\bm M_D\bm y$.
For a retained node set $T$ with $m=|T|$, let $M_m=m(m-1)/2$
and let $\bm v_T$ stack the observations $v_{ij}$ with $i<j$
and $i,j\in T$. Define the $M_m\times m$ node-indicator matrix
$\bm D_T$, with columns indexed by $k\in T$, by
\[
 (\bm D_T)_{(i,j),k}
 =\ind\{k=i\}+\ind\{k=j\},
 \qquad i<j,\quad i,j,k\in T.
\]
Thus $(\bm D_T\bm\alpha_T)_{ij}=\alpha_i+\alpha_j$, where
$\bm\alpha_T=(\alpha_k)_{k\in T}$. Set
\begin{equation}
 \bm P_{D,T}
 =\bm D_T(\bm D_T^\top\bm D_T)^+\bm D_T^\top,
 \qquad
 \bm M_{D,T}=\bm I_{M_m}-\bm P_{D,T}.
 \label{eq:fe_retained_projection}
\end{equation}
Here $\bm P_{D,T}\bm v_T$ is the vector of fitted values from
regressing $\bm v_T$ on the retained node indicators, and
$\bm M_{D,T}\bm v_T$ is the vector of within residuals.
For $m>2$, the generalized inverse is an ordinary inverse:
\[
 (\bm D_T^\top\bm D_T)^{-1}
 =\frac1{m-2}
 \left(\bm I_m-\frac{\bm1_m\bm1_m^\top}{2(m-1)}\right).
\]
For $T=\{1,\ldots,n\}$, these matrices equal
$\bm D$, $\bm P_D$, and $\bm M_D$, respectively.
Each retained complete sample uses the within formula above
with its own node count and sums. For singular designs, use the
Moore-Penrose inverse; the Gaussian proof below establishes
nonsingularity of the full and neighborhood-deleted designs
with probability tending to one.

\subsection{Proof of Theorem~\ref{thm:fe_gaussian}}
\begin{proof}
\textit{Step 1: Compare sample and population demeaning.}\quad
For a retained node set $T$, let
$\breve{\bm{\beta}}_T$ be the infeasible OLS estimator using regressors $\bm{x}_{ij}^\circ$ and
outcomes $\bm{x}_{ij}^{\circ\top}\bm{\beta}+u_{ij}^\circ$, for $i<j$ with $i,j\in T$. Set
\[
 \bm{Q}_T^o=\frac{\bm{X}_T^{\circ\top}\bm{X}_T^\circ}{M_m},\quad
 \bm{J}_T=\frac{\bm{X}_T^{\circ\top}\bm{P}_{D,T}\bm{X}_T^\circ}{M_m},\quad
 \bm{R}_T=\frac{\bm{X}_T^{\circ\top}\bm{P}_{D,T}\bm{u}_T^\circ}{M_m},\quad
 \bm{Q}_T^{\mathrm{FE}}=\bm{Q}_T^o-\bm{J}_T.
\]
The projection formula \eqref{eq:fe_projection}, the equality
$\bm{M}_{D,T}\bm{v}_T=\bm{M}_{D,T}\bm{v}_T^\circ$, and Frisch-Waugh-Lovell give
\begin{equation}
 \widehat{\bm{\beta}}_{\mathrm{FE},T}-\breve{\bm{\beta}}_T
 =(\bm{Q}_T^{\mathrm{FE}})^{-1}
 \{\bm{J}_T(\breve{\bm{\beta}}_T-\bm{\beta})-\bm{R}_T\}.
 \label{eq:fe_oracle_comparison}
\end{equation}
By Lemma~\ref{lem:fe_rows},
\begin{equation}
 E\|\bm{J}_T\|^p\le C m^{-p}\ (p=2,4),\qquad
 E\|\bm{R}_T\|^2\le C m^{-2}.
 \label{eq:fe_projection_rates}
\end{equation}
For the second bound, use
$\|\bm{X}_T^{\circ\top}\bm{P}_{D,T}\bm{u}_T^\circ\|^2
\le\|\bm{P}_{D,T}\bm{X}_T^\circ\|_F^2\|\bm{P}_{D,T}\bm{u}_T^\circ\|^2$
and Cauchy-Schwarz.

Let $\mathcal T_n$ contain the full node set and all deleted node
sets, and fix $c=\lambda_{\min}(\bm{Q}^\circ)/4$. Define the design event
\[
 \mathcal E_n=\left\{
 \min_{T\in\mathcal T_n}
 \min\{\lambda_{\min}(\bm{Q}_T^o),\lambda_{\min}(\bm{Q}_T^{\mathrm{FE}})\}
 \ge c\right\}.
\]
Both normalized inverses have norm at most $c^{-1}$ on this event.
We have $P(\mathcal E_n)\to1$. Indeed, the
oracle design obeys the dyadic LLN and deletion changes its average by
$O(L^\nu/n)$, since $\bm{x}^\circ$ is bounded. For deleted-sample row sums,
\[
 \sum_{\substack{i=1\\i\notin B_\ell}}^n\|\bm{r}_{i,-\ell}^x\|^2
 \le2\sum_{i=1}^n\|\bm{r}_i^x\|^2+C nL^{2\nu},
 \qquad
 \max_\ell\|\bm{J}_{-\ell}\|=O_p(n^{-1})+O(L^{2\nu}/n^2)=o_p(1).
\]
Thus $\bm{Q}_T^{\mathrm{FE}}\to^P \bm{Q}^\circ$ uniformly over these sets.
The indicator $\ind_{\mathcal E_n}$ is retained in every moment
bound involving an inverse below. This avoids imposing inverse
moments on nearly singular sample designs.

\textit{Step 2: Bound the difference in the slow Gaussian case.}\quad
In case (i), the projection decomposition and the bounded mixing
moment inequalities give
$E[\ind_{\mathcal E_n}\|\breve{\bm{\beta}}_T-\bm{\beta}\|^4]\le C m^{-2}$.
Indeed, Lemma~\ref{lem:spatial_moments} gives
\[
 \begin{aligned}
 E\Big\|\sum_{i\in T}\bm{\gamma}_i^\circ\Big\|^4&\le Cm^2,\\
 E\Big\|\sum_{\substack{i,j\in T\\i<j}}(\bm{\xi}_{ij}^\circ+\bm{\zeta}_{ij}^\circ)\Big\|^4
 &\le Cm^4,\\
 E\left\|M_m^{-1}\sum_{\substack{i,j\in T\\i<j}}\bm{s}_{ij}^\circ\right\|^4
 &\le C M_m^{-4}(m^4m^2+m^4)\le Cm^{-2}.
 \end{aligned}
\]
The inverse of $\bm{Q}_T^o$ is bounded on $\mathcal E_n$.
On $\mathcal E_n$, the inverse of $\bm Q_T^{\mathrm{FE}}$ is also
bounded. Using \eqref{eq:fe_oracle_comparison} and Cauchy-Schwarz,
\[
 \begin{aligned}
 E[\ind_{\mathcal E_n}
   \|\widehat{\bm\beta}_{\mathrm{FE},T}-\breve{\bm\beta}_T\|^2]
 &\le C(E\|\bm J_T\|^4)^{1/2}
       (E[\ind_{\mathcal E_n}\|\breve{\bm\beta}_T-\bm\beta\|^4])^{1/2}
       +CE\|\bm R_T\|^2\\
 &\le C(m^{-2}m^{-1}+m^{-2}).
 \end{aligned}
\]
Thus, uniformly over the retained sets in question,
\begin{equation}
 E[\ind_{\mathcal E_n}
 \|\widehat{\bm{\beta}}_{\mathrm{FE},T}-\breve{\bm{\beta}}_T\|^2]\le C m^{-2}.
 \label{eq:fe_error_slow}
\end{equation}
The leading full-sample score is consequently
$2n^{-1}\sum_{i=1}^n\bm{\gamma}_i^\circ+o_p(n^{-1/2})$, proving the first CLT.
Here we only require $E_\perp[\bm{s}_{12}^\circ]=\bm0$.
The canonical remainder may have means that depend on $\bm s_j-\bm s_i$;
these are already included in its
$O_p(n^{-1})$ average.

\textit{Step 3: Bound the difference in the fast Gaussian case.}\quad
In case (ii), conditional on $\mathcal Z,\bm{X}$, the dyad disturbances
$u_{ij}^\circ$ are independent, centered, and bounded. Independence
follows from the independent dyad shocks in the representation, and
persists after conditioning on their dyad-specific regressors.
Consequently
\[
 E\big[\|\bm{X}_T^{\circ\top}\bm{P}_{D,T}\bm{u}_T^\circ\|^2
       \mid\mathcal Z,\bm{X}\big]
 \le C\|\bm{P}_{D,T}\bm{X}_T^\circ\|_F^2,
 \qquad E\|\bm{R}_T\|^2\le C m^{-3}.
\]
Conditional independence also gives
$E[\ind_{\mathcal E_n}\|\breve{\bm{\beta}}_T-\bm{\beta}\|^4\mid\mathcal Z,\bm{X}]
\le C m^{-4}$, since $\mathcal E_n$ is measurable with respect
to $(\mathcal Z,\bm{X})$. The same inequality used in case (i) now
has upper bound
\[
 C(m^{-2}m^{-2}+m^{-3})\le Cm^{-3}.
\]
It follows that
\begin{equation}
 E[\ind_{\mathcal E_n}
 \|\widehat{\bm{\beta}}_{\mathrm{FE},T}-\breve{\bm{\beta}}_T\|^2]\le C m^{-3}.
 \label{eq:fe_error_fast}
\end{equation}
To verify the oracle's  degeneracy, put
$\bm{m}^\circ(z,z')=E[\bm{s}_{ij}^\circ\mid Z_i=z,Z_j=z']$.
Conditional centering gives $\bm{m}^\circ(Z_i,Z_j)=\bm0$ almost surely
for every actual pair. By Berbee's coupling lemma and the geometric $\beta$-mixing bound,
boundedness of this kernel gives
\[
 E_\perp[\|\bm{m}^\circ(Z_i,Z_j)\|^2]
 \le E\|\bm{m}^\circ(Z_i,Z_j)\|^2+C\rho^{d_{ij}}
 =C\rho^{d_{ij}}.
\]
Let $d_{ij}\to\infty$. The left-hand side is zero, whence
$\bm{\gamma}_i^\circ=\bm0$ under the marginal distribution and
$\bm{\xi}_{ij}^\circ=\bm{m}^\circ(Z_i,Z_j)-\bm{\gamma}_i^\circ-\bm{\gamma}_j^\circ=\bm0$
under every actual pair law. The oracle fast Gaussian CLT and
\eqref{eq:fe_error_fast} establish the second CLT.

\textit{Step 4: Compare the jackknife quadratic forms.}\quad
Let
$\bm{e}_T=\widehat{\bm{\beta}}_{\mathrm{FE},T}-\breve{\bm{\beta}}_T$ and $\bm{e}_n$ denote
the corresponding full-sample difference. On $\mathcal E_n$, Markov's
inequality, $m_\ell\asymp L^\nu$, and the bounds just proved give the
following orders. They hold without restriction because
$P(\mathcal E_n^c)\to0$:
\[
 \sum_{\ell=1}^{n}\frac1{m_\ell}\|\bm{e}_{-\ell}-\bm{e}_n\|^2
 =\begin{cases}
 O_p((nL^\nu)^{-1})=o_p(n^{-1}),&\text{case (i)},\\
 O_p((n^2L^\nu)^{-1})=o_p(n^{-2}),&\text{case (ii)}.
 \end{cases}
\]
The Cauchy-Schwarz inequality for the sums of outer products now
shows that the uncorrected FE and oracle jackknives differ by
$o_p(n^{-1})$ and $o_p(n^{-2})$, respectively. The oracle jackknife
limits follow from the earlier proof: in case (i) it uses
$E_\perp[\bm{s}_{12}^\circ]=\bm0$ and the canonical remainder bound, and in case (ii)
all assumptions of the fast Gaussian result hold directly.

\textit{Step 5: Compare the incidence corrections.}\quad
Observe that
\[
 \sum_{1\le i<j\le n}\|\ddot{\bm{x}}_{ij}-\bm{x}_{ij}^\circ\|^2=O_p(n),\qquad
 \sum_{1\le i<j\le n}|(\bm{M}_D\bm{u}^\circ)_{ij}-u_{ij}^\circ|^2=O_p(n).
\]
The within transformation has uniformly bounded absolute row sums
by \eqref{eq:fe_within}. Thus bounded population residuals, the slope
rates already proved, and Cauchy-Schwarz imply
\[
 \frac1{M_n}\sum_{1\le i<j\le n}
 \|\widehat{\bm{s}}_{ij,\mathrm{FE}}-\bm{s}_{ij}^\circ\|^2=o_p(1),\qquad
 \frac1{M_n}\sum_{1\le i<j\le n}\widehat{\bm{s}}_{ij,\mathrm{FE}}
 \widehat{\bm{s}}_{ij,\mathrm{FE}}^{\top}
 \to^P E_\perp[\bm{s}_{12}^\circ \bm{s}_{12}^{\circ\top}].
\]
The weights $K^\cup_{ij}-1$ are bounded. Also $K^\cup_{ij}=2$
outside $o(n^2)$ dyads by Lemma~\ref{lem:dyad_deletion_weights}.
On events with probability tending to one, the within scores are
uniformly bounded, by bounded population residuals and consistency
of the slope. Thus the same limits hold with weights $K^\cup_{ij}-1$.
The correction is $O_p(n^{-2})$ in case (i). In case (ii), its
$n^2$-scaled limit is $\bm{V}_{\zeta^\circ}$, subtracting one copy
from the oracle uncorrected limit $2\bm{V}_{\zeta^\circ}$.

\textit{Step 6: Establish the slow-rate DN-Dyadic result.}\quad
The displayed projection
bounds and the root-$n$ slope rate strengthen the score bound to
$\sum_{1\le i<j\le n}\|\widehat{\bm{s}}_{ij,\mathrm{FE}}-\bm{s}_{ij}^\circ\|^2=O_p(n)$.
For $\Delta \bm{G}_i=\sum_{\substack{j=1\\j\ne i}}^n(\widehat{\bm{s}}_{ij,\mathrm{FE}}-\bm{s}_{ij}^\circ)$,
\[
 \sum_{i=1}^n\|\Delta \bm{G}_i\|^2\le
 Cn\sum_{1\le i<j\le n}\|\widehat{\bm{s}}_{ij,\mathrm{FE}}-\bm{s}_{ij}^\circ\|^2
 =O_p(n^2).
\]
The Wendland node-kernel matrix is positive semidefinite and has
operator norm $O(L^\nu)$. Its score-error quadratic form is $O_p(n^2L^\nu)$;
the cross term is $O_p(n^{5/2}L^{\nu/2})=o_p(n^3)$.
The DN/node-kernel difference involves at most $C n^2L^{2\nu}$
dyad pairs and bounded scores, hence is $o_p(n^3)$.
The oracle meat divided by $n^3$ converges to
$\bm{\Omega}_{\gamma^\circ}$, completing the DN assertion after sandwiching.
\end{proof}

\subsection{Fast limit with node fixed effects}
\label{app:fe_fast}

We state the additional conditions used for the non-Gaussian
fixed-effect result. Let
$\bm f_i=(f_1(Z_i),\ldots,f_r(Z_i))^\top$ be centered and bounded,
with fixed dimension $r$, and suppose
\begin{equation}
 \begin{aligned}
 \xi_{ij,k}^\circ&=\bm f_i^\top\bm B_k\bm f_j,\\
 E[x_{ij,k}^\circ\mid\mathcal Z]
 &=\bm f_i^\top\bm A_{x,k}\bm f_j,
 & E[u_{ij}^\circ\mid\mathcal Z]
 &=\bm f_i^\top\bm A_u\bm f_j.
 \end{aligned}
 \label{eq:fe_finite_representations}
\end{equation}
The coefficient matrices are fixed and symmetric. These
representations hold under both the actual spatial distribution
and independent node draws from $F$, with conditional means
defined by averaging over the dyad shock. The first is the
finite-dimensional case of Assumption~\ref{ass:finite_interaction};
the other two permit an explicit calculation of the sample
demeaning term.

Define the centered dyad shocks and their covariance by
\[
 e_{ij,k}^x=x_{ij,k}^\circ-\bm f_i^\top\bm A_{x,k}\bm f_j,
 \qquad
 e_{ij}^u=u_{ij}^\circ-\bm f_i^\top\bm A_u\bm f_j,
 \qquad
 \kappa_k=E_\perp[e_{12,k}^xe_{12}^u].
\]
For this fixed vector, write $\bm C_f=\bm C_f^{(r)}$ and
$\bm\Omega_f=\bm\Omega_f^{(r)}$, as defined in
Section~\ref{subsec:nongaussian}, and let
$\bm{\mathcal G}\sim N(\bm0,\bm\Omega_f)$.
The interaction and demeaning limits use
\[
 \begin{aligned}
 \bm{\mathcal Q}_B(\bm{\mathcal G})
 &=\left(\bm{\mathcal G}^\top\bm B_k\bm{\mathcal G}
         -\operatorname{tr}(\bm B_k\bm C_f)\right)_{k=1}^K,\\
 \bm{\mathcal R}_A(\bm{\mathcal G})
 &=\left(\bm{\mathcal G}^\top\bm A_{x,k}\bm C_f\bm A_u
               \bm{\mathcal G}+\kappa_k\right)_{k=1}^K.
 \end{aligned}
\]
Both depend on the same Gaussian vector.

\begin{theorem}[Fixed effects with a non-Gaussian fast limit]
\label{thm:fe_fast}
Suppose Assumption~\ref{ass:fe} holds,
$\bm\gamma_i^\circ=\bm0$ almost surely, and the finite
representations above hold with the stated boundedness
and symmetry conditions. Then
\begin{equation}
 n(\widehat{\bm\beta}_{\mathrm{FE}}-\bm\beta)
 \Rightarrow
 (\bm Q^\circ)^{-1}
 \{\bm{\mathcal Q}_B(\bm{\mathcal G})
   +\bm{\mathcal E}_\zeta^\circ
   -2\bm{\mathcal R}_A(\bm{\mathcal G})\},
 \label{eq:fe_fast_limit}
\end{equation}
where $\bm{\mathcal E}_\zeta^\circ
\sim N(\bm0,\bm\Omega_{\zeta^\circ})$ is independent of
$\bm{\mathcal G}$. The same limit holds on growing lattice
cubes, with their node count replacing $n$.
\end{theorem}

\begin{proof}
\textit{Step 1: Expand the sample projection term.}\quad
Let $\bm{S}_f=\sum_{i=1}^n \bm{f}_i$, $\bm{F}_n=\bm{S}_f/\sqrt n$, and
$\bm{C}_{f,n}=n^{-1}\sum_{i=1}^n \bm{f}_i\bm{f}_i^{\top}$. For coordinate $k$, put
$r_{x,k,i}=\sum_{\substack{j=1\\j\ne i}}^nx_{ij,k}^\circ$ and
$r_{u,i}=\sum_{\substack{j=1\\j\ne i}}^nu_{ij}^\circ$. Equation~\eqref{eq:fe_projection}
gives
\begin{align}
 \frac n{M_n}\bm{X}_k^{\circ\top}\bm{P}_D\bm{u}^\circ
 =\frac n{M_n(n-2)}\left\{
 \sum_{i=1}^n r_{x,k,i}r_{u,i}
 -\frac{(\sum_{i=1}^n r_{x,k,i})(\sum_{i=1}^n r_{u,i})}{2(n-1)}\right\}.
 \label{eq:fe_projection_expansion}
\end{align}
Write $E_{x,k,i}=\sum_{\substack{j=1\\j\ne i}}^ne_{ij,k}^x$ and
$E_{u,i}=\sum_{\substack{j=1\\j\ne i}}^ne_{ij}^u$. Then
\[
 r_{x,k,i}=\bm{f}_i^{\top}\bm{A}_{x,k}\bm{S}_f-\bm{f}_i^{\top}\bm{A}_{x,k}\bm{f}_i+E_{x,k,i},\qquad
 r_{u,i}=\bm{f}_i^{\top}\bm{A}_u\bm{S}_f-\bm{f}_i^{\top}\bm{A}_u\bm{f}_i+E_{u,i}.
\]
Bounded $\bm{f}_i$ and $\|\bm{S}_f\|=O_p(\sqrt n)$ imply that the product
of the two node parts, summed and divided by $n^2$, is
$\bm{F}_n^{\top}\bm{A}_{x,k}\bm{C}_{f,n}\bm{A}_u\bm{F}_n+o_p(1)$.
For a cross term involving one $E$ and one node part, conditional
independence of dyad shocks yields the variance bound
\[
 \Var\left\{\left.\sum_{i=1}^n w_iE_{u,i}\right|\mathcal Z\right\}
 \le C n\sum_{i=1}^n w_i^2,
 \qquad w_i=\bm{f}_i^{\top}\bm{A}_{x,k}\bm{S}_f.
\]
Since $\sum_{i=1}^nw_i^2=O_p(n^2)$, division by $n^2$ makes this term
$o_p(1)$. The diagonal subtractions $\bm{f}_i^{\top}\bm{A} \bm{f}_i$ are bounded and
are also negligible at this scale.

\textit{Step 2: Retain the residual covariance term.}\quad
For the product of the two residual row sums,
\[
 \frac1{n^2}E\left[\left.\sum_{i=1}^nE_{x,k,i}E_{u,i}\right|\mathcal Z\right]
 =\frac2{n^2}\sum_{1\le i<j\le n}E[e_{ij,k}^xe_{ij}^u\mid\mathcal Z]
 \to^P\kappa_k.
\]
To bound the conditional variance, expand
\[
 \sum_{i=1}^n E_{x,k,i}E_{u,i}
 =2\sum_{1\le i<j\le n}e_{ij,k}^xe_{ij}^u
 +\sum_{i=1}^n\sum_{\substack{j=1\\j\ne i}}^n\sum_{\substack{l=1\\l\notin\{i,j\}}}^n
                         e_{ij,k}^xe_{il}^u.
\]
The first sum has conditional variance at most $Cn^2$.
In the second sum, a covariance is zero unless the two unordered
dyad-label sets coincide. There are $n(n-1)(n-2)$ terms and
at most two matching terms for each one. Thus
\[
 \Var\left(\sum_{i=1}^nE_{x,k,i}E_{u,i}\,\middle|\,\mathcal Z\right)
 \le C(n^2+n^3),\qquad
 n^{-2}\sum_{i=1}^n\{E_{x,k,i}E_{u,i}
       -E[E_{x,k,i}E_{u,i}\mid\mathcal Z]\}=o_p(1).
\]
Therefore
\begin{equation}
 \frac1{n^2}\sum_{i=1}^n r_{x,k,i}r_{u,i}
 =\bm{F}_n^{\top}\bm{A}_{x,k}\bm{C}_{f,n}\bm{A}_u\bm{F}_n+\kappa_k+o_p(1).
 \label{eq:fe_row_product_limit}
\end{equation}
Also $\sum_{i=1}^n r_{x,k,i}=O_p(n)$ and $\sum_{i=1}^n r_{u,i}=O_p(n)$,
by the finite-dimensional conditional means and conditional
independence of the residual dyad shocks. The second term in braces
in \eqref{eq:fe_projection_expansion} is consequently negligible
after its prefactor, which is asymptotic to $2/n^2$.
Equations~\eqref{eq:fe_projection_expansion}-\eqref{eq:fe_row_product_limit}
thus give the limit $2\bm{\mathcal R}_A(\bm{\mathcal G})$.

\textit{Step 3: Combine the score and projection limits.}\quad
The proof of Theorem~\ref{thm:fast_nongaussian} applies to the oracle
score $\bm{s}^\circ$, yielding its quadratic and Gaussian components jointly with $\bm{F}_n$,
since $\bm{\gamma}_i^\circ=\bm0$. That proof only requires
$E_\perp[\bm{s}_{12}^\circ]=\bm0$; the means for individual pairs may depend on $\bm s_j-\bm s_i$
and need not be zero. These means remain in the centering term of the
quadratic form. Finally
$\bm{X}^{\circ\top}\bm{M}_D\bm{X}^\circ/M_n\to^P \bm{Q}^\circ$ by
Lemmas~\ref{lem:fe_rows} and \ref{lem:dyadic_lln}. Substitution into
\eqref{eq:fe_exact_score} proves \eqref{eq:fe_fast_limit}.

For spatial-cube FE estimators, all additive effects cancel
exactly, so the estimates are functions only of the stationary spatial
$(\bm{x}_{ij}^\circ,u_{ij}^\circ)$ array inside their cube. Translation invariance
and Lemma~\ref{lem:spatial_clt} give the same limiting law with $q^\nu$ nodes. The indicator variance
bound and centering argument in the proof of
Theorem~\ref{thm:conservative_coverage} apply without change.
This proves the stated FE coverage conclusions.
\end{proof}

\section{Technical Lemmas}

\subsection{Moments and mean score under independent node draws}
\begin{proof}[Proof of Lemma~\ref{lem:independent_node_moments}] 
For any scalar or vector function
$v_{ij}=v(Z_i,Z_j,Q_{ij})$ bounded in norm by $H$,
integrating out the independent dyad shock gives a function
of $(Z_i,Z_j)$ with the same bound.
The definition of the $\beta$-mixing coefficient and
Assumption~\ref{ass:mixing} therefore imply
\begin{equation}
 \|E[v_{ij}]-E_\perp[v_{12}]\|
 \le CH(1+d_{ij})^{-\mathfrak d}.
 \label{eq:bounded_independent_expectations}
\end{equation}
Thus, for fixed $H$, the difference tends to zero as
$d_{ij}\to\infty$.

Apply this bound to
$\min\{\|\bm x_{ij}\|^{8(\lambda+\delta)},H\}$.
Since Assumption~\ref{ass:projection} bounds the untruncated
moment uniformly over all dyads,
\[
 \begin{aligned}
 E_\perp\!\left[
 \min\{\|\bm x_{12}\|^{8(\lambda+\delta)},H\}\right]
 &=
 \lim_{d_{ij}\to\infty}
 E\!\left[
 \min\{\|\bm x_{ij}\|^{8(\lambda+\delta)},H\}\right]\le C.
 \end{aligned}
\]
As $H\to\infty$, the truncated quantity increases to
$\|\bm x_{12}\|^{8(\lambda+\delta)}$.
Monotone convergence therefore gives the stated bound
under $E_\perp$. The same argument applies to
$|u_{ij}|^{8(\lambda+\delta)}$.

To establish $\bm\mu=\bm0$, recall that
$\bm s_{ij}=\bm x_{ij}u_{ij}$.
Cauchy-Schwarz and the moment bounds give
\[
 E\|\bm s_{ij}\|^2
 \le
 \{E\|\bm x_{ij}\|^4\,E|u_{ij}|^4\}^{1/2}
 \le C,
 \qquad
 E_\perp\|\bm s_{12}\|^2\le C.
\]
Consequently, the expected norm of the score on
$\{\|\bm s_{ij}\|>H\}$ is at most $C/H$, under either
expectation. Apply \eqref{eq:bounded_independent_expectations}
to the bounded score
$\bm s_{ij}\ind\{\|\bm s_{ij}\|\le H\}$ and bound the
two remaining tails to obtain
\[
 \|E[\bm s_{ij}]-E_\perp[\bm s_{12}]\|
 \le CH(1+d_{ij})^{-\mathfrak d}+\frac{C}{H}.
\]
First let $d_{ij}\to\infty$ with $H$ fixed, and then let
$H\to\infty$. This proves
$E[\bm s_{ij}]\to E_\perp[\bm s_{12}]=\bm\mu$.
Since $E[\bm s_{ij}]=\bm0$ for every dyad by
Assumption~\ref{ass:projection}, $\bm\mu=\bm0$.
\end{proof}
\subsection{Spatial moments and limit theory}

\begin{lemma}[Spatial moment bounds]
\label{lem:spatial_moments}
Suppose Assumptions~\ref{ass:representation} and
\ref{ass:mixing} hold. We state the bounds using the score components
$\bm\gamma_i$, $\bm\xi_{ij}$, and $\bm\zeta_{ij}$ defined above.
For some $\delta>0$, suppose the moments specified below are uniformly
bounded.
\begin{enumerate}[label=(\roman*)]
\item \textit{Node sums.}
The centered node function $\bm\gamma_i=\bm\gamma(Z_i)$ satisfies
$E\|\bm\gamma_i\|^{4(\lambda+\delta)}\le C$.
For any deterministic set $T$ of node indices with $|T|=m$
(corresponding to $m$ distinct lattice locations),
\[
 E\Big\|\sum_{i\in T}\bm\gamma_i\Big\|^2\le Cm,\qquad
 E\Big\|\sum_{i\in T}\bm\gamma_i\Big\|^4\le Cm^2.
\]
\item \textit{Interaction sums.}
The symmetric node function $\bm\xi_{ij}=\bm\xi(Z_i,Z_j)$ satisfies
\(
 E_\perp[\bm\xi_{ij}\mid Z_i]
 =E_\perp[\bm\xi_{ij}\mid Z_j]=\bm0
\)
and has uniformly bounded $4(\lambda+\delta)$ moments both under the
actual spatial distribution of the two nodes and when they are
drawn independently from $F$. Then
\begin{equation}
\begin{aligned}
 \sum_{1\le i<j\le n}\sum_{1\le k<l\le n}
       \|E[\bm\xi_{ij}\bm\xi_{kl}^{\top}]\|&\le Cn^2,\qquad \max_{1\le i\le n}E\Big\|\sum_{\substack{j=1\\j\ne i}}^n
       \bm\xi_{ij}\Big\|^2&\le Cn.
\end{aligned}
 \label{eq:xi_covariance_summability}
\end{equation}
Also,
\(
 E\Big\|\sum_{1\le i<j\le n}\bm\xi_{ij}\Big\|^4\le Cn^4.
\)
\item \textit{Interaction plus dyad noise.}
Let $\bm\zeta_{ij}=\bm\zeta(Z_i,Z_j,Q_{ij})$ be symmetric in its
node arguments, with
$E_Q[\bm\zeta(z,z',Q_{12})]=\bm0$ for every $(z,z')$,
and uniformly bounded $4(\lambda+\delta)$ moments both under the actual
spatial distribution of the two nodes and when they are drawn
independently from $F$. Thus, conditional on $\mathcal Z$, these
variables are independent and centered. All three bounds in (ii)
also hold with $\bm\xi_{ij}$ replaced by
$\bm\xi_{ij}+\bm\zeta_{ij}$.
\end{enumerate}
\end{lemma}

\begin{proof}
We write $\gamma_i$, $\xi_{ij}$, and $\zeta_{ij}$ for scalar coordinates of
the corresponding vectors. The same arguments apply to products of
different coordinates, giving the stated vector bounds in fixed dimension.
To shorten the decay bounds, write
\[
 \tau=\mathfrak d\left(1-\frac1{\lambda+\delta}\right)>8\nu.
\]

\textit{(i) Node sums.}\quad
For $i\ne j$, Berbee's coupling lemma and the $\beta$-mixing
assumption give a variable $Z_i^*$ with the same distribution as
$Z_i$, independent of $Z_j$, such that
\[
P(Z_i^*\ne Z_i)\le C(1+d_{ij})^{-\mathfrak d}.
\]
Set $\gamma_i^*=\gamma(Z_i^*)$. Since $E[\gamma_i^*]=0$ and $\gamma_i^*$ is
independent of $\gamma_j$, we have $E[\gamma_i^*\gamma_j]=0$. Moreover,
$\gamma_i-\gamma_i^*=0$ whenever $Z_i=Z_i^*$. Hence H\"older's inequality gives
\[
\begin{aligned}
|E[\gamma_i \gamma_j]|
&=|E[(\gamma_i-\gamma_i^*)\gamma_j]|=\left|E\left[(\gamma_i-\gamma_i^*)\gamma_j
       \mathbf{1}\{Z_i\ne Z_i^*\}\right]\right|\\
&\le
\left\{E|(\gamma_i-\gamma_i^*)\gamma_j|^{\lambda+\delta}\right\}^{1/(\lambda+\delta)}
P(Z_i\ne Z_i^*)^{1-1/(\lambda+\delta)}\le C(1+d_{ij})^{-\tau},
\end{aligned}
\]
where the moment factor is uniformly bounded by the assumed moments.
For $i=j$, the same bound follows from $E[\gamma_i^2]\le C$ and $d_{ii}=0$.

For each fixed $i$, there are at most $C(h+1)^{\nu-1}$ lattice
sites $j$ satisfying $h\le d_{ij}<h+1$. Thus, with $|T|=m$,
\[
\begin{aligned}
E\left(\sum_{i\in T}\gamma_i\right)^2
&\le \sum_{i\in T}\sum_{j\in T}|E[\gamma_i \gamma_j]|\le C\sum_{i\in T}\sum_{h=0}^{\infty}
        (h+1)^{\nu-1}(h+1)^{-\tau}=Cm\sum_{h=0}^{\infty}(h+1)^{\nu-1-\tau}
\le Cm,
\end{aligned}
\]
where the series converges because $\tau>\nu$.

For the fourth moment, we first expand
\[
 E\left(\sum_{i\in T}\gamma_i\right)^4
 =\sum_{i,j,k,l\in T}E[\gamma_i \gamma_j \gamma_k \gamma_l].
\]
If every index is used at least twice, there are at most two distinct
indices. These terms have total absolute contribution at most $Cm^2$.

Otherwise, among the indices appearing once, choose the one whose
distance to its nearest other index is largest. Relabel it $i$ and put
$r=\min\{d_{ij},d_{ik},d_{il}\}$. Berbee's lemma now
couples $Z_i$ to a copy $Z_i^*$ independent of $(Z_j,Z_k,Z_l)$, with
$P(Z_i^*\ne Z_i)\le C(1+r)^{-\mathfrak d}$.
The mixing bound applies to one node and at most three other nodes.
Since $\gamma_i^*=\gamma(Z_i^*)$ is centered and independent of the remaining
factors, $E[\gamma_i^*\gamma_j\gamma_k\gamma_l]=0$. Hence, we have that
\begin{equation}
 \begin{aligned}
 |E[\gamma_i \gamma_j \gamma_k \gamma_l]|
 &=|E[(\gamma_i-\gamma_i^*)\gamma_j\gamma_k\gamma_l]|\\
 &\le
 \{E|(\gamma_i-\gamma_i^*)\gamma_j\gamma_k\gamma_l|^{\lambda+\delta}\}^{1/(\lambda+\delta)}
 P(Z_i^*\ne Z_i)^{1-1/(\lambda+\delta)}\\
 &\le C(1+r)^{-\tau}.
 \end{aligned}
 \label{eq:spatial_coupling_bound}
\end{equation}
The moment factor is bounded by H\"older's inequality and the
$4(\lambda+\delta)$ moment assumption.

For $2^h\le r<2^{h+1}$, group nodes joined by links shorter than
$2^{h+1}$. By the choice of $r$, an index appearing once cannot
form a group by itself. Each group therefore contributes at least
two of the four factors, giving at most two groups.
Choose one node per group ($\le m^2$ choices). Each remaining node
lies within distance $4\cdot2^{h+1}$ of a chosen node
($\le C2^{h\nu}$ choices). Since $r\ge 2^h$, the preceding bound satisfies
$C(1+r)^{-\tau}\le C2^{-h\tau}$.
Combining this count with the preceding coupling bound with $\tau>4\nu$ gives
\[
\begin{aligned}
E\left(\sum_{i\in T}\gamma_i\right)^4
&\le Cm^2+
C\sum_{h=0}^{\infty}
\underbrace{m^2 2^{4h\nu}}_{\text{number of terms}}\cdot
\underbrace{2^{-h\tau}}_{\text{bound per term}}\le Cm^2.
\end{aligned}
\]

\textit{(ii) Pair sums.}\quad
The relevant products are
\[
\xi_{ij}\xi_{kl}
\qquad\text{and}\qquad
\xi_{i_1j_1}\xi_{i_2j_2}\xi_{i_3j_3}\xi_{i_4j_4},
\]
containing four and eight node indices, respectively, counting
repetitions. If an index appears only once, choose among such indices
the one farthest from its nearest other index, as above. Call it $i$ in
$\xi_{ij}\xi_{kl}$. Let $Z_i^*$ be its Berbee copy, independent of
the remaining node variables. The canonical condition gives
\[
E[\xi(Z_i^*,Z_j)\xi_{kl}\mid Z_j,Z_k,Z_l]
=\xi_{kl}E_\perp[\xi_{ij}\mid Z_j]=0.
\]
Consequently, the same H\"older argument as in
\eqref{eq:spatial_coupling_bound} yields
\[
|E[\xi_{ij}\xi_{kl}]|
=
\left|E[\{\xi(Z_i,Z_j)-\xi(Z_i^*,Z_j)\}\xi_{kl}]\right|
\le C(1+r)^{-\tau},
\]
where $r$ is the selected node's distance to its nearest other node.
For the four-factor product, conditioning on all remaining nodes
gives the same cancellation and bound. The required moments remain
bounded: the replaced factor uses two independent node draws,
while each unchanged factor retains its actual pair distribution.
Both distributions are covered by the assumptions.

The preceding count now has at most two groups for two dyads and
at most four groups for four dyads. At distance $2^h\le r<2^{h+1}$,
the respective counts are at most $Cn^2 2^{4h\nu}$ and
$Cn^4 2^{8h\nu}$. Terms without an endpoint used once have at most
two or four distinct nodes and contribute at most $Cn^2$ or $Cn^4$.
Consequently,
\begin{equation}
 \begin{aligned}
 \sum_{1\le i<j\le n}\sum_{1\le k<l\le n}|E[\xi_{ij}\xi_{kl}]|
 &\le Cn^2\left\{1+\sum_{h=0}^{\infty}2^{-h(\tau-4\nu)}\right\}
 \le Cn^2,\\
 E\left(\sum_{1\le i<j\le n}\xi_{ij}\right)^4
 &\le\sum_{1\le i_1<j_1\le n}\cdots\sum_{1\le i_4<j_4\le n}
       |E[\xi_{i_1j_1}\xi_{i_2j_2}\xi_{i_3j_3}\xi_{i_4j_4}]|\\
 &\le Cn^4\left\{1+\sum_{h=0}^{\infty}2^{-h(\tau-8\nu)}\right\}
 \le Cn^4.
 \end{aligned}
 \label{eq:index_product_sum}
\end{equation}
Both series converge because $\tau>8\nu$.

For a row sum, fix $i$. Expanding its square sums over both $j$ and $k$,
with each different from $i$ but allowing $j=k$:
\[
 E\left(\sum_{\substack{j=1, j\ne i}}^n\xi_{ij}\right)^2
 =\sum_{\substack{j=1\\j\ne i}}^n\sum_{\substack{k=1\\k\ne i}}^nE[\xi_{ij}\xi_{ik}].
\]
The terms with $j=k$ contribute at most $Cn$. For the other terms,
the group containing $i$ has its location fixed; at most one other
group remains. The same count therefore gives $Cn2^{4h\nu}$ in
place of $Cn^2 2^{4h\nu}$, and
\[
 \sup_i\sum_{\substack{j=1\\j\ne i}}^n\sum_{\substack{k=1\\k\ne i}}^n|E[\xi_{ij}\xi_{ik}]|
 \le Cn\left\{1+\sum_{h=0}^{\infty}2^{-h(\tau-4\nu)}\right\}
 \le Cn.
\]

\textit{(iii) Interaction plus dyad noise.}\quad
Conditional on the entire node field $\mathcal Z$, the variables
$\zeta_{ij}$ are independent and centered, so
\[
 E[\xi_{ij}+\zeta_{ij}\mid\mathcal Z]=\xi_{ij}.
\]
If a node $i$ appears only once in a product, its dyad $(i,j)$ also
appears only once. Conditional independence therefore allows us to
integrate out $Q_{ij}$, replacing that factor by $\xi_{ij}$.
For example, when $i$ appears only in the first factor,
\[
 E[(\xi_{ij}+\zeta_{ij})(\xi_{kl}+\zeta_{kl})]
 =E[\xi_{ij}(\xi_{kl}+\zeta_{kl})].
\]
Choose $i$ as in (ii), and construct $Z_i^*$ independently of the
remaining nodes and all dyad shocks. The canonical condition gives
\[
 E[\xi(Z_i^*,Z_j)(\xi_{kl}+\zeta_{kl})]=0.
\]
Consequently, H\"older's inequality gives
\[
\begin{aligned}
 &\left|E[(\xi_{ij}+\zeta_{ij})(\xi_{kl}+\zeta_{kl})]\right|=\left|E[\{\xi(Z_i,Z_j)-\xi(Z_i^*,Z_j)\}
                     (\xi_{kl}+\zeta_{kl})]\right|
 \le C(1+r)^{-\tau}.
\end{aligned}
\]
The same integration and replacement apply to products of four
factors. If no node appears only once, the moment bound applies
directly. The counts in (ii), including the count with one node
fixed for a row sum, are unchanged. This proves all three bounds
for $\bm\xi_{ij}+\bm\zeta_{ij}$.

\end{proof}

\begin{lemma}[Dyadic law of large numbers]
\label{lem:dyadic_lln}
Suppose Assumptions~\ref{ass:representation} and
\ref{ass:mixing} hold. Let
$\bm{\psi}_{ij}=\bm{\psi}(Z_i,Z_j,Q_{ij})$ be a fixed-dimensional
vector or matrix with uniformly bounded $2(\lambda+\delta)$ moments,
for some $\delta>0$, both under the actual spatial distribution
of the two nodes and when they are drawn independently from $F$. Then
\[
 \frac1{M_n}\sum_{1\le i<j\le n}\bm{\psi}_{ij}
 =E_\perp[\bm{\psi}_{12}]+O_p(n^{-1/2}),
 \qquad
 \left\|\frac1{M_n}\sum_{1\le i<j\le n}E[\bm{\psi}_{ij}]
             -E_\perp[\bm{\psi}_{12}]\right\|\le \frac Cn.
\]
\end{lemma}

\begin{proof}
Use the same decay exponent
$\tau=\mathfrak d\{1-1/(\lambda+\delta)\}>8\nu$
as in the proof of Lemma~\ref{lem:spatial_moments}.
We first bound the difference between the expected average and
$E_\perp[\bm{\psi}_{12}]$, and then the variance of the average.

For each pair $(i,j)$, Berbee's lemma gives a copy $Z_i^*$ independent
of $Z_j$ with
$P(Z_i^*\ne Z_i)\le C(1+d_{ij})^{-\mathfrak d}$.
Construct the copy using randomness independent of the $Q$ variables.
Then
\[
 E[\bm{\psi}(Z_i^*,Z_j,Q_{ij})]=E_\perp[\bm{\psi}_{12}].
\]
The two kernels agree when $Z_i^*=Z_i$. H\"older's inequality,
as in \eqref{eq:spatial_coupling_bound}, therefore gives
\[
 \|E[\bm{\psi}_{ij}]-E_\perp[\bm{\psi}_{12}]\|
 \le C P(Z_i^*\ne Z_i)^{1-1/(\lambda+\delta)}
 \le C(1+d_{ij})^{-\tau}.
\]
There are at most $C(r+1)^{\nu-1}$ sites at distance in $[r,r+1)$
from each node. Since $M_n=n(n-1)/2$,
\[
 \left\|\frac1{M_n}\sum_{1\le i<j\le n}E[\bm{\psi}_{ij}]
                       -E_\perp[\bm{\psi}_{12}]\right\|
 \le \frac{Cn}{M_n}\sum_{r=0}^{\infty}(r+1)^{\nu-1-\tau}
 \le \frac Cn.
\]

For the variance, consider scalar coordinates. If $(i,j)$ and $(k,l)$
have no common node, couple $(Z_k,Z_l)$ jointly to a copy independent
of $(Z_i,Z_j)$, preserving the joint distribution of $(Z_k,Z_l)$.
Integrating the independent $Q$ variables first and applying the
same H\"older bound yields
\[
 |\Cov(\psi_{ij},\psi_{kl})|
 \le C\{1+\Delta((i,j),(k,l))\}^{-\tau}.
\]
Here $\Delta$ is the minimum of the four endpoint distances defined
in \eqref{eq:endpoint_distance}. If a node is shared, $\Delta=0$
and the second-moment bound gives the same inequality.
For each fixed $(i,j)$,
\[
 \#\{(k,l)\in\mathcal D_n:
       r\le\Delta((i,j),(k,l))<r+1\}
 \le Cn(r+1)^{\nu-1}:
\]
one endpoint is at distance in $[r,r+1)$ from $i$ or $j$, and the
other has at most $n$ choices. Summing covariances coordinatewise gives
\[
 E\left\|\frac1{M_n}\sum_{1\le i<j\le n}
       \{\bm{\psi}_{ij}-E[\bm{\psi}_{ij}]\}\right\|^2
 \le \frac{CM_n n}{M_n^2}
           \sum_{r=0}^{\infty}(r+1)^{\nu-1-\tau}
 \le \frac Cn.
\]
Chebyshev's inequality and the preceding mean bound prove the result.
\end{proof}

\begin{lemma}[Spatial central limit theorem]
\label{lem:spatial_clt}
Suppose 
Assumptions~\ref{ass:representation} and \ref{ass:mixing} hold.
Let $\bm{g}_i=\bm{g}(Z_i)$ be a fixed-dimensional vector with
$E[\bm{g}_i]=\bm{0}$ and $E\|\bm{g}_i\|^{4(\lambda+\delta)}<\infty$
for some $\delta>0$. Over the expanding rectangles $\Lambda_n$,
\[
 \frac1{\sqrt n}\sum_{i=1}^n\bm{g}_i
 \Rightarrow N(\bm{0},\bm{\Omega}_g),\qquad
 \bm{\Omega}_g
 =\sum_{\bm{h}\in\mathbb Z^\nu}
 E[\bm{g}(Z(\bm{0}))\bm{g}(Z(\bm{h}))^{\top}].
\]
The covariance series converges absolutely; $\bm{\Omega}_g$ may be singular.
\end{lemma}

\begin{proof}
The covariance bound in the proof of Lemma~\ref{lem:spatial_moments}
gives, with $\tau=\mathfrak d\{1-1/(\lambda+\delta)\}>8\nu$,
\[
 \|E[\bm{g}_i\bm{g}_j^{\top}]\|\le C(1+d_{ij})^{-\tau}.
\]
The bound on the number of sites at each distance makes the covariance
series absolutely convergent. For every fixed
$\bm{h}$, the fraction
$|\Lambda_n\cap(\Lambda_n-\bm{h})|/n$ tends to one. Thus stationarity
and dominated convergence give
\[
 \frac1n\Var\left(\sum_{i=1}^n\bm{g}_i\right)
 =\sum_{\bm{h}\in\mathbb Z^\nu}
   \frac{|\Lambda_n\cap(\Lambda_n-\bm{h})|}{n}
   E[\bm{g}(Z(\bm{0}))\bm{g}(Z(\bm{h}))^{\top}]
 \longrightarrow\bm{\Omega}_g.
\]

We next prove asymptotic normality by dividing the node sum into
spatial blocks and then using Berbee's lemma to replace their sums
by independent copies. Set
\[
 p_n=\lfloor n^{1/(2\nu)}\rfloor,
 \qquad r_n=\lfloor n^{1/(4\nu)}\rfloor.
\]
Use all cubes contained in $\Lambda_n$ of the form
\[
 \prod_{k=1}^{\nu}
 \bigl\{a_k(p_n+r_n)+1,\ldots,a_k(p_n+r_n)+p_n\bigr\},
 \qquad a_k=0,1,\ldots.
\]
Each cube contains $p_n$ lattice positions in each coordinate
direction, hence $p_n^\nu$ nodes. The spacing $p_n+r_n$ leaves
$r_n$ omitted positions between consecutive cubes along each
coordinate direction.

Let $T_1,\ldots,T_J$ contain the node indices in these cubes, and
let $T_0$ contain the nodes in the gaps and unused boundary strips.
Since the cubes are disjoint, $J\le n/p_n^\nu$.
The gaps occupy a fraction at most $Cr_n/p_n$ of the rectangle.
The unused boundary strips have width at most
$p_n+r_n\le2p_n$ in each coordinate direction. Consequently,
\[
 \frac{|T_0|}{n}
 \le C\left(
 \underbrace{\frac{r_n}{p_n}}_{\text{gaps}}
 +
 \underbrace{\frac{p_n}{\min_k N_{k,n}}}_{\text{boundary strips}}
 \right)
 \longrightarrow0.
\]
Here $r_n/p_n\to0$, and the bounded ratio of the rectangle's side
lengths gives $\min_kN_{k,n}\asymp n^{1/\nu}$, so
$p_n/\min_kN_{k,n}\to0$.
Lemma~\ref{lem:spatial_moments}(i) therefore gives
\[
 \frac1n E\left\|\sum_{i\in T_0}\bm{g}_i\right\|^2
 \le C\frac{|T_0|}{n}\longrightarrow0.
\]

Write $\bm{U}_j=\sum_{i\in T_j}\bm{g}_i$ for the retained block sums.
Starting with the last block and proceeding backwards, apply Berbee's
lemma to each block and the preceding blocks. This gives mutually
independent copies $\bm{U}_j^*\overset d=\bm{U}_j$ satisfying
\[
 \begin{aligned}
 P\{\text{some }\bm{U}_j^*\ne\bm{U}_j\}
 &\le C\sum_{j=1}^J |T_j|\sum_{k=1}^{j-1}|T_k|\,r_n^{-\mathfrak d}\le Cn^2r_n^{-\mathfrak d}\longrightarrow0.
 \end{aligned}
\]
The last limit follows from $\mathfrak d>8\nu$.
Thus the original and copied block sums agree with probability
tending to one.

It remains to check the covariance limit and Lindeberg's condition
for the independent blocks. Summing over all ordered pairs of distinct
blocks, whose distance is at least $r_n$, the covariance bound gives
\[
 \frac1n\sum_{j=1}^J\sum_{\substack{k=1\\k\ne j}}^J\|\Cov(\bm{U}_j,\bm{U}_k)\|
 \le C\sum_{r=r_n}^{\infty}(r+1)^{\nu-1-\tau}\longrightarrow0.
\]
Together with the negligible omitted sum, this implies
\[
 \frac1n\sum_{j=1}^J\Var(\bm{U}_j^*)
 =\frac1n\sum_{j=1}^J\Var(\bm{U}_j)
 \longrightarrow\bm{\Omega}_g.
\]
By Lemma~\ref{lem:spatial_moments}(i),
$E\|\bm{U}_j\|^4\le Cp_n^{2\nu}$. For every $\varepsilon>0$,
\[
 \begin{aligned}
 \frac1n\sum_{j=1}^J E\left[\|\bm{U}_j^*\|^2
          \ind\{\|\bm{U}_j^*\|>\varepsilon\sqrt n\}\right]
 &\le \frac1{\varepsilon^2 n^2}
          \sum_{j=1}^J E\|\bm{U}_j\|^4\le \frac{C}{\varepsilon^2}\frac{p_n^\nu}{n}
 \longrightarrow0.
 \end{aligned}
\]
Lindeberg-Feller and Cram\'er-Wold now yield
\[
 \frac1{\sqrt n}\sum_{i=1}^n\bm{g}_i
 =\frac1{\sqrt n}\sum_{j=1}^J\bm{U}_j^*+o_p(1)
 \Rightarrow N(\bm{0},\bm{\Omega}_g).
\]
If a fixed contrast has zero limiting variance, its normalized sum
converges to zero in $L^2$. This also covers singular $\bm{\Omega}_g$.
\end{proof}

\subsection{Spatial weights and covariance estimation}

\begin{lemma}[Neighborhood sizes and node weights]
\label{lem:spatial_weights}
Suppose  Assumptions~\ref{ass:representation}
and ~\ref{ass:bandwidth} hold.
Use the neighborhoods $B_\ell$, their sizes $m_\ell$, and the weights
$K^\cap_{rs}$ and $K^\cup_{ij}$ defined in
\eqref{eq:block}, \eqref{eq:block_overlap_kernel}, and
\eqref{eq:dyad_incidence_weight}.
For all sufficiently large $n$,
\[
 cL^\nu\le m_\ell\le CL^\nu,\qquad
 0\le K^\cap_{rs}\le C\ind\{d_{rs}\le2L\},\qquad
 0\le K^\cup_{ij}\le C.
\]
Both node-weight matrices, $[K^\cap_{rs}]_{r,s=1}^n$ and
$[K^{\mathrm H}_{rs}]_{r,s=1}^n$, are positive semidefinite and have
operator norm at most $CL^\nu$.
\end{lemma}

\begin{proof}
Every lattice site in $\Lambda_n$ is occupied by a node. Since
$L/\min_kN_{k,n}\to0$, for large $n$ a cube extending
$\lfloor L/\sqrt\nu\rfloor$ lattice steps from $\bm{s}_\ell$ toward the
farther boundary in each coordinate lies inside both $\Lambda_n$
and the radius-$L$ ball. Its sites give a lower bound for $m_\ell$.
An enclosing cube extending $\lfloor L\rfloor$ steps in both directions
in each coordinate gives the upper bound. Thus, uniformly in $\ell$,
\[
 cL^\nu\le(\lfloor L/\sqrt\nu\rfloor+1)^\nu
 \le m_\ell\le(2\lfloor L\rfloor+1)^\nu\le CL^\nu.
\]
By symmetry of distance, $r\in B_\ell$ if and only if $\ell\in B_r$.
The number of blocks containing node $r$ is therefore exactly
\[
 \#\{\ell:r\in B_\ell\}=m_r\le CL^\nu.
\]
Notice that a block can contain both $r$ and $s$ only if
$d_{rs}\le d_{r\ell}+d_{\ell s}\le2L$.
Using $m_\ell^{-1}\le(cL^\nu)^{-1}$ gives
\[
 \begin{aligned}
 0\le K^\cap_{rs}
 &=\sum_{\ell=1}^n\frac1{m_\ell}\ind\{r,s\in B_\ell\}\le \frac{m_r}{cL^\nu}\ind\{d_{rs}\le2L\}
 \le C\ind\{d_{rs}\le2L\}.
 \end{aligned}
\]
Similarly, the number of blocks containing at least one of $i$
and $j$ is at most $m_i+m_j$, so
\[
 \begin{aligned}
 0\le K^\cup_{ij}
 &\le\sum_{\ell=1}^n\frac1{m_\ell}
       \bigl(\ind\{i\in B_\ell\}+\ind\{j\in B_\ell\}\bigr)\le\frac{m_i+m_j}{cL^\nu}\le C.
 \end{aligned}
\]

For any real numbers $a_1,\ldots,a_n$, rearranging the sums gives
\[
 \begin{aligned}
 \sum_{r,s=1}^n a_rK^\cap_{rs}a_s
 &=\sum_{\ell=1}^n\frac1{m_\ell}
       \sum_{r\in B_\ell}\sum_{s\in B_\ell}a_ra_s=\sum_{\ell=1}^n\frac1{m_\ell}
       \left(\sum_{r\in B_\ell}a_r\right)^2\ge0.
 \end{aligned}
\]
Hence the block-weight matrix is positive semidefinite.
Positive semidefiniteness of $[K^{\mathrm H}_{rs}]_{r,s=1}^n$
follows from the positive-definiteness property of the Wendland
kernel in \eqref{eq:bartlett} for the dimensions considered here.

Finally, let $K_{rs}$ denote either $K^\cap_{rs}$ or $K^{\mathrm H}_{rs}$.
Both matrices are symmetric, their entries are uniformly bounded,
and their entries vanish when $d_{rs}>2L$.
Each row therefore has at most $CL^\nu$ nonzero entries.
For a symmetric matrix, the operator norm is bounded by the
largest absolute row sum, giving
\[
 \bigl\|[K_{rs}]_{r,s=1}^n\bigr\|
 \le\max_r\sum_{s=1}^n|K_{rs}|
 \le C\max_r\#\{s:d_{rs}\le2L\}
 \le CL^\nu.
\]
\end{proof}

\begin{lemma}[Consistency of spatial covariance estimators]
\label{lem:spatial_covariance}
Suppose  Assumptions~\ref{ass:representation}, \ref{ass:mixing}, and ~\ref{ass:bandwidth} hold.
Let
$\bm g_i=\bm g(Z_i)$ be centered and satisfy the moment condition
in Lemma~\ref{lem:spatial_moments}(i). Then
\begin{equation}
 \frac1n\sum_{r,s=1}^nK^\cap_{rs}\bm g_r\bm g_s^{\top}\to^P\bm\Omega_g,
 \qquad
 \frac1n\sum_{r,s=1}^n K^{\mathrm H}_{rs}\bm g_r\bm g_s^{\top}\to^P\bm\Omega_g,
 \label{eq:spatial_hac_lln}
\end{equation}
where $\bm\Omega_g$ is the long-run covariance defined in
Lemma~\ref{lem:spatial_clt}.
\end{lemma}

\begin{proof}
Write $K_{rs}=K^\cap_{rs}$ or $K_{rs}=K^{\mathrm H}_{rs}$;
the proof applies to both choices. By Lemma~\ref{lem:spatial_weights},
these weights are uniformly bounded and zero when $d_{rs}>2L$.

\textit{Expectation.}\quad
Fix $\bm h\in\mathbb Z^\nu$ and consider nodes $i,j$ satisfying
$\bm s_j-\bm s_i=\bm h$. If both nodes are farther than $2L$ from the observation
boundary, every block containing either node is a full lattice ball.
The overlap weight is therefore
\[
 K^\cap_{ij}
 =\frac{\#\{\ell:d_{i\ell}\le L,\ d_{j\ell}\le L\}}
        {\#\{\ell:d_{i\ell}\le L\}}.
\]
The two sets $\{\ell:d_{i\ell}\le L\}$ and
$\{\ell:d_{j\ell}\le L\}$ differ by at most
$C(1+\|\bm h\|)L^{\nu-1}$ sites, so
\[
 |K^\cap_{ij}-1|
 \le\frac{C(1+\|\bm h\|)L^{\nu-1}}{cL^\nu}\longrightarrow0.
\]
For the other weights,
$K^{\mathrm H}_{ij}=K(\|\bm h\|/L)\to1$ because $K$ is continuous at zero
and $K(0)=1$. The fraction of nodes within $2L$ of the boundary is
$O(L/n^{1/\nu})=o(1)$, and $n^{-1}\#\{(i,j):\bm s_j-\bm s_i=\bm h\}\to1$.
Thus, for either choice of $K_{ij}$,
\[
 \frac1n\sum_{\substack{1\le i,j\le n\\\bm s_j-\bm s_i=\bm h}}K_{ij}\longrightarrow1,
 \qquad
 \left|\frac1n\sum_{\substack{1\le i,j\le n\\\bm s_j-\bm s_i=\bm h}}K_{ij}\right|\le C.
\]
Stationarity and absolute convergence of the covariance series in
Lemma~\ref{lem:spatial_clt} give
\[
 \begin{aligned}
 E\left[\frac1n\sum_{i,j=1}^nK_{ij}\bm g_i\bm g_j^{\top}\right]
 &=\sum_{\bm h\in\mathbb Z^\nu}
   \left(\frac1n\sum_{\substack{1\le i,j\le n\\\bm s_j-\bm s_i=\bm h}}K_{ij}\right)
   E[\bm g(Z(\bm0))\bm g(Z(\bm h))^{\top}]\longrightarrow\bm\Omega_g.
 \end{aligned}
\]

\textit{Variance.}\quad
Let $g_{r,a}$ denote coordinate $a$ of $\bm g_r$, and fix coordinates
$a,b$. If $\{r,s\}$ and $\{p,q\}$ are disjoint, apply Berbee's lemma
to $(Z_p,Z_q)$ and $(Z_r,Z_s)$. The copy $(Z_p^*,Z_q^*)$ has the
same joint distribution as $(Z_p,Z_q)$, is independent of $(Z_r,Z_s)$, and
satisfies
\[
 P\{(Z_p^*,Z_q^*)\ne(Z_p,Z_q)\}
 \le C\{1+\Delta((r,s),(p,q))\}^{-\mathfrak d}.
\]
Write $g_{p,a}^*=g_a(Z_p^*)$ and $g_{q,b}^*=g_b(Z_q^*)$.
Independence and preservation of the joint distribution imply
\[
 \Cov(g_{r,a}g_{s,b},g_{p,a}g_{q,b})
 =E\bigl[g_{r,a}g_{s,b}
        (g_{p,a}g_{q,b}-g_{p,a}^*g_{q,b}^*)\bigr].
\]
The difference is zero when the copied pair equals the original pair.
The same H\"older argument as in \eqref{eq:spatial_coupling_bound}
therefore gives, with
$\tau=\mathfrak d\{1-1/(\lambda+\delta)\}>8\nu$,
\[
 |\Cov(g_{r,a}g_{s,b},g_{p,a}g_{q,b})|
 \le C\{1+\Delta((r,s),(p,q))\}^{-\tau}.
\]
Here $\Delta((r,s),(p,q))=\min\{d_{rp},d_{rq},d_{sp},d_{sq}\}$,
also when $r=s$ or $p=q$. If the two sets share a node, $\Delta=0$
and the fourth-moment bound gives the same inequality.

There are at most $CnL^\nu$ pairs $(r,s)$ with $d_{rs}\le2L$.
For each such pair, consider $(p,q)$ satisfying
\[
 d_{pq}\le2L,\qquad h\le\Delta((r,s),(p,q))<h+1.
\]
One of $p,q$ is at distance in $[h,h+1)$ from $r$ or $s$, giving
at most $C(h+1)^{\nu-1}$ choices. The other lies within $2L$ of that
endpoint, giving at most $CL^\nu$ choices. Consequently,
\[
 \begin{aligned}
 \Var\left(\frac1n\sum_{r,s=1}^nK_{rs}g_{r,a}g_{s,b}\right)
 &\le\frac C{n^2}
       \sum_{\substack{1\le r,s\le n\\d_{rs}\le2L}}
       \sum_{\substack{1\le p,q\le n\\d_{pq}\le2L}}
       \{1+\Delta((r,s),(p,q))\}^{-\tau}\\
 &\le\frac C{n^2}(nL^\nu)L^\nu
       \sum_{h=0}^{\infty}(h+1)^{\nu-1-\tau}
 \le C\frac{L^{2\nu}}n\longrightarrow0.
 \end{aligned}
\]
The series is finite because $\tau>\nu$, and the bandwidth assumption
implies $L^{2\nu}/n\to0$. This variance bound and the expectation
limit prove \eqref{eq:spatial_hac_lln} coordinatewise.
\end{proof}

\begin{lemma}[Dyad deletion weights]
\label{lem:dyad_deletion_weights}
Under the conditions of Lemma~\ref{lem:spatial_weights}, the joint
deletion weight $K^\cup_{ij,pq}$ defined in \eqref{eq: Kupijpq}
satisfies
\[
 \begin{aligned}
 0\le K^\cup_{ij,pq}
 &\le C\ind\{\Delta((i,j),(p,q))\le2L\},\\
 \sum_{1\le i<j\le n}\sum_{1\le p<q\le n}
       \ind\{\Delta((i,j),(p,q))\le2L\}
 &\le Cn^3L^\nu.
 \end{aligned}
\]
Finally, $K^\cup_{ij}=2$ except for at most
$C\{n^2L/n^{1/\nu}+nL^\nu\}=o(n^2)$ dyads.
\end{lemma}

\begin{proof}
If a block deletes both $(i,j)$ and $(p,q)$, it contains at least
one endpoint of each pair. These two endpoints are at most $2L$
apart, so $\Delta((i,j),(p,q))\le2L$. Also, when $(p,q)=(i,j)$, we have $K^\cup_{ij,pq}=K^\cup_{ij}$,
which is bounded by Lemma~\ref{lem:spatial_weights}. Thus
\[
 0\le K^\cup_{ij,pq}
 \le K^\cup_{ij}\ind\{\Delta((i,j),(p,q))\le2L\}
 \le C\ind\{\Delta((i,j),(p,q))\le2L\}.
\]

For each fixed $(i,j)$, a dyad $(p,q)$ at endpoint distance at most
$2L$ has one endpoint within $2L$ of $i$ or $j$ ($\le CL^\nu$
choices), and its other endpoint has at most $n$ choices. Therefore
\[
 \sum_{1\le i<j\le n}\sum_{1\le p<q\le n}\ind\{\Delta((i,j),(p,q))\le2L\}
 \le M_n\,CnL^\nu\le Cn^3L^\nu.
\]

To determine when $K^\cup_{ij}=2$, first note that counting
blocks containing either endpoint gives
\[
 K^\cup_{ij}=K^\cap_{ii}+K^\cap_{jj}-K^\cap_{ij}.
\]
The subtraction removes blocks containing both endpoints,
which are included in each of the first two terms.
Suppose a node $i$ is farther than $2L$ from the observation
boundary. Every block containing $i$ is then a full ball, and the
number of such blocks equals the common size of those blocks. Thus
\[
 K^\cap_{ii}=\sum_{\substack{\ell=1\\i\in B_\ell}}^n\frac1{m_\ell}=1.
\]
If $j$ is also farther than $2L$ from the boundary and $d_{ij}>2L$,
no block contains both nodes, so $K^\cap_{ij}=0$ and
\[
 K^\cup_{ij}=K^\cap_{ii}+K^\cap_{jj}-K^\cap_{ij}=1+1-0=2.
\]
Exceptions can only involve a boundary node or a pair with
$d_{ij}\le2L$. There are at most $CnL/n^{1/\nu}$ boundary nodes
and $CnL^\nu$ nearby pairs. The number of exceptions is therefore
at most
\[
 n\left(\frac{CnL}{n^{1/\nu}}\right)+CnL^\nu=o(n^2).
\]
\end{proof}

\subsection{Uniform bounds for deleted samples}

\begin{lemma}[Uniform deleted-block LLN for dyadic sample averages]
\label{lem:uniform_deleted_block_lln}
Under Assumptions~\ref{ass:representation}-\ref{ass:projection},
$\bm{Q}_n\to^P \bm{Q}$. If Assumption~\ref{ass:bandwidth} also holds, put
\[
 M_{n,-\ell}=\frac{(n-m_\ell)(n-m_\ell-1)}2,
 \qquad
 \bm{Q}_{n,-\ell}=\frac1{M_{n,-\ell}}
       \sum_{(i,j)\in\mathcal D_{n,-\ell}}\bm{x}_{ij}\bm{x}_{ij}^{\top}.
\]
Then
 $\max_{1\le\ell\le n}\|\bm{Q}_{n,-\ell}-\bm{Q}\|=o_p(1).$
Thus the full and deleted-sample Gram matrices are simultaneously
invertible with probability tending to one. 
\end{lemma}

\begin{proof}
Applying Lemma~\ref{lem:dyadic_lln} to
$\bm{\psi}_{ij}=\bm{x}_{ij}\bm{x}_{ij}^{\top}$ gives
\[
 \bm Q_n
 =\frac1{M_n}\sum_{1\le i<j\le n}
       \bm x_{ij}\bm x_{ij}^{\top}
 \to^P\bm Q.
\]
We next show that deleting any block changes this normalized
matrix by $o_p(1)$, uniformly over the blocks.
For
$\bm H_\ell=\sum_{(i,j)\in\mathcal A_\ell}
\bm x_{ij}\bm x_{ij}^{\top}$,
Cauchy-Schwarz and the uniform fourth-moment bound give
\[
 \begin{aligned}
 E\|\bm H_\ell\|^2
 &\le E\left(\sum_{(i,j)\in\mathcal A_\ell}
                    \|\bm x_{ij}\|^2\right)^2\le |\mathcal A_\ell|
       \sum_{(i,j)\in\mathcal A_\ell}E\|\bm x_{ij}\|^4
 \le C|\mathcal A_\ell|^2
 \le Cn^2m_\ell^2.
 \end{aligned}
\]
Here $|\mathcal A_\ell|\le nm_\ell$, since each of the $m_\ell$
nodes in the block belongs to at most $n-1$ dyads.
The union bound and Markov's inequality therefore imply
\[
 \begin{aligned}
 P\left\{\max_{1\le\ell\le n}\|\bm H_\ell\|>\varepsilon n^2\right\}
 &\le\sum_{\ell=1}^n
       \frac{E\|\bm H_\ell\|^2}{\varepsilon^2n^4}\le\frac{C}{\varepsilon^2n^2}
       \sum_{\ell=1}^n m_\ell^2
 \le\frac{CL^{2\nu}}{\varepsilon^2n}\to0,
 \end{aligned}
\]
where the last inequality uses the uniform bound
$m_\ell\le CL^\nu$ from Lemma~\ref{lem:spatial_weights}.
Also $\max_\ell|M_n/M_{n,-\ell}-1|=O(nL^\nu/n^2)=o(1)$.
The identity
\[
 \bm{Q}_{n,-\ell}-\bm{Q}
 =\frac{M_n}{M_{n,-\ell}}(\bm{Q}_n-\bm{Q})
   +\left(\frac{M_n}{M_{n,-\ell}}-1\right)\bm{Q}
   -\frac{\bm{H}_\ell}{M_{n,-\ell}}
\]
proves uniform convergence. Positive definiteness of $\bm{Q}$ indicates that the matrices are invertible with probability tending to one.
\end{proof}

\begin{lemma}[Uniform rate of the delete-block estimators]
\label{lem:delete_block_beta_rate}
Under the assumptions of Theorem~\ref{thm:hac_jk},
\[
 \max_{1\le\ell\le n}\|\widetilde{\bm{\beta}}_{(-\ell)}-\bm{\beta}\|
 =\begin{cases}
 O_p(n^{-1/2}),&\text{case (i)},\\
 O_p(n^{-1}),&\text{case (ii)}.
 \end{cases}
\]
\end{lemma}

\begin{proof}
On the simultaneous-invertibility event from the preceding lemma,
\[
 \max_\ell\|\widetilde{\bm{\beta}}_{(-\ell)}-\bm{\beta}\|
 \le O_p(n^{-2})\left\{
 \Big\|\sum_{1\le i<j\le n}\bm{s}_{ij}\Big\|
 +\max_\ell\Big\|\sum_{(i,j)\in\mathcal A_\ell}\bm{s}_{ij}\Big\|\right\}.
\]
In case (i), the full score is $O_p(n^{3/2})$. The node part of
the deleted score has the exact form
\[
 (n-1-m_\ell)\sum_{i\in B_\ell}\bm{\gamma}_i
                  +m_\ell\sum_{i=1}^n\bm{\gamma}_i.
\]
By Lemma~\ref{lem:spatial_moments}(i),
\[
 P\left\{\max_\ell\Big\|\sum_{i\in B_\ell}\bm{\gamma}_i\Big\|
                     >\varepsilon\sqrt n\right\}
 \le\frac{CL^{2\nu}}{\varepsilon^4 n}\to0.
\]
For $\bm{R}_{ij}=\bm{\xi}_{ij}+\bm{\zeta}_{ij}$, the absolute second-moment
sum gives $E\|\sum_{(i,j)\in\mathcal A_\ell}\bm{R}_{ij}\|^2\le Cn^2$.
A union bound therefore gives
$\max_\ell\|\sum_{(i,j)\in\mathcal A_\ell}\bm{R}_{ij}\|=O_p(n^{3/2})$.
This proves case (i).

In case (ii), $\bm{s}_{ij}=\bm{\zeta}_{ij}$. Conditional on
$\mathcal Z$, these vectors are independent and centered.
Conditional Rosenthal's inequality, applied coordinatewise, gives
\[
\begin{aligned}
 E\Big\|\sum_{(i,j)\in\mathcal A_\ell}\bm{\zeta}_{ij}\Big\|^4
 &\le C\sum_{(i,j)\in\mathcal A_\ell}E\|\bm{\zeta}_{ij}\|^4+C E\left[
   \left\{\sum_{(i,j)\in\mathcal A_\ell}
       E(\|\bm{\zeta}_{ij}\|^2\mid\mathcal Z)\right\}^2
   \right]\le C|\mathcal A_\ell|+C|\mathcal A_\ell|^2
 \le Cn^2m_\ell^2.
\end{aligned}
\]
Here, the second inequality follows from Cauchy-Schwarz, conditional
Jensen's inequality, and the uniform fourth-moment bound:
\[
 E\left[
   \left\{\sum_{(i,j)\in\mathcal A_\ell}
       E(\|\bm{\zeta}_{ij}\|^2\mid\mathcal Z)\right\}^2
   \right]
 \le |\mathcal A_\ell|
       \sum_{(i,j)\in\mathcal A_\ell}E\|\bm{\zeta}_{ij}\|^4
 \le C|\mathcal A_\ell|^2.
\]
The union bound and Markov's inequality now yield
\[
\begin{aligned}
 P\left\{\max_{1\le\ell\le n}
     \Big\|\sum_{(i,j)\in\mathcal A_\ell}\bm{s}_{ij}\Big\|
     >\varepsilon n\right\}
 &\le \frac{1}{\varepsilon^4n^4}
       \sum_{\ell=1}^n
       E\Big\|\sum_{(i,j)\in\mathcal A_\ell}
                    \bm{s}_{ij}\Big\|^4\le \frac{C}{\varepsilon^4n^2}\sum_{\ell=1}^n m_\ell^2
 \le \frac{CL^{2\nu}}{\varepsilon^4n}\to0.
\end{aligned}
\]
Finally, conditional independence and centering also give
\[
 E\Big\|\sum_{1\le i<j\le n}\bm{s}_{ij}\Big\|^2
 =\sum_{1\le i<j\le n}E\|\bm{\zeta}_{ij}\|^2
 \le Cn^2.
\]
Thus the full score is $O_p(n)$, while the largest deleted
score is $o_p(n)$, establishing case (ii).
\end{proof}

\subsection{Node fixed effect moments}
\begin{lemma}[Moments of within-node sums]
\label{lem:fe_rows}
Under Assumption~\ref{ass:fe}, let $v_{ij}^\circ$ be a coordinate of
$\bm{x}_{ij}^\circ$ or $u_{ij}^\circ$. For a deterministic set $T$ of
$m$ retained nodes, $i\in T$, and $p=1,\ldots,4$,
\begin{equation}
 E\left|\sum_{j\in T\setminus\{i\}}v_{ij}^\circ\right|^{2p}
 \le C m^p.
 \label{eq:fe_row_moments}
\end{equation}
The constants are uniform over $i,T$. Let $\bm v_T^\circ$ stack
$v_{ij}^\circ$ over $i<j$ with $i,j\in T$, and use the projection
matrix $\bm P_{D,T}$ defined in \eqref{eq:fe_retained_projection}.
The fitted node-effect component satisfies
\[
 E\|\bm{P}_{D,T}\bm{v}_T^\circ\|^{2p}\le C m^p,
 \qquad p=1,\ldots,4.
\]
The same statements hold coordinatewise for vector arrays.
\end{lemma}

\begin{proof}
Fix $i\in T$. Expanding the even power gives
\[
 E\left|\sum_{j\in T\setminus\{i\}}v_{ij}^\circ\right|^{2p}
 =\sum_{\substack{j_1,\ldots,j_{2p}\in T\setminus\{i\}}}
       E[v_{ij_1}^\circ\cdots v_{ij_{2p}}^\circ].
\]
Each summation index ranges separately over $T\setminus\{i\}$;
the indices may equal one another. We bound the sum of the absolute
expectations on the right.

If every index among $j_1,\ldots,j_{2p}$ is used at least twice,
there are at most $p$ distinct indices. There are therefore at most
$Cm^p$ such terms. Since $v_{ij}^\circ$ is bounded, their total
absolute contribution is at most $Cm^p$.

Now consider a term with an index used only once. Among these indices,
choose the one farthest from its nearest other index, including the
fixed index $i$. Relabel it $j_1$, and denote this distance by
\[
 r=\min\{d_{j_1i},d_{j_1j_2},\ldots,d_{j_1j_{2p}}\}.
\]
By Berbee's coupling lemma and the geometric $\beta$-mixing assumption,
we can construct $Z_{j_1}^*$ with the same marginal distribution as
$Z_{j_1}$, independent of the remaining node variables in this term,
such that $P(Z_{j_1}^*\ne Z_{j_1})\le C\rho^r$.
Write $v_{ij_1}^{\circ *}$ for the resulting dyad variable, keeping
its original independent dyad shock. By the definition of population
demeaning,
\[
 E_\perp[v_{ij_1}^\circ\mid Z_i=z]=0.
\]
The index $j_1$ and the shock for dyad $(i,j_1)$ enter no other
factor. Integrating them out therefore gives
$E[v_{ij_1}^{\circ *}v_{ij_2}^\circ\cdots v_{ij_{2p}}^\circ]=0$.
Consequently,
\[
 \begin{aligned}
 |E[v_{ij_1}^\circ\cdots v_{ij_{2p}}^\circ]|
 &=|E[(v_{ij_1}^\circ-v_{ij_1}^{\circ *})
                   v_{ij_2}^\circ\cdots v_{ij_{2p}}^\circ]|\le C P(Z_{j_1}^*\ne Z_{j_1})\le C\rho^r.
 \end{aligned}
\]

It remains to count these terms. Fix an integer $h\ge1$ and consider
$h\le r<h+1$. By the choice of $r$, every index used once is within
distance $h+1$ of another index. Put indices connected by steps of
length less than $h+1$ in the same group. Each group not containing
$i$ must account for at least two of the $2p$ factors. There are at
most $p$ such groups, so choosing one node in each has at most $m^p$
possibilities. Every remaining node is within $2p(h+1)$ of one of
these chosen nodes or of $i$. The lattice bound gives at most
$C(h+1)^\nu$ choices for each remaining node. Thus the number of
terms at this distance is at most $Cm^p(h+1)^{2p\nu}$, and
\[
 E\left|\sum_{j\in T\setminus\{i\}}v_{ij}^\circ\right|^{2p}
 \le Cm^p+Cm^p\sum_{h=1}^{\infty}(h+1)^{2p\nu}\rho^h
 \le Cm^p.
\]
The series is finite because $0<\rho<1$. Its bound is independent
of the location of $i$ and the shape of $T$.

For the projection bound, the entry of
$\bm D_T^\top\bm v_T^\circ$ corresponding to node $i\in T$ is
$\sum_{j\in T\setminus\{i\}}v_{ij}^\circ$.
Since $\bm P_{D,T}$ is symmetric and idempotent, its definition
in \eqref{eq:fe_retained_projection} and \eqref{eq:fe_projection}
with $n=m$ give, for $m>2$,
\[
 \begin{aligned}
 \|\bm P_{D,T}\bm v_T^\circ\|^2
 &= (\bm D_T^\top\bm v_T^\circ)^\top
       (\bm D_T^\top\bm D_T)^{-1}
       (\bm D_T^\top\bm v_T^\circ)\\
 &\le\frac1{m-2}\sum_{i\in T}
                 \left(\sum_{j\in T\setminus\{i\}}v_{ij}^\circ\right)^2,\\
 E\|\bm P_{D,T}\bm v_T^\circ\|^{2p}
 &\le\frac{m^{p-1}}{(m-2)^p}\sum_{i\in T}
       E\left|\sum_{j\in T\setminus\{i\}}v_{ij}^\circ\right|^{2p}
 \le Cm^p.
 \end{aligned}
\]
The second inequality uses $(\sum_{i=1}^m a_i)^p\le
m^{p-1}\sum_{i=1}^m a_i^p$ for nonnegative $a_i$.
For $m\le2$, boundedness and $\|\bm P_{D,T}\|\le1$ give the result directly.
\end{proof}

\section{Implementation Details}

\label{app:implementation}

The following steps implement the spatial estimators and MJS interval.

\paragraph{Spatial bandwidth and neighborhoods.}
Normalize coordinates by the fixed lattice spacing. We select
the bandwidth using spatial correlations in the estimated
node scores \(\widehat{\bm G}_i\) from \eqref{eq: twoterm}.
Consider the candidate bandwidths
\[
 \mathcal L_n
 =\{c\,n^{1/(4\nu)}:c\in\{1,1.5,2,3\}\}.
\]
Each candidate diverges and satisfies
\(L^{3\nu}/n\to0\). The data determine which candidate is used.

Center the node scores as
\(\widetilde{\bm G}_i
=\widehat{\bm G}_i-n^{-1}\sum_{r=1}^n\widehat{\bm G}_r\).
For each distance band \(h=1,\ldots,H_n\), where
\(H_n=\lceil\max\mathcal L_n\rceil+1\), define
\[
 \mathcal P_h
 =\{(i,j):1\le i<j\le n,\ h-1<d_{ij}\le h\},
 \qquad
 \widehat\rho(h)
 =\max_{1\le k\le K}
 \frac{2\left|\sum_{(i,j)\in\mathcal P_h}
                \widetilde G_{i,k}\widetilde G_{j,k}\right|}
      {\sum_{(i,j)\in\mathcal P_h}
                (\widetilde G_{i,k}^2+\widetilde G_{j,k}^2)}.
\]
Set a ratio to zero when its denominator is zero.
This measure pools node pairs at similar geographic distances
and is unchanged by relabeling the nodes.

Let \(\widehat r\) be the first distance band for which
\(\widehat\rho(h)\) and \(\widehat\rho(h+1)\) are both below
\(\sqrt{\log n/n}\); set \(\widehat r=H_n\) if no such band
exists. Select \(\widehat L\) as the smallest member of
\(\mathcal L_n\) at least as large as \(\widehat r\), using
the largest candidate if none qualifies. Thus, correlations
that remain above the threshold over longer distances lead
to a larger selected bandwidth. The threshold guides bandwidth
selection and is not a formal test of spatial independence.

Using this full-sample choice throughout, form
\(B_\ell=\{i:d_{i\ell}\le\widehat L\}\) and retain
\(m_\ell=|B_\ell|\). The DN estimator uses the same
\(\widehat L\), although its Wendland weights differ from
the jackknife overlap weights.

\paragraph{Spatial deletion and incidence correction.}
Let \texttt{id1} and \texttt{id2} contain the two endpoints of
each dyad, and let \texttt{loc} contain the node coordinates.
The following MATLAB code computes the OLS jackknife. The
vectors \texttt{y}, \texttt{id1}, and \texttt{id2} are columns.
The vector \texttt{Kcup} stores $K^\cup_{ij}$ in the same dyad order as \texttt{y}.
\begin{verbatim}
n = size(loc,1); K = size(X,2);
dist = sqrt(sum((reshape(loc,n,1,[]) ...
              - reshape(loc,1,n,[])).^2,3));
A = X'*X; Ainv = pinv(A);
bhat = Ainv*(X'*y);
score = X.*(y-X*bhat);
Bdel = zeros(n,K); m = zeros(n,1);
Kcup = zeros(size(y));
for ell = 1:n
    inside = dist(ell,:) <= L;
    m(ell) = sum(inside);
    touch = inside(id1) | inside(id2);
    touch = touch(:); keep = ~touch;
    Xk = X(keep,:); yk = y(keep);
    Bdel(ell,:) = (pinv(Xk'*Xk)*(Xk'*yk))';
    Kcup(touch) = Kcup(touch) + 1/m(ell);
end
D = (Bdel-bhat')./sqrt(m);
V0 = D'*D;
correction = score'*((Kcup-1).*score);
Vjk = V0-Ainv*correction*Ainv;
Vjk = (Vjk+Vjk')/2;
\end{verbatim}
The Moore-Penrose inverse supplies a finite-sample definition;
rank and retained sample size should also be reported. The
theory makes singular designs asymptotically negligible.

\paragraph{Fixed effects and MJS inference.}
For fixed effects, use \eqref{eq:fe_within} separately on the
full graph and on every retained graph, with its own node
count and row sums. Use the full-sample within Gram matrix and
within residual scores in the incidence correction. Keeping
full-sample transformed data fixed during deletions does not
implement the estimator in Theorem~\ref{thm:fe_gaussian}.

For subsampling, choose $q=\lfloor n^{\kappa/\nu}\rfloor$,
$0<\kappa<1$. Enumerate all contained
lattice cubes of this side length, retain only dyads with both
endpoints inside each cube, and refit the slope and any node
effects. In finite samples, choose cubes large enough to
identify the regression. Compute the empirical $1-\alpha$
quantile of $q^\nu |\bm{a}^{\top}(\widehat{\bm{\beta}}_{\ell,q}-\widehat{\bm{\beta}})|$,
divide it by $n$, and use this value as the subsampling term in
\eqref{eq:conservative_interval} to construct $C_{n,\alpha}$.
The theory does not specify
a finite-sample optimal cube size.

\clearpage
\section{Simulation Designs and Additional Results}
\label{app:simulation_results}

\subsection{Data-generating processes}
\label{app:simulation_dgp}
Set \(x_{ij,1}=1\) and \(\beta_k=1\) for all \(k\).
The remaining regressors and dyad shocks are mutually independent,
with
\[
 P(x_{ij,k}=1)=P(x_{ij,k}=-1)=\frac12,
 \qquad
 P(\epsilon_{ij}=1)=P(\epsilon_{ij}=-1)=\frac12,
 \qquad k=2,\ldots,K.
\]
Node effects enter the disturbance, rather than the regressors.

To generate spatially dependent node effects, draw the \(n\)-vectors
\[
 \bm G_1^u,\ldots,\bm G_K^u,\quad
 \bm G_1^v,\ldots,\bm G_K^v,\quad
 \bm G^\sigma
 \ \stackrel{\mathrm{ind}}{\sim}
 N\!\left(
 \bm0,\,
 \left[\rho^{\|\bm s_i-\bm s_j\|_1}\right]_{i,j=1}^n
 \right),
\]
independently of the regressors and dyad shocks. Define their signs by
\[
 Z_{i,k}^u=\operatorname{sgn}(G_{i,k}^u),\qquad
 Z_{i,k}^v=\operatorname{sgn}(G_{i,k}^v),\qquad
 Z_i^\sigma=\operatorname{sgn}(G_i^\sigma).
\]
Thus each node effect takes values \(-1\) and \(1\) with equal
probability, has mean zero and variance one, and is bounded.
This construction satisfies the moment requirements without
changing the marginal distribution as $\rho$ varies. Different
Gaussian vectors are independent, whereas entries within each
vector are spatially correlated. For example,
\[
 E[Z_{i,k}^u Z_{j,k}^u]
 =\frac{2}{\pi}\arcsin\!\left(
      \rho^{\|\bm s_i-\bm s_j\|_1}\right).
\]
Hence $\rho=0$ gives independent node effects, while a larger
$\rho$ increases their spatial correlation. The covariance
construction uses the sum of coordinate differences;
the estimation weights continue to use Euclidean distances.

The variable $Z_i^\sigma$ determines conditional heteroskedasticity.
Specifically, set
\[
 \sigma_{ij}
 =\frac{1+\tfrac14(Z_i^\sigma+Z_j^\sigma)}{\sqrt{9/8}}.
\]
The numerator makes the disturbance scale depend symmetrically
on both endpoint nodes. It equals $1/2$, $1$, or $3/2$, so
the scale is always positive and bounded. Since $\epsilon_{ij}$
has mean zero and variance one and is independent of the node
variables and regressors,
\[
 E[\sigma_{ij}\epsilon_{ij}\mid\mathcal Z,\bm X]=0,
 \qquad
 \Var(\sigma_{ij}\epsilon_{ij}\mid\mathcal Z,\bm X)
 =\sigma_{ij}^2\in\{2/9,\,8/9,\,2\}.
\]
Thus the design permits different conditional variances without
introducing a conditional mean. The denominator normalizes the
variance under independent node draws because
\[
 E_\perp\!\left[
  \left\{1+\tfrac14(Z_i^\sigma+Z_j^\sigma)\right\}^2
 \right]=1+\frac{2}{16}=\frac98,
 \qquad E_\perp[\sigma_{ij}^2]=1.
\]
Under the actual spatial distribution,
$E[\sigma_{ij}^2]=1+E[Z_i^\sigma Z_j^\sigma]/9$;
the normalization therefore does not impose identical
unconditional variances at all endpoint distances. Even in the
degenerate Gaussian design, spatial dependence remains in the
conditional variances, although distinct dyad scores are
uncorrelated because the dyad shocks are independent and centered.

We generate \(y_{ij}=\bm x_{ij}^{\top}\bm\beta+u_{ij}\), where
\begin{equation}
 u_{ij}=\sigma_{ij}\epsilon_{ij}+
 \begin{cases}
 \displaystyle
 \frac1{\sqrt K}\sum_{k=1}^K
 x_{ij,k}(Z_{i,k}^u+Z_{j,k}^u),
 & \text{nondegenerate Gaussian},\\[1.2ex]
 0,
 & \text{degenerate Gaussian},\\[1.2ex]
 \displaystyle
 \frac{\omega}{\sqrt{2K}}\sum_{k=1}^K
 x_{ij,k}(Z_{i,k}^uZ_{j,k}^v+Z_{i,k}^vZ_{j,k}^u),
 & \text{degenerate non-Gaussian}.
 \end{cases}
 \label{eq:simulation_dgp}
\end{equation}
The $1/\sqrt K$ normalization prevents the variance of the node
contribution from increasing with the number of regressors.
The interaction has two symmetric products; the additional
$1/\sqrt2$ gives it variance $\omega^2$ under independent
node draws. In particular, the three designs satisfy
\[
 E_\perp[u_{ij}^2]
 =3,\qquad 1,\qquad 1+\omega^2,
\]
respectively, for every $K$. The first design generates
shared-node dependence through additive node effects. The third
uses two independent sets of centered node effects, so its
interaction has zero expectation after averaging over either
endpoint independently. Their independence also ensures the
actual score moment $E[\bm x_{ij}u_{ij}]=\bm0$ at every
endpoint distance.

Since $E[\bm x_{ij}\bm x_{ij}^{\top}\mid\mathcal Z]=\bm I_K$
and the dyad shock is centered, the score decomposition gives
\[
 (\gamma_{i,k},\xi_{ij,k})
 =
 \begin{cases}
 (Z_{i,k}^u/\sqrt K,\,0),
 & \text{nondegenerate Gaussian},\\[0.5ex]
 (0,\,0),
 & \text{degenerate Gaussian},\\[0.5ex]
 \left(0,\,
 \dfrac{\omega}{\sqrt{2K}}
 (Z_{i,k}^uZ_{j,k}^v+Z_{i,k}^vZ_{j,k}^u)\right),
 & \text{degenerate non-Gaussian}.
 \end{cases}
\]
The third design has a non-Gaussian limit when \(\omega>0\);
setting \(\omega=0\) recovers the degenerate Gaussian design.
Unless varied, \(n=36\), \(K=10\), \(\rho=0.25\), and \(\omega=1\).

For fixed effects, we omit the regressor intercept and add
$\alpha_i+\alpha_j$ to the outcome, where
$\alpha_i=2Z_i^\sigma+0.2s_{i,1}$; $K$ then counts slopes. The node
effects are refitted in every full, deleted, and cube sample.
The degenerate Gaussian disturbance satisfies
$E[u_{ij}\mid\mathcal Z,\bm X]=0$, as required by
Theorem~\ref{thm:fe_gaussian}(ii).

\subsection{Implementation}

TW clusters by the first and second recorded endpoints after a
fixed random permutation of node labels. Spatial weights and
deletions use Euclidean distances. Bandwidths follow
Appendix~\ref{app:implementation}, with candidates excluded when
a deletion leaves too few observations to fit the regression.
MJS uses all contained cubes, starting with
$q=\max\{2,\lfloor n^{0.65/2}\rfloor\}$ and increasing $q$
only when needed to fit the cube regressions. No degrees-of-freedom
multipliers are used; negative estimated contrast variances are
set to zero when forming Gaussian intervals.

\subsection{Additional results}

Tables~\ref{tab:sim_fast}-\ref{tab:sim_fe_fast} report the
degenerate Gaussian and fixed-effect Gaussian designs.
At $n=100$, JK-DN rejection frequencies in the degenerate
Gaussian designs are 0.057 without fixed effects and 0.051 with
fixed effects. The corresponding MJS frequencies are 0.028
and 0.021, showing finite-sample conservatism despite pointwise
asymptotically exact coverage. White performs well in these
designs because the scores are centered and independent
conditional on the node field. In the nondegenerate fixed-effect
design, MJS rejects with frequency 0.079 at $n=100$, compared
with 0.163 for Dyadic. The uncorrected estimators are included
to show the effect of the duplication correction.
Table~\ref{tab:sim_nongaussian_more} provides further
non-Gaussian comparisons. All entries are based on 10,000
replications with no failed fits; the Monte Carlo standard
error of a rejection frequency of 0.05 is approximately 0.0022.

\begin{table}[p]
\centering
\caption{Degenerate Gaussian design: rejection frequencies}
\label{tab:sim_fast}
\begingroup
\footnotesize
\setlength{\tabcolsep}{4pt}
\renewcommand{\arraystretch}{1.05}
\begin{tabular}{lrrrrrrrr}
\toprule
Setting & White & TW & Dyadic & DN & JK-DN & MJS & DN$_0$ & JK$_0$ \\
\midrule
\multicolumn{9}{l}{\textit{Panel A: Varying $n$}}\\
16 & 0.067 & 0.166 & 0.235 & 0.238 & 0.014 & 0.011 & 0.022 & 0.003 \\
25 & 0.054 & 0.117 & 0.156 & 0.160 & 0.020 & 0.011 & 0.014 & 0.003 \\
36 & 0.054 & 0.089 & 0.115 & 0.123 & 0.032 & 0.009 & 0.011 & 0.004 \\
49 & 0.050 & 0.076 & 0.096 & 0.102 & 0.040 & 0.012 & 0.009 & 0.005 \\
64 & 0.052 & 0.070 & 0.079 & 0.085 & 0.050 & 0.010 & 0.007 & 0.005 \\
100 & 0.054 & 0.065 & 0.073 & 0.075 & 0.057 & 0.028 & 0.007 & 0.007 \\
\midrule
\multicolumn{9}{l}{\textit{Panel B: Varying $K$}}\\
2 & 0.052 & 0.091 & 0.117 & 0.122 & 0.039 & 0.016 & 0.009 & 0.005 \\
4 & 0.055 & 0.091 & 0.121 & 0.128 & 0.040 & 0.015 & 0.011 & 0.005 \\
6 & 0.049 & 0.082 & 0.113 & 0.120 & 0.035 & 0.012 & 0.009 & 0.004 \\
8 & 0.050 & 0.091 & 0.117 & 0.127 & 0.033 & 0.011 & 0.010 & 0.006 \\
10 & 0.054 & 0.089 & 0.115 & 0.123 & 0.032 & 0.009 & 0.011 & 0.004 \\
12 & 0.054 & 0.092 & 0.121 & 0.132 & 0.030 & 0.009 & 0.011 & 0.004 \\
14 & 0.054 & 0.090 & 0.124 & 0.135 & 0.028 & 0.005 & 0.012 & 0.004 \\
16 & 0.051 & 0.088 & 0.117 & 0.128 & 0.025 & 0.006 & 0.010 & 0.004 \\
18 & 0.052 & 0.089 & 0.120 & 0.136 & 0.022 & 0.004 & 0.010 & 0.003 \\
20 & 0.057 & 0.091 & 0.119 & 0.129 & 0.025 & 0.016 & 0.012 & 0.004 \\
\midrule
\multicolumn{9}{l}{\textit{Panel C: Varying $\rho$}}\\
0 & 0.048 & 0.088 & 0.121 & 0.134 & 0.033 & 0.009 & 0.011 & 0.004 \\
0.1 & 0.053 & 0.083 & 0.115 & 0.125 & 0.031 & 0.010 & 0.011 & 0.004 \\
0.2 & 0.057 & 0.095 & 0.121 & 0.128 & 0.034 & 0.009 & 0.009 & 0.005 \\
0.3 & 0.053 & 0.084 & 0.116 & 0.125 & 0.032 & 0.010 & 0.010 & 0.004 \\
0.4 & 0.051 & 0.091 & 0.119 & 0.128 & 0.032 & 0.007 & 0.009 & 0.003 \\
0.5 & 0.050 & 0.088 & 0.114 & 0.123 & 0.026 & 0.007 & 0.010 & 0.004 \\
\bottomrule
\end{tabular}
\endgroup
\begin{flushleft}\footnotesize
Notes: Rejection frequencies at the 5\% level, based on 10,000 replications. Unless varied, $n=36$, $K=10$, $\rho=0.25$, and $\omega=1$. DN denotes the node-score implementation described in Appendix~\ref{app:simulation_results}; JK-DN denotes the corrected spatial jackknife. DN$_0$ and JK$_0$ omit their respective duplication corrections. \end{flushleft}
\end{table}

\begin{table}[p]
\centering
\caption{Nondegenerate Gaussian design with node fixed effects}
\label{tab:sim_fe_slow}
\begingroup
\footnotesize
\setlength{\tabcolsep}{4pt}
\renewcommand{\arraystretch}{1.05}
\begin{tabular}{lrrrrrrrr}
\toprule
Setting & White & TW & Dyadic & DN & JK-DN & MJS & DN$_0$ & JK$_0$ \\
\midrule
\multicolumn{9}{l}{\textit{Panel A: Varying $n$}}\\
16 & 0.272 & 0.265 & 0.259 & 0.261 & 0.028 & 0.027 & 0.114 & 0.017 \\
25 & 0.343 & 0.239 & 0.204 & 0.203 & 0.050 & 0.045 & 0.115 & 0.033 \\
36 & 0.410 & 0.237 & 0.186 & 0.183 & 0.070 & 0.062 & 0.119 & 0.050 \\
49 & 0.469 & 0.234 & 0.176 & 0.171 & 0.075 & 0.074 & 0.121 & 0.059 \\
64 & 0.530 & 0.237 & 0.168 & 0.157 & 0.081 & 0.081 & 0.122 & 0.067 \\
100 & 0.598 & 0.236 & 0.163 & 0.142 & 0.079 & 0.079 & 0.120 & 0.070 \\
\midrule
\multicolumn{9}{l}{\textit{Panel B: Varying $K$}}\\
2 & 0.690 & 0.283 & 0.204 & 0.198 & 0.095 & 0.095 & 0.179 & 0.087 \\
4 & 0.568 & 0.248 & 0.182 & 0.178 & 0.084 & 0.083 & 0.147 & 0.075 \\
6 & 0.507 & 0.251 & 0.187 & 0.181 & 0.078 & 0.077 & 0.135 & 0.064 \\
8 & 0.444 & 0.238 & 0.185 & 0.182 & 0.072 & 0.067 & 0.126 & 0.055 \\
10 & 0.410 & 0.237 & 0.186 & 0.183 & 0.070 & 0.062 & 0.119 & 0.050 \\
12 & 0.372 & 0.225 & 0.180 & 0.178 & 0.064 & 0.052 & 0.109 & 0.045 \\
14 & 0.350 & 0.223 & 0.181 & 0.181 & 0.059 & 0.058 & 0.104 & 0.037 \\
16 & 0.327 & 0.217 & 0.178 & 0.178 & 0.055 & 0.054 & 0.098 & 0.035 \\
18 & 0.308 & 0.205 & 0.175 & 0.175 & 0.048 & 0.047 & 0.089 & 0.029 \\
20 & 0.281 & 0.202 & 0.170 & 0.170 & 0.046 & 0.044 & 0.080 & 0.026 \\
\midrule
\multicolumn{9}{l}{\textit{Panel C: Varying $\rho$}}\\
0 & 0.299 & 0.134 & 0.092 & 0.096 & 0.035 & 0.029 & 0.047 & 0.021 \\
0.1 & 0.331 & 0.166 & 0.120 & 0.121 & 0.046 & 0.040 & 0.069 & 0.029 \\
0.2 & 0.378 & 0.210 & 0.165 & 0.164 & 0.062 & 0.053 & 0.103 & 0.043 \\
0.3 & 0.449 & 0.269 & 0.217 & 0.212 & 0.077 & 0.072 & 0.145 & 0.059 \\
0.4 & 0.508 & 0.327 & 0.273 & 0.262 & 0.100 & 0.096 & 0.196 & 0.084 \\
0.5 & 0.565 & 0.403 & 0.349 & 0.329 & 0.127 & 0.123 & 0.261 & 0.110 \\
\bottomrule
\end{tabular}
\endgroup
\begin{flushleft}\footnotesize
Notes: Rejection frequencies at the 5\% level, based on 10,000 replications. Unless varied, $n=36$, $K=10$, $\rho=0.25$, and $\omega=1$. DN denotes the node-score implementation described in Appendix~\ref{app:simulation_results}; JK-DN denotes the corrected spatial jackknife. DN$_0$ and JK$_0$ omit their respective duplication corrections. Node effects are refitted in every deletion and cube regression. \end{flushleft}
\end{table}

\begin{table}[p]
\centering
\caption{Degenerate Gaussian design with node fixed effects}
\label{tab:sim_fe_fast}
\begingroup
\footnotesize
\setlength{\tabcolsep}{4pt}
\renewcommand{\arraystretch}{1.05}
\begin{tabular}{lrrrrrrrr}
\toprule
Setting & White & TW & Dyadic & DN & JK-DN & MJS & DN$_0$ & JK$_0$ \\
\midrule
\multicolumn{9}{l}{\textit{Panel A: Varying $n$}}\\
16 & 0.089 & 0.174 & 0.229 & 0.232 & 0.005 & 0.004 & 0.032 & 0.002 \\
25 & 0.066 & 0.111 & 0.149 & 0.152 & 0.010 & 0.004 & 0.014 & 0.002 \\
36 & 0.064 & 0.095 & 0.118 & 0.125 & 0.024 & 0.005 & 0.012 & 0.004 \\
49 & 0.055 & 0.079 & 0.091 & 0.097 & 0.035 & 0.006 & 0.011 & 0.005 \\
64 & 0.055 & 0.071 & 0.081 & 0.087 & 0.043 & 0.005 & 0.010 & 0.006 \\
100 & 0.052 & 0.062 & 0.067 & 0.070 & 0.051 & 0.021 & 0.007 & 0.005 \\
\midrule
\multicolumn{9}{l}{\textit{Panel B: Varying $K$}}\\
2 & 0.057 & 0.088 & 0.108 & 0.115 & 0.028 & 0.008 & 0.012 & 0.004 \\
4 & 0.057 & 0.086 & 0.111 & 0.120 & 0.028 & 0.007 & 0.011 & 0.004 \\
6 & 0.063 & 0.094 & 0.112 & 0.120 & 0.027 & 0.007 & 0.012 & 0.004 \\
8 & 0.054 & 0.088 & 0.108 & 0.116 & 0.025 & 0.006 & 0.012 & 0.004 \\
10 & 0.064 & 0.095 & 0.118 & 0.125 & 0.024 & 0.005 & 0.012 & 0.004 \\
12 & 0.058 & 0.090 & 0.112 & 0.122 & 0.022 & 0.004 & 0.011 & 0.004 \\
14 & 0.058 & 0.091 & 0.115 & 0.125 & 0.019 & 0.012 & 0.011 & 0.003 \\
16 & 0.066 & 0.094 & 0.114 & 0.125 & 0.020 & 0.012 & 0.013 & 0.004 \\
18 & 0.061 & 0.092 & 0.117 & 0.128 & 0.018 & 0.011 & 0.015 & 0.003 \\
20 & 0.063 & 0.093 & 0.117 & 0.126 & 0.018 & 0.012 & 0.014 & 0.003 \\
\midrule
\multicolumn{9}{l}{\textit{Panel C: Varying $\rho$}}\\
0 & 0.057 & 0.092 & 0.114 & 0.124 & 0.025 & 0.004 & 0.011 & 0.004 \\
0.1 & 0.059 & 0.091 & 0.113 & 0.122 & 0.024 & 0.003 & 0.011 & 0.003 \\
0.2 & 0.055 & 0.088 & 0.112 & 0.121 & 0.022 & 0.003 & 0.010 & 0.003 \\
0.3 & 0.062 & 0.090 & 0.112 & 0.123 & 0.022 & 0.005 & 0.012 & 0.003 \\
0.4 & 0.058 & 0.093 & 0.115 & 0.121 & 0.023 & 0.003 & 0.011 & 0.003 \\
0.5 & 0.062 & 0.091 & 0.114 & 0.125 & 0.024 & 0.004 & 0.013 & 0.004 \\
\bottomrule
\end{tabular}
\endgroup
\begin{flushleft}\footnotesize
Notes: Rejection frequencies at the 5\% level, based on 10,000 replications. Unless varied, $n=36$, $K=10$, $\rho=0.25$, and $\omega=1$. DN denotes the node-score implementation described in Appendix~\ref{app:simulation_results}; JK-DN denotes the corrected spatial jackknife. DN$_0$ and JK$_0$ omit their respective duplication corrections. Node effects are refitted in every deletion and cube regression. DN is a numerical comparison here; Theorem~\ref{thm:fe_gaussian}(ii) establishes the fast fixed-effect result for JK-DN. \end{flushleft}
\end{table}

\begin{table}[p]
\centering
\caption{Additional degenerate non-Gaussian results}
\label{tab:sim_nongaussian_more}
\begingroup
\footnotesize
\setlength{\tabcolsep}{4pt}
\renewcommand{\arraystretch}{1.05}
\begin{tabular}{lrrrrrrrr}
\toprule
Setting & White & TW & Dyadic & DN & JK-DN & MJS & DN$_0$ & JK$_0$ \\
\midrule
\multicolumn{9}{l}{\textit{Panel A: Varying $K$}}\\
2 & 0.113 & 0.101 & 0.122 & 0.122 & 0.034 & 0.018 & 0.017 & 0.008 \\
4 & 0.088 & 0.098 & 0.117 & 0.121 & 0.039 & 0.019 & 0.016 & 0.008 \\
6 & 0.072 & 0.091 & 0.113 & 0.118 & 0.034 & 0.016 & 0.015 & 0.006 \\
8 & 0.067 & 0.092 & 0.114 & 0.120 & 0.031 & 0.013 & 0.015 & 0.006 \\
10 & 0.065 & 0.094 & 0.114 & 0.122 & 0.032 & 0.012 & 0.014 & 0.005 \\
12 & 0.063 & 0.095 & 0.123 & 0.129 & 0.030 & 0.011 & 0.014 & 0.006 \\
14 & 0.063 & 0.094 & 0.116 & 0.124 & 0.031 & 0.009 & 0.015 & 0.005 \\
16 & 0.060 & 0.092 & 0.116 & 0.126 & 0.027 & 0.005 & 0.012 & 0.004 \\
18 & 0.058 & 0.092 & 0.113 & 0.123 & 0.025 & 0.005 & 0.012 & 0.004 \\
20 & 0.064 & 0.095 & 0.119 & 0.128 & 0.026 & 0.018 & 0.014 & 0.004 \\
\midrule
\multicolumn{9}{l}{\textit{Panel B: Varying $\rho$}}\\
0 & 0.052 & 0.086 & 0.113 & 0.122 & 0.033 & 0.009 & 0.010 & 0.003 \\
0.1 & 0.056 & 0.087 & 0.115 & 0.124 & 0.034 & 0.010 & 0.011 & 0.005 \\
0.2 & 0.063 & 0.094 & 0.118 & 0.127 & 0.034 & 0.011 & 0.013 & 0.005 \\
0.3 & 0.070 & 0.096 & 0.114 & 0.122 & 0.035 & 0.013 & 0.016 & 0.006 \\
0.4 & 0.095 & 0.117 & 0.129 & 0.132 & 0.032 & 0.013 & 0.020 & 0.008 \\
0.5 & 0.121 & 0.131 & 0.140 & 0.143 & 0.036 & 0.019 & 0.033 & 0.013 \\
\midrule
\multicolumn{9}{l}{\textit{Panel C: Varying $\omega$}}\\
0 & 0.047 & 0.091 & 0.118 & 0.127 & 0.034 & 0.009 & 0.010 & 0.004 \\
0.2 & 0.054 & 0.089 & 0.117 & 0.128 & 0.036 & 0.010 & 0.011 & 0.004 \\
0.4 & 0.057 & 0.093 & 0.120 & 0.128 & 0.033 & 0.011 & 0.013 & 0.005 \\
0.6 & 0.061 & 0.098 & 0.121 & 0.131 & 0.034 & 0.011 & 0.014 & 0.005 \\
0.8 & 0.065 & 0.098 & 0.119 & 0.127 & 0.034 & 0.011 & 0.014 & 0.006 \\
1 & 0.065 & 0.094 & 0.114 & 0.122 & 0.032 & 0.012 & 0.014 & 0.005 \\
1.2 & 0.071 & 0.094 & 0.119 & 0.125 & 0.035 & 0.012 & 0.016 & 0.007 \\
1.4 & 0.071 & 0.101 & 0.119 & 0.126 & 0.030 & 0.011 & 0.014 & 0.005 \\
1.6 & 0.073 & 0.092 & 0.116 & 0.123 & 0.032 & 0.012 & 0.015 & 0.006 \\
1.8 & 0.073 & 0.094 & 0.115 & 0.121 & 0.031 & 0.011 & 0.014 & 0.007 \\
2 & 0.077 & 0.095 & 0.114 & 0.121 & 0.032 & 0.012 & 0.015 & 0.006 \\
\bottomrule
\end{tabular}
\endgroup
\begin{flushleft}\footnotesize
Notes: Rejection frequencies at the 5\% level, based on 10,000 replications. Unless varied, $n=36$, $K=10$, $\rho=0.25$, and $\omega=1$. DN denotes the node-score implementation described in Appendix~\ref{app:simulation_results}; JK-DN denotes the corrected spatial jackknife. DN$_0$ and JK$_0$ omit their respective duplication corrections. At $\omega=0$ the design is Gaussian; at $\omega>0$ its limit contains a non-Gaussian interaction. \end{flushleft}
\end{table}

\bibliographystyle{chicago}
\bibliography{multiway_clustering}
\end{document}